\documentclass[indent=latex,review=false,theme=default,marginsize=3cm,fontsize=12pt]{acmhomework}
\usepackage{xurl}

\newcommand{\Path}{\textsf{Path}}
\newcommand{\BookUA}{\textsf{BookUA}}
\newcommand{\PointedUA}{\textsf{PointedUA}}

\newcommand{\El}{\textsf{El}}
\newcommand{\tEl}{\widetilde{\El}}

\newcommand{\HId}{\textsf{HId}}
\newcommand{\LHInv}{\textsf{LHInv}}
\newcommand{\RHInv}{\textsf{RHInv}}

\newcommand{\CompPair}{\textsf{CompPair}}
\newcommand{\src}{\textsf{src}}
\newcommand{\dest}{\textsf{tgt}}

\newcommand{\PtdUA}{\textsf{PtdUA}}
\newcommand{\UA}{\textsf{UA}}

\renewcommand\fracsquareslash[1]{
  \mathrel{\stackunder[1pt]{\tikz{\draw (0,0) rectangle (\squareheight,\squareheight);\draw(0,0) -- (\squareheight,\squareheight)}}{$\scriptscriptstyle #1$}}}

\title{(Pointed) Univalence in Universe Category Models of Type Theory}

\author{Krzysztof Kapulkin \and Yufeng Li}

\begin{abstract}
  We provide a formulation of the univalence axiom in a universe category model of dependent type theory that is convenient to verify in homotopy-theoretic settings.
  We further develop a strengthening of the univalence axiom, called pointed univalence, that is both computationally desirable and semantically natural, and verify its closure under Artin--Wraith gluing and formation of inverse diagrams.
\end{abstract}

\begin{document}

\maketitle

\section*{Introduction}

Voevodsky's Univalence Axiom is perhaps the most fundamental logical principle introduced in the 21st century.
A universe of (small) types satisfies the univalence axiom if, for any two types in that universe, their type of equalities (or identifications) is equivalent to the type of equivalences between them.
Thus, for example, two logically equivalent propositions or bijective sets can be treated as equal, and hence substituted for each other in all contexts.
While convenient and often assumed in informal mathematical practice, the univalence axiom requires formal justification.
The first such justification comes from Voevodsky's celebrated simplicial model \cite{kl21}; the model that not only satisfies, but in fact inspired, the univalence axiom.
With types interpreted as spaces (or Kan complexes to be precise), it is not surprising that paths in the universe of small types would be given by homotopy equivalences thereof.
Since then, a plethora of other models of univalence were introduced, including cubical models
\cite{cchm15, op16, ang+21, ahh18, lops18} and very general higher-categorical models in Grothendieck $\infty$-toposes \cite{shu-inv-el-ua,shu-elegant-ua,shu-strict-ua}.

The aims of this paper are twofold.
First, to more systematically study the univalence axiom in the language of Voevodsky's universe categories \cite{voe15} and/or Uemura's categories with representable maps \cite{uem23}.
Second, to discuss a slight strengthening of the univalence axiom, termed here \emph{pointed univalence}, that holds in most models and allows for pattern matching on equivalences, perhaps in the spirit of Voevodsky's original work.
We discuss each of these goals in turn.

Since their introduction \cite{voe15,voe17,kl21}, universe categories have been recognized as a convenient language to speak of semantics of type theory \cite{uem23}, including recasting other notions of models as universe categories. Perhaps the most prominent example is the rephrasing of categories with families \cite{dyb05} as natural models \cite{awo18}, which are universe category structures on presheaf categories.
In essence, a universe category consists of a category $\bC$ with a map $\tMcU \to \McU$ and a choice of pullbacks of $\tMcU \to \McU$.
The map $\tMcU \to \McU$ is thought of as an external universe and its pullbacks serve to interpret type dependency.
Some work is required to explain when the resulting model supports different type constructors and the details of this were worked out in \cite{kl21}.

When using universe categories to model homotopy type theory \cite{hottbook}, attention needs to be paid to how the univalence axiom is treated.
In \cite{kl21}, the goal is to build a single model of homotopy type theory (in simplicial sets), and when it comes to verifying univalence, the authors do just that: they translate the statement into a statement about simplicial sets and check it directly.
In contrast, the purpose of the present paper is to provide a more general semantic treatment of univalence at the level of universe categories by rephrasing the condition in a way that would make it easy to check in practice across various models.
We also aim to avoid using the internal language formulations which, while elegant, might occasionally conceal important details, possibly leading to incomplete arguments.
Altogether, the upshot of our formulation is that, when verifying univalence in a given universe category model, one does not need to work with syntax at all.
We contrast it with the closely related work of \citeauthor{uem23} \cite{uem23}, who provides the treatment of univalence in categories with representable maps, a close cousin of universe categories, taking a much more syntactic approach.

Interestingly, to arrive at a formulation of the univalence axiom that can be easily verified in a universe category, we actually give a slightly stronger statement than the one commonly used and given in the HoTT Book \cite{hottbook}, which brings us to the second contribution of this paper.
As mentioned before, the univalence axiom is most generally formulated by saying that a certain map from the identity type between two types in a universe to the type of equivalences between them is itself an equivalence.
Our strengthening requires that the homotopy inverse of this map send the identity equivalence to the reflexivity term.
To differentiate the two, we call the version found in \cite{hottbook} \emph{book univalence} and our new version \emph{pointed univalence}.
Pointed univalence is a very natural strengthening motivated by semantics coming from Quillen model categories \cite{qui67,hov99}, provided that the maps $A, B \colon 1 \to \McU$ picking the two types are cofibrations.
This is true, for instance, if sections of fibrations are cofibrations, which is the case in all Cisinski model structures on presheaf categories.
In contrast, this is not the case in the classifying category of cubical type theory \cite{cchm15}, since terms are not, in general, cofibrations.
On the other hand, since maps from the left class to fibrations axiomatize pattern matching, our pointed univalence is also computationally desirable, as it justifies performing pattern matching on equivalences.
As a result, we believe that this strengthening constitutes a new notion of independent interest with strong semantic justification.

At its core, our rephrasing is based on a lifting condition (using the notion of a dual of a deformation retraction \cite{cis19}) that involves only the object of h-isomorphisms; in particular, we do not need to know that the external universe admits an $\Id$-type structure.
While book univalence asks for a certain commutative square to admit a diagonal filler making only the lower triangle commute, pointed univalence requires that both resulting triangles commute.
As such, it is very natural to verify in models coming from Quillen model categories, where fillers make both triangles commute.

Our semantic study of pointed univalence includes studying closure properties of models under two fundamental constructions: Artin--Wraith gluing and inverse diagrams.
For book univalence, these were verified in the seminal work of Shulman \cite{shu15}, and we provide the ``pointed counterpart'' of his results.
These results require \emph{pointed function extensionality}, a similar strengthening of function extensionality.
Recalling that function extensionality asserts that two functions that are pointwise equal must be equal (as functions), pointed function extensionality asserts that functions pointwise equal via reflexivity terms must be equal by the reflexivity term.
Once again, this strengthening is given by another natural lifting condition.

As indicated above, the paper is naturally split into two parts, the first concerning (book and pointed) univalence in the framework of universe categories (\cref{sec:background,sec:gen-hiso,sec:univalence}), and the second `zooming in' on closure properties of pointed univalence (\cref{sec:funext,sec:glue-universe,sec:inverse-universe}).
In a bit more detail: in \cref{sec:background}, we review the basic theory of universe categories and, in particular, when a universe category supports different type formers needed to formulate the univalence axiom.
In \cref{sec:gen-hiso}, we construct the object of h-isomorphisms, which is required to state the univalence axiom.
In \cref{sec:univalence}, we describe three formulations of univalence: \cref{def:axm-univalence} contains the statement that is the easiest to make in a universe category; \cref{thm:univalence-sdf} rephrases it using a relatively easy to state condition; and \cref{thm:univalence-tt} compares it to a type-theoretic formulation.
In \cref{sec:funext}, we give a formulation of a variant of functional extensionality, again in terms of a lifting condition, called \emph{pointed functional extensionality}.
In \cref{sec:glue-universe}, we show that type-theoretic constructions, including pointed functional extensionality and univalence, are stable under formation of Artin--Wraith gluing categories.
Finally, in \cref{sec:inverse-universe}, we use iterated Artin--Wraith gluing to show that, under pointed functional extensionality, pointed univalence is stable under formation of inverse diagrams.


\tableofcontents
\section{Background}\label{sec:background}
In this paper we work in the framework of universe category models of type theory.
We start by reviewing some necessary background about pushforwards, polynomial
functors, universe category models and how one axiomatises various logical
constructions in this framework.

\subsection{Polynomial Functors}
We will work extensively with exponentiable maps as well as sometimes with the
polynomial functor and generic post-composite associated with an exponentiable
map.
\begin{definition}\label{def:pushforward}
  A map $p \colon E \to B$ in a category $\bC$ is \emph{exponentiable} if
  pullbacks along $p$ exist and the pullback functor
  $p^* \colon \sfrac{\bC}{B} \to \sfrac{\bC}{E}$ admits a right adjoint
  $p_* \colon \sfrac{\bC}{E} \to \sfrac{\bC}{B}$ which we call to be the
  \emph{pushforwards}.
\end{definition}

Using pushforwards, one can define the polynomial functor associated with an
exponentiable map.
\begin{definition}\label{def:expn-polynomial}
  Let $\bC$ be a finitely complete category and fix an exponentiable map
  $p \colon E \to B \in \bC$ along with a specific choice of pushforward $p_*$.
  Then, the \emph{polynomial functor} associated with $p$ is the composite functor
  \begin{equation*}
    \bP_p \coloneqq \left( \bC \xrightarrow{E \times -} \sfrac{\bC}{E} \xrightarrow{p_*} \sfrac{\bC}{B} \xrightarrow{B_!} \bC \right)
  \end{equation*}
  where the last map $\sfrac{\bC}{B} \xrightarrow{B_!} \bC$ just forgets away
  the base.
  Explicitly, given $X \in \bC$, it returns the domain object of the map
  $p_!(E \times X) \xrightarrow{p_!(\proj_1)} B$.
\end{definition}

Closely related to the polynomial functor associated with an exponential map is
the generic post-composite with an exponentiable map, constructed as follows.
\begin{construction}\label{constr:gen-comp}
  Let $\bC$ be a finitely complete category and fix an exponentiable map
  $p \colon E \to B \in \bC$ with pushforward $p_*$.
  Let $p' \colon E' \to B' \in \bC$ be a map.
  Then, the map
  \begin{equation*}
   \GenComp(p',p) \colon \ev^*(E \times E') \to p^*p_*(E \times B') \to p_*(E \times B')
  \end{equation*}
  is constructed as follows by first taking the pushforward of
  $E \times B' \to E$ along $p$ to obtain $p_*(E \times B')$, then taking the
  pullback along $p$ to obtain the counit $p^*p_*(E \times B') \to E \times B'$
  and finally pulling back $E' \to B'$ along the composite
  $p^*p_*(E \times B') \to E \times B' \to B'$.
  \begin{equation*}
    \begin{tikzcd}[cramped]
      {\ev^*(E \times E')} \\
      {p^*p_*(E \times B')} & {E \times E'} & {E'} \\
      & {E \times B'} & {B'} \\
      {p_*(E \times B')} && E \\
      && B
      \arrow[from=1-1, to=2-1]
      \arrow[from=1-1, to=2-2]
      \arrow["\ev"{description}, from=2-1, to=3-2]
      \arrow[from=2-1, to=4-1]
      \arrow[curve={height=18pt}, from=2-1, to=4-3]
      \arrow[from=2-2, to=2-3]
      \arrow[from=2-2, to=3-2]
      \arrow[from=2-3, to=3-3]
      \arrow["{p'}", from=3-2, to=3-3]
      \arrow["{\proj_1}"{description, pos=0.3}, from=3-2, to=4-3]
      \arrow[from=4-1, to=5-3]
      \arrow["p", from=4-3, to=5-3]
      \arrow["\lrcorner"{anchor=center, pos=0.15, scale=1.5}, draw=none, from=1-1, to=3-2]
      \arrow["\lrcorner"{anchor=center, pos=0.15, scale=1.5}, draw=none, from=2-1, to=5-3]
      \arrow["\lrcorner"{anchor=center, pos=0.15, scale=1.5}, draw=none, from=2-2, to=3-3]
    \end{tikzcd}
  \end{equation*}
\end{construction}

The generic composite constructed in the above \Cref{constr:gen-comp} lives up
to its name.
\begin{lemma}\label{lem:gen-comp}
  Let $p \colon E \to B \in \bC$ be an exponentiable map in finitely complete
  category $\bC$.
  Take $p' \colon E' \to B' \in \bC$ to be just any map.

  Suppose that one has pullbacks as on the left.
  Then the composite $X_2 \to X_1 \to X_0$ arises as a pullback of
  $\GenComp(p,p') \colon \ev^*(E \times E') \to p^*p_*(E \times B') \to p_*(E
  \times B')$ along a map $\ceil{X_1}.\ceil{X_2}$ as on the right
  \begin{center}
    \begin{minipage}{0.45\linewidth}
      \begin{equation*}
        \begin{tikzcd}[cramped]
          {X_2} & {E'} \\
          {X_1} & B' \\
          {X_0} & {E} \\
          & B
          \arrow["{\widetilde{\ceil{X_2}}}", from=1-1, to=1-2]
          \arrow["{q_2}"', from=1-1, to=2-1]
          \arrow[from=1-2, to=2-2]
          \arrow["{\ceil{X_2}}"{description}, from=2-1, to=2-2]
          \arrow["{q_1}"', from=2-1, to=3-1]
          \arrow["{\widetilde{\ceil{X_1}}}"{description}, from=2-1, to=3-2]
          \arrow["{\ceil{X_1}}"', from=3-1, to=4-2]
          \arrow[from=3-2, to=4-2]
          \arrow["\lrcorner"{anchor=center, pos=0.15, scale=1.5}, draw=none, from=1-1, to=2-2]
          \arrow["\lrcorner"{anchor=center, pos=0.15, scale=1.5}, draw=none, from=2-1, to=4-2]
        \end{tikzcd}
      \end{equation*}
    \end{minipage}
    \begin{minipage}{0.45\linewidth}
      \begin{equation*}
        \begin{tikzcd}[cramped]
          {X_2} & {\ev^*(E' \times E')} \\
          {X_1} & {p^*p_*(E \times B')} \\
          {X_0} & {p_*(E \times B')}
          \arrow[dashed, from=1-1, to=1-2]
          \arrow["{q_2}"', from=1-1, to=2-1]
          \arrow[from=1-2, to=2-2]
          \arrow[dashed, from=2-1, to=2-2]
          \arrow["{q_1}"', from=2-1, to=3-1]
          \arrow[from=2-2, to=3-2]
          \arrow["{\ceil{X_1}.\ceil{X_2}}"', dashed, from=3-1, to=3-2]
          \arrow["\lrcorner"{anchor=center, pos=0.05, scale=1.5}, draw=none, from=1-1, to=3-2]
          \arrow["\lrcorner"{anchor=center, pos=0.05, scale=1.5}, draw=none, from=2-1, to=3-2]
        \end{tikzcd}
      \end{equation*}
    \end{minipage}
  \end{center}
  where $X_0 \xrightarrow{\ceil{X_1}.\ceil{X_2}} p_*(E \times B')$ is the unique
  map such that
  $(X_0 \xrightarrow{\ceil{X_1}.\ceil{X_2}} p_*(E \times B) \to B) = (X_0
  \xrightarrow{\ceil{X_1}} B)$ and
  $(X_0 \xrightarrow{\ceil{X_1}.\ceil{X_2}} p_*(E \times B'))^\dagger = (X_1
  \xrightarrow{(\widetilde{\ceil{X_1}}, \ceil{X_2})} E \times B')$.
\end{lemma}
\begin{proof}
  One explicitly constructs the required map $\ceil{X_1}.\ceil{X_2}$ in the following steps.
  \begin{enumerate}
    \item One first constructs the map
    $X_0 \textcolor{magenta0}{{}\xrightarrow{\ceil{X_1}.\ceil{X_2}}{}}
    p_*(E \times B')$
    as a map
    $X_0 \to p_*(E \times B')$ over $B$ into the pushforward.
    By the universal property of the pushforward, we choose the transpose of
    this map to be the map
    $p^*X_0 = X_1 \textcolor{magenta0}{{}\xrightarrow{(\widetilde{\ceil{X_1}},
        \ceil{X_2})}{}} E \times B'$ over $E$.
    \item Next one computes the pullback of
    $p^*p_*(E \times B') \to p_*(E \times B')$ along
    $X_0 \textcolor{magenta0}{{}\xrightarrow{\ceil{X_1}.\ceil{X_2}}{}}
    p_*(E \times B')$ as $X_1 \to X_0$ because of the pullback lemma.
    \begin{equation*}
      \begin{tikzcd}[cramped,row sep=small]
        {X_1} && {p^*p_*(E \times B')} & E \\
        {X_0} && {p_*(E \times B')} & B
        \arrow["{p^*(\ceil{X_1}.\ceil{X_2})}"'{}, color=magenta0, from=1-1, to=1-3]
        \arrow["{\widetilde{\ceil{X_1}}}", color=magenta0, curve={height=-18pt}, from=1-1, to=1-4]
        \arrow[from=1-1, to=2-1]
        \arrow[from=1-3, to=1-4]
        \arrow[from=1-3, to=2-3]
        \arrow["p", from=1-4, to=2-4]
        \arrow["{\ceil{X_1}.\ceil{X_2}}"{description}, color=magenta0, from=2-1, to=2-3]
        \arrow["{\ceil{X_1}}"', color=magenta0, curve={height=18pt}, from=2-1, to=2-4]
        \arrow[from=2-3, to=2-4]
        \arrow["\lrcorner"{anchor=center, pos=0.15, scale=1.5}, draw=none, from=1-1, to=2-3]
        \arrow["\lrcorner"{anchor=center, pos=0.15, scale=1.5}, draw=none, from=1-3, to=2-4]
      \end{tikzcd}
    \end{equation*}
    \item Finally, one notes that by definition the transpose is
    $\ev \cdot \textcolor{magenta0}{p^*(\ceil{X_1}.\ceil{X_2})} =
    (\ceil{X_1}.\ceil{X_2})^\dagger =
    \textcolor{magenta0}{(\widetilde{\ceil{X_1}}, \ceil{X_2})}$.
    Therefore, once again by the pullback lemma, pulling back
    $\ev^*(E \times B') \to p^*p_*(E \times B')$ along
    $E_1
    \textcolor{magenta0}{{}\xrightarrow{p^*(\ceil{X_1}.\ceil{X_2})}{}}
    \ev^*(E \times B')$ is exactly $X_2 \to X_1$.
    \begin{equation*}
      \begin{tikzcd}[row sep=small]
        {X_2} && {\ev^*(E \times B')} & {E \times E'} & E' \\
        {X_1} && {p^*p_*(E \times B')} & {E \times B'} & B'
        \arrow[draw={magenta0}, from=1-1, to=1-3]
        \arrow[from=1-1, to=2-1]
        \arrow[from=1-3, to=1-4]
        \arrow[from=1-3, to=2-3]
        \arrow[from=1-4, to=1-5]
        \arrow[from=1-4, to=2-4]
        \arrow[from=1-5, to=2-5]
        \arrow["{p^*(\ceil{X_1}.\ceil{X_2})}"{}, color={magenta0}, from=2-1, to=2-3]
        \arrow["{(\widetilde{\ceil{X_1}}, \ceil{X_2})}"{description}, color={magenta0}, curve={height=18pt}, from=2-1, to=2-4]
        \arrow["{\ceil{X_2}}"{description}, color={magenta0}, curve={height=42pt}, from=2-1, to=2-5]
        \arrow["\ev"{description}, from=2-3, to=2-4]
        \arrow[from=2-4, to=2-5]
        \arrow["\lrcorner"{scale=1.5, anchor=center, pos=0.15, rotate=0}, draw=none, from=1-1, to=2-3]
        \arrow["\lrcorner"{scale=1.5, anchor=center, pos=0.15, rotate=0}, draw=none, from=1-3, to=2-4]
        \arrow["\lrcorner"{scale=1.5, anchor=center, pos=0.15}, draw=none, from=1-4, to=2-5]
      \end{tikzcd}
    \end{equation*}
  \end{enumerate}
  Putting everything together, we obtain the following diagram.
  \begin{equation*}
    \begin{tikzcd}[cramped, row sep=small]
      {X_2} && {\ev^*(E \times E')} \\
      {X_1} && {p^*p_*(E \times B')} & {E \times E'} & {E'} \\
      &&& {E \times B'} & {B'} \\
      {X_0} && {p_*(E \times B')} \\
      &&&& E \\
      &&&& B
      \arrow[draw=magenta0, from=1-1, to=1-3]
      \arrow[from=1-1, to=2-1]
      \arrow[from=1-3, to=2-3]
      \arrow[from=1-3, to=2-4]
      \arrow["{p^*(\ceil{X_1}.\ceil{X_2})}"{}, color=magenta0, from=2-1, to=2-3]
      \arrow["\ev"{description}, from=2-3, to=3-4]
      \arrow[from=2-3, to=4-3]
      \arrow[curve={height=18pt}, from=2-3, to=5-5]
      \arrow[from=2-4, to=2-5]
      \arrow[from=2-4, to=3-4]
      \arrow[from=2-5, to=3-5]
      \arrow[from=3-4, to=3-5]
      \arrow["{\proj_1}"{description, pos=0.3}, from=3-4, to=5-5]
      \arrow["{\ceil{X_1}.\ceil{X_2}}"{description}, color=magenta0, from=4-1, to=4-3]
      \arrow["{\ceil{X_1}}"', color=magenta0, curve={height=12pt}, from=4-1, to=6-5]
      \arrow[from=4-3, to=6-5]
      \arrow["p", from=5-5, to=6-5]
      \arrow[from=2-1, to=4-1]
      \arrow[crossing over, "{(\widetilde{\ceil{X_1}}, \ceil{X_2})}"{description}, color=magenta0, curve={height=12pt}, from=2-1, to=3-4]
      \arrow[crossing over, "{\ceil{X_2}}"{description, pos=0.3}, color=magenta0, from=2-1, to=3-5]
      \arrow[crossing over, "{\widetilde{\ceil{X_1}}}"{description, pos=0.3}, color=magenta0, curve={height=18pt}, from=2-1, to=5-5]
    \end{tikzcd}
  \end{equation*}

  Uniqueness of the map $\xrightarrow{\ceil{X_1}.\ceil{X_2}}$ then follows by
  uniqueness of the adjoint transpose.
\end{proof}

The generic composite is related to the associated polynomial functor in the
following sense.
\begin{lemma}\label{lem:gen-comp-psfw}
  Pushing forwards the first map $\ev^*(E \times E') \to p^*p_*(E \times B')$,
  in the context of the generic composite from \Cref{constr:gen-comp}, viewed as
  an object of the slice category over $p^*p_*(E \times B')$, along the second
  map $p^*p_*(E \times B') \to p_*(E \times B')$ in $\GenComp(p',p)$ obtains the
  same map $p_*(E \times E')\to p_*(E \times B')$ as applying the polynomial
  functor from \Cref{def:expn-polynomial} associated with $p$ to $E' \to B'$.
\end{lemma}
\begin{proof}
  This is immediate by the distributivity law of pushforwards as in \cite[Lemma
  2.3]{inv-psfw}.
\end{proof}

\subsection{Universe Category Models of Type Theory}
Universe category models \cite{voe15} of type theory and formulations of type
theories as categories with representable maps (CwRs) \cite{uem23} are general
high power-to-weight ratio frameworks that can make models of dependent type
theory where context extension is modelled as certain pullback-stable projection
maps arise as special cases of them.
\begin{definition}[{\cite[Definition 2.1]{voe15}}]\label{def:univ-struct}
  Let $\bC$ be a finitely complete category.
  A \emph{universe structure} on an exponentiable map
  $\pi \colon \tMcU \to \McU$ is a specific choice of a right adjoint
  $\var \colon \sfrac{\bC}{\McU} \to \sfrac{\bC}{\tMcU}$
  \begin{equation*}
    \begin{tikzcd}[cramped]
      {\sfrac{\bC}{\McU}} && {\sfrac{\bC}{\tMcU}}
      \arrow[""{name=0, anchor=center, inner sep=0}, "\var"{}, shift left=2, dashed, from=1-1, to=1-3]
      \arrow[""{name=1, anchor=center, inner sep=0}, "{\pi_!}", shift left=2, from=1-3, to=1-1]
      \arrow["\dashv"{anchor=center, rotate=90}, draw=none, from=1, to=0]
    \end{tikzcd}
  \end{equation*}
  of $\pi_! \colon \sfrac{\bC}{\tMcU} \to \sfrac{\bC}{\McU}$ the
  post-composition functor (i.e. just a choice of pullbacks along $\pi$).

  For each $A \colon \Gamma \to \McU \in \sfrac{\bC}{\McU}$, we denote by
  $\pi_A$ the counit at $A$ and $\Gamma.A$ the domain of $\pi_A$
  so that one has a pullback square
  \begin{equation*}
    \begin{tikzcd}[cramped]
      {\Gamma.A} & \tMcU \\
      \Gamma & \McU
      \arrow["{\var_A}", from=1-1, to=1-2]
      \arrow["{\pi_A}"', from=1-1, to=2-1]
      \arrow["\lrcorner"{anchor=center, pos=0.15, scale=1.5}, draw=none, from=1-1, to=2-2]
      \arrow[from=1-2, to=2-2]
      \arrow["A"', from=2-1, to=2-2]
    \end{tikzcd}
  \end{equation*}
  The map $\pi_A$ is also the \emph{selected pullback} by the universe
  structure.

  The map $\pi \colon \tMcU \to \McU$ equipped with such a universe structure is
  called a \emph{universal map}.
  A \emph{universe structure} on a category $\bC$ consists of a choice of a
  distinguished universal map $\tMcU \to \McU$ and a universe structure on it.
\end{definition}

\begin{example}
  By taking $\bC$ to be a presheaf category, we recover the concept of a
  category with families (cwf) \cite{dyb05}.
\end{example}

Sometimes, given a universal map $\tMcU \to \McU$, one wishes to talk about
arbitrary pullbacks of it, as opposed to the specific pullback maps selected by
its universe structure.
We fix some vocabulary for this purpose.
\begin{definition}\label{def:univ-fibration}
  Let $\pi \colon \tMcU \to \McU$ be a universal map in a finitely complete
  category $\bC$.

  The \emph{$\pi$-name} and \emph{$\pi$-point} of a map $E \to B$ are maps
  $\ceil{E}$ and $\widetilde{\ceil{E}}$ as below such that one has a pullback
  square as follows.
  \begin{equation*}
    \begin{tikzcd}[cramped]
      E & \tMcU \\
      B & \McU
      \arrow["{\widetilde{\ceil{E}}}", from=1-1, to=1-2]
      \arrow[from=1-1, to=2-1]
      \arrow["\lrcorner"{anchor=center, pos=0.15, scale=1.5}, draw=none, from=1-1, to=2-2]
      \arrow[from=1-2, to=2-2]
      \arrow["{\ceil{E}}"', from=2-1, to=2-2]
    \end{tikzcd}
  \end{equation*}
  Together, the $\pi$-(name,point) pair $(\ceil{E}, \widetilde{\ceil{E}})$ forms
  a \emph{$\pi$-fibrancy structure} for $E \twoheadrightarrow B$.
  Maps that can be equipped with a $\pi$-fibrancy structure are called
  \emph{$\pi$-fibrations} and are denoted $E \twoheadrightarrow B$.
  As an object in the slice over $B$, we say
  $E \twoheadrightarrow B \in \sfrac{\bC}{B}$ is a $\pi$-fibrant object.

  Although $\pi$-fibrancy structures are not necessarily unique, given
  $A \colon \Gamma \to \McU$, the selected pullback $\Gamma.A \to \Gamma$ admits
  a \emph{canonical} $\pi$-fibrancy structure $(A,\var_A)$.
  Furthermore, given a $\pi$-fibration $E \twoheadrightarrow B$ with fibrancy
  structure $(\ceil{E},\widetilde{\ceil{E}})$, its \emph{canonical $\pi$-fibrant
    replacement} is the isomorphic map
  $\pi_{\ceil{E}} \colon B.\ceil{E} \cong E \to B$.
\end{definition}

Quite sometimes, such as in the context of talking about univalence, we need to
work in categories equipped with multiple universal maps.
A common occurrence of a category equipped with two universal maps is the case
of internal universes.
\begin{definition}\label{def:int-univ}
  Let $\pi \colon \tMcU \to \McU$ be a universal map in a finitely complete
  category $\bC$.
  An \emph{internal fibrant universe} structure on $\tMcU \to \McU$ is a
  universal map $\pi_0 \colon \tMcU_0 \to \McU_0$ and a terminal object $1$
  such that $\tMcU_0 \twoheadrightarrow \McU_0$ and
  $\McU_0 \twoheadrightarrow 1$ can be equipped with \emph{canonical}
  $\pi$-fibrancy structures $(\El,\tEl)$ and
  $(\ceil{\McU_0}, \widetilde{\ceil{\McU_0}})$.
  \begin{equation*}
    \begin{tikzcd}[cramped]
      & \tMcU \\
      {\tMcU_0} & \McU \\
      {\McU_0} & {\tMcU} \\
      1 & {\McU}
      \arrow[from=1-2, to=2-2, two heads]
      \arrow["\tEl"{description}, from=2-1, to=1-2]
      \arrow[from=2-1, to=3-1, two heads]
      \arrow["\El"{description}, from=3-1, to=2-2]
      \arrow["{\widetilde{\ceil{\McU_0}}}"{description}, from=3-1, to=3-2]
      \arrow[from=3-1, to=4-1, two heads]
      \arrow[from=3-2, to=4-2, two heads]
      \arrow["{\ceil{\McU_0}}"', from=4-1, to=4-2]
      \arrow["\lrcorner"{anchor=center, pos=0.15, scale=1.5, rotate=45}, draw=none, from=2-1, to=2-2]
      \arrow["\lrcorner"{anchor=center, pos=0.15, scale=1.5}, draw=none, from=3-1, to=4-2]
    \end{tikzcd}
  \end{equation*}
\end{definition}

For now, we fix a particular universal map (be it internal or otherwise) and
equip it with various logical constructions.
We start with the $\Id$-type structure, which is the syntactical version of a
generic fibred very good path object for fibrations and whose axiomatisation is
described in \cite{voe15a}, which we now recall.
\begin{definition}[{\cite[Definition 2.7]{voe15a}}]\label{def:pre-id-type}
  Let $\pi \colon \tMcU \to \McU$ be a universal map in a finitely complete
  category $\bC$.

  A \emph{pre-$\Id$-type structure} (also known as a
  \emph{$\textsf{J1}$-structure} in \cite[Definition 2.7]{voe15a}) on
  $\tMcU \to \McU$ consists of a pair of dashed maps $(\Id,\ceil{\refl})$ as
  below from the diagonal $\tMcU \to \tMcU \times_\McU \tMcU$ to
  $\tMcU \to \McU$.
  \begin{equation*}\label{eqn:pre-Id-def}\tag{\textsc{pre-$\Id$-def}}
    \begin{tikzcd}[cramped]
      \tMcU & \tMcU \\
      {\tMcU \times_\McU \tMcU} & \McU
      \arrow["{\ceil{\refl}}", dashed, from=1-1, to=1-2]
      \arrow[from=1-1, to=2-1, "{\Delta}"']
      \arrow[from=1-2, to=2-2]
      \arrow["\Id"', dashed, from=2-1, to=2-2]
    \end{tikzcd}
  \end{equation*}
\end{definition}

\begin{construction}[{\cite[Equation (7)]{voe15a}}]\label{constr:J-prob}
  Let $\pi \colon \tMcU \to \McU$ be a universal map in a finitely complete
  category $\bC$ equipped with a pre-$\Id$-type structure $(\Id,\ceil{\refl})$
  as in \Cref{eqn:pre-Id-def}.

  Denote by
  $\ev_\partial \colon \Id_\McU(\tMcU) \twoheadrightarrow \tMcU \times_\McU
  \tMcU$ the selected pullback of $\pi$ along
  $\Id \colon \tMcU \times_\McU \tMcU \to \McU$ by the universe structure so
  that there is a map $\refl \colon \tMcU \hookrightarrow \Id_\McU(\tMcU)$ into
  the pullback, as below.
  \begin{equation*}\label{eqn:Id-def}\tag{\textsc{$\Id$-def}}
    \begin{tikzcd}[cramped]
      \tMcU \\
      & {\Id_\McU(\tMcU)} & \tMcU \\
      & {\tMcU \times_\McU \tMcU} & \McU
      \arrow["\refl"{description}, hook, dashed, from=1-1, to=2-2]
      \arrow["{\ceil{\refl}}", curve={height=-12pt}, from=1-1, to=2-3,dashed]
      \arrow["\Delta"', curve={height=12pt}, from=1-1, to=3-2]
      \arrow["{\widetilde{\Id}}"{description}, from=2-2, to=2-3,dashed]
      \arrow["{\ev_\partial}"{description}, from=2-2, to=3-2, two heads]
      \arrow["\lrcorner"{anchor=center, pos=0.15, scale=1.5}, draw=none, from=2-2, to=3-3]
      \arrow[from=2-3, to=3-3, two heads]
      \arrow["\Id"', from=3-2, to=3-3,dashed]
    \end{tikzcd}
  \end{equation*}
  Then, over $\McU$, we have the pair of maps
  $\refl \colon \tMcU \hookrightarrow \Id_\McU(\tMcU)$ (in which
  $\Id_\McU(\tMcU)$ is viewed as an object over $\McU$ via the map
  $\Id_\McU(\tMcU) \to \tMcU \times_\McU \tMcU \rightrightarrows \tMcU \to
  \McU$) and the rebased map $\tMcU \times \McU \to \McU \times \McU$, as below.
  \begin{equation*}
    \begin{tikzcd}[cramped, column sep=small]
      \tMcU && {\McU \times \tMcU} \\
      {\Id_\McU(\tMcU)} && {\McU \times \tMcU} \\
      & \McU
      \arrow["\refl"', hook, from=1-1, to=2-1]
      \arrow["{\McU \times \pi}", from=1-3, to=2-3]
      \arrow[from=2-1, to=3-2]
      \arrow["{\proj_1}", from=2-3, to=3-2]
    \end{tikzcd}
  \end{equation*}

  We construct the pullback-Hom map of $\refl$ with $\tMcU \times \pi$ in the
  slice over $\McU$ as follows
  \begin{equation*}\label{eqn:J-prob}\tag{\textsc{$\textsf{J}$-prob}}\small
    \left(\begin{tikzcd}[cramped]
        {\textsf{Diag}_\McU(\refl, \tMcU \times \pi)} \\
        {\textsf{Sq}_\McU(\refl, \tMcU \times \pi)}
        \arrow[from=1-1, to=2-1]
      \end{tikzcd}\right)
    \coloneqq
    \left(
      \begin{tikzcd}[cramped]
        {[\Id_\McU(\tMcU), \McU \times \tMcU]_\McU} \\
        {[\tMcU, \McU \times \tMcU]_\McU \times_{[\tMcU, \McU \times \tMcU]_\McU} [\Id_\McU(\tMcU), \McU \times \McU]_\McU}
        \arrow[from=1-1, to=2-1]
      \end{tikzcd}
    \right)
  \end{equation*}
  It encodes the generic $\MsJ$-elimination lifting problem.
\end{construction}

\begin{definition}[{\cite[Definition 2.8]{voe15a}}]\label{def:J-elim}
  Let $\pi \colon \tMcU \to \McU$ be a universal map in a finitely complete
  category $\bC$ equipped with a pre-$\Id$-type structure $(\Id,\ceil{\refl})$
  as in \Cref{eqn:pre-Id-def}.

  A \emph{J-elimination structure} (also known as a \emph{$\textsf{J2}$-structure} in
  \cite[Definition 2.8]{voe15a}) on the pre-$\Id$-type structure
  ($\textsf{J1}$-structure) is a section of the pullback-Hom map of
  \Cref{eqn:J-prob}.
  \begin{equation*}\label{eqn:J-def}\tag{\textsc{$\textsf{J}$-def}}
    \begin{tikzcd}[cramped]
      {\textsf{Diag}_\McU(\refl, \tMcU \times \pi)} \\
      {\textsf{Sq}_\McU(\refl, \tMcU \times \pi)}
      \arrow[shift left, from=1-1, to=2-1]
      \arrow["\MsJ", shift left, hook, from=2-1, to=1-1]
    \end{tikzcd}
  \end{equation*}
\end{definition}

\begin{remark}
  The above definition of the $\MsJ$-elimination structure is equivalently
  expressed by saying that $\MsJ$ is a stable lifting structure, in the sense of
  \cite[Definitions 1.4 and 3.1]{struct-lift}, of
  $\refl \colon \tMcU \hookrightarrow \Id_\McU(\tMcU)$ against
  $\McU \times \tMcU \to \McU \times \McU$, in the slice $\sfrac{\bC}{\McU}$.
  \begin{equation*}
    \MsJ \in \left(
      \begin{tikzcd}[cramped]
        \tMcU \ar[d,hook,"{\refl}"'] \\ \Id_\McU(\tMcU)
      \end{tikzcd}
      \fracsquareslash{\McU}
      \begin{tikzcd}[cramped]
        \McU\times\tMcU \ar[d, two heads] \\ \McU\times\McU
      \end{tikzcd}
    \right)
  \end{equation*}
\end{remark}

\begin{definition}\label{def:Id-type}
  Let $\pi \colon \tMcU \to \McU$ be a universal map in a finitely complete
  category $\bC$.

  A (full) \emph{$\Id$-type} structure (also known as a \emph{full $\MsJ$-structure}
  in \cite{voe15a}) on $\tMcU \to \McU$ consists of a pre-$\MsJ$-structure
  $(\Id,\ceil{\refl})$ as in \Cref{eqn:Id-def} along with a $\MsJ$-elimination
  structure as in \Cref{eqn:J-def}.
  \begin{center}
    \begin{minipage}{0.45\linewidth}
      \begin{equation*}\tag{\ref{eqn:Id-def}}
        \begin{tikzcd}[cramped]
          \tMcU \\
          & {\Id_\McU(\tMcU)} & \tMcU \\
          & {\tMcU \times_\McU \tMcU} & \McU
          \arrow["\refl"{description}, hook, dashed, from=1-1, to=2-2]
          \arrow["{\ceil{\refl}}", curve={height=-12pt}, from=1-1, to=2-3,dashed]
          \arrow["\Delta"', curve={height=12pt}, from=1-1, to=3-2]
          \arrow["{\widetilde{\Id}}"{description}, from=2-2, to=2-3,dashed]
          \arrow["{\ev_\partial}"{description}, from=2-2, to=3-2, two heads]
          \arrow["\lrcorner"{anchor=center, pos=0.15, scale=1.5}, draw=none, from=2-2, to=3-3]
          \arrow[from=2-3, to=3-3, two heads]
          \arrow["\Id"', from=3-2, to=3-3,dashed]
        \end{tikzcd}
      \end{equation*}
    \end{minipage}
    \begin{minipage}{0.45\linewidth}
      \begin{equation*}\tag{\text{\ref{eqn:J-def}}}
        \begin{tikzcd}[cramped]
          {\textsf{Diag}_\McU(\refl, \tMcU \times \pi)} \\
          {\textsf{Sq}_\McU(\refl, \tMcU \times \pi)}
          \arrow[shift left, from=1-1, to=2-1]
          \arrow["\MsJ", shift left, hook, from=2-1, to=1-1]
        \end{tikzcd}
      \end{equation*}
    \end{minipage}
  \end{center}
\end{definition}

\begin{example}
  In CCHM-style cubical type theory \cite{cchm15}, $\Path$-types are only
  pre-$\Id$-types as they do not necessarily have the required (definitional)
  $\MsJ$-elimination structure as in \Cref{def:J-elim}.
\end{example}

\begin{remark}
  Let $\pi \colon \tMcU \to \McU$ be a universal map in a finitely complete
  category $\bC$ equipped with a pre-$\Id$-type structure.

  For each $\pi$-fibrant $E \twoheadrightarrow B$ with $\pi$-fibrancy structure
  $(\ceil{E},\widetilde{\ceil{E}})$ as on the left, we also denote by
  $\Id_{B}(E)$ and $\Id_{\ceil{E}}(\widetilde{\ceil{E}})$ and $\refl_E$ the
  correspondingly labelled maps as on the right.
  \begin{center}
    \begin{minipage}{0.40\linewidth}
      \begin{equation*}
        \begin{tikzcd}
          E & \tMcU \\
          {B} & \McU
          \arrow["{\widetilde{\ceil{E}}}", from=1-1, to=1-2]
          \arrow[from=1-1, to=2-1, two heads]
          \arrow["\lrcorner"{anchor=center, pos=0.15, scale=1.5}, draw=none, from=1-1, to=2-2]
          \arrow[from=1-2, to=2-2, two heads]
          \arrow["{\ceil{E}}"', from=2-1, to=2-2]
        \end{tikzcd}
      \end{equation*}
    \end{minipage}
    \begin{minipage}{0.50\linewidth}
      \begin{equation*}
        \begin{tikzcd}[column sep=large]
          E & \tMcU \\
          {\Id_{B}(E)} & {\Id_\McU(\tMcU)} & \tMcU \\
          {E \times_{B} E} & {\tMcU \times_\McU \tMcU} & \McU \\
          {B} & \McU
          \arrow["\widetilde{\ceil{E}}", from=1-1, to=1-2]
          \arrow["{\refl_E}"{description}, from=1-1, to=2-1, hook]
          \arrow["\lrcorner"{anchor=center, pos=0.15, scale=1.5}, draw=none, from=1-1, to=2-2]
          \arrow["\refl"{description}, from=1-2, to=2-2, hook]
          \arrow["{\Id_{\ceil{E}}(\widetilde{\ceil{E}})}", from=2-1, to=2-2]
          \arrow["{\ev_\partial}"{description}, from=2-1, to=3-1, two heads]
          \arrow["\lrcorner"{anchor=center, pos=0.15, scale=1.5}, draw=none, from=2-1, to=3-2]
          \arrow[from=2-2, to=2-3]
          \arrow["{\ev_\partial}"{description}, from=2-2, to=3-2, two heads]
          \arrow["\lrcorner"{anchor=center, pos=0.15, scale=1.5}, draw=none, from=2-2, to=3-3]
          \arrow[from=2-3, to=3-3, two heads]
          \arrow["{\widetilde{\ceil{E}} \times_{\ceil{E}} \widetilde{\ceil{E}}}", from=3-1, to=3-2]
          \arrow[from=3-1, to=4-1]
          \arrow["\lrcorner"{anchor=center, pos=0.15, scale=1.5}, draw=none, from=3-1, to=4-2]
          \arrow["\Id"', from=3-2, to=3-3]
          \arrow[from=3-2, to=4-2]
          \arrow["{\ceil{E}}"', from=4-1, to=4-2]
        \end{tikzcd}
      \end{equation*}
    \end{minipage}
  \end{center}
  where $\Id_B(E) \twoheadrightarrow E \times_B E$ is chosen by the universe
  structure.
\end{remark}

The axiomatisation of logical quantifiers in universe categories syntactically
encodes their treatment as adjoints.
In the spirit of logical quantifiers as adjoints, $\Sigma$-types are
type-theoretic manifestations of existentials and so $\Sigma$-type structures
are axiomatised using the generic composite from \Cref{constr:gen-comp}.
\begin{definition}\label{def:sigma-type}
  Let $\pi \colon \tMcU \to \McU$ be a universal map in a finitely complete
  category $\bC$.

  An \emph{extensional $\Sigma$-type} structure on $\pi\colon\tMcU \to \McU$ is
  a $\pi$-fibrancy structure $(\Sigma,\pair)$ on $\GenComp(\pi, \pi)$.
  \begin{equation*}
    \begin{tikzcd}[row sep=small, column sep=small]
      {\ev^*(\tMcU \times \tMcU)} &&& \tMcU \\
      {\pi^*\pi_*(\tMcU \times \McU)} & {\tMcU \times \tMcU} & \tMcU \\
      & {\tMcU \times \McU} & \McU \\
      {\pi_*(\tMcU \times \McU)} &&& \McU \\
      && \tMcU \\
      && \McU
      \arrow["\pair", dashed, from=1-1, to=1-4]
      \arrow[from=1-1, to=2-1, two heads]
      \arrow[from=1-1, to=2-2]
      \arrow["\pi", from=1-4, to=4-4, two heads]
      \arrow["\ev\relax"{description}, from=2-1, to=3-2]
      \arrow[from=2-1, to=4-1, two heads]
      \arrow[from=2-2, to=2-3]
      \arrow[from=2-2, to=3-2, two heads]
      \arrow[from=2-3, to=3-3, two heads]
      \arrow[from=3-2, to=3-3]
      \arrow["\pi", from=5-3, to=6-3, two heads]
      \arrow[from=4-1, to=6-3]
      \arrow["\Sigma"'{pos=0.8}, dashed, from=4-1, to=4-4]
      \arrow[curve={height=18pt}, from=2-1, to=5-3, crossing over]
      \arrow["{\proj_1}"{description, pos=0.3}, from=3-2, to=5-3, crossing over]
    \end{tikzcd}
  \end{equation*}
\end{definition}

Also in the spirit of logical quantifiers as adjoints, and by the Beck-Chevalley
argument as described in \Cref{lem:gen-comp-psfw}, $\Pi$-types are
type-theoretic manifestations of universal quantifiers and so $\Pi$-type
structures as axiomatised using the polynomial functor from
\Cref{def:expn-polynomial} associated with the universal fibration.
\begin{definition}[{\cite[Definition 2.2]{voe17}}]\label{def:pi-type}
  Let $\pi \colon \tMcU \to \McU$ be a universal map in a finitely complete category
  $\bC$.

  An \emph{extensional $\Pi$-type} structure on $\tMcU \to \McU$ is a
  $\pi$-fibrancy structure $(\lam, \Pi)$ on the image of the polynomial functor
  associated with $\tMcU \to \McU$ applied to itself.
  \begin{equation*}
    \begin{tikzcd}[cramped]
      {\tMcU \times \tMcU} & {\pi_*(\tMcU \times \tMcU)} & \tMcU \\
      {\tMcU \times \McU} & {\pi_*(\tMcU \times \McU)} & \McU \\
      \tMcU \\
      & \McU
      \arrow[from=1-1, to=2-1]
      \arrow["\lam", dashed, from=1-2, to=1-3]
      \arrow[from=1-2, to=2-2]
      \arrow["\lrcorner"{anchor=center, pos=0.15, scale=1.5}, draw=none, from=1-2, to=2-3]
      \arrow["\pi", from=1-3, to=2-3]
      \arrow[from=2-1, to=3-1]
      \arrow["\Pi"', dashed, from=2-2, to=2-3]
      \arrow[from=2-2, to=4-2]
      \arrow["\pi"', from=3-1, to=4-2]
    \end{tikzcd}
  \end{equation*}
\end{definition}

The $\Sigma$- and $\Pi$-type structures from \Cref{def:pi-type,def:sigma-type}
are indeed the generic fibred logical quantifiers in the following sense.
\begin{proposition}\label{prop:Pi-Sigma-generic}
  Let $\pi \colon \tMcU \to \McU$ be a universal map in a finitely complete
  category $\bC$.
  \begin{enumerate}
    \item If $\tMcU \to \McU$ is equipped with an extensional $\Sigma$-type
    structure then $\pi$-fibrations are closed under composition.
    \item If $\tMcU \to \McU$ is equipped with an extensional $\Pi$-type
    structure then pushforwards of $\pi$-fibrations along $\pi$-fibrations exist
    and remain $\pi$-fibrations.
  \end{enumerate}
\end{proposition}
\begin{proof}
  The first part is by \Cref{lem:gen-comp} and the second part is by
  \Cref{lem:gen-comp-psfw}.
\end{proof}

\begin{corollary}\label{cor:local-hom-fib}
  Let $\pi \colon \tMcU \to \McU$ be a universal map in a finitely complete
  category $\bC$.
  If $\tMcU \to \McU$ has a $\Pi$-type structure then following maps
  $\pi$-fibrations.
  \begin{align*}
    [\tMcU \times \McU, \McU \times \tMcU]_{\McU \times \McU} \to \McU \times \McU
    &&
    [\McU \times \tMcU, \tMcU \times \McU]_{\McU \times \McU} \to \McU \times \McU
       &&
          [\tMcU, \tMcU]_{\McU} \to \McU
  \end{align*}
\end{corollary}
\begin{proof}
  All cases are similar.
  For
  $[\tMcU \times \McU, \McU \times \tMcU]_{\McU \times \McU} \to \McU \times
  \McU$ for example, this internal-Hom is the pushforward of
  $(\tMcU \times \McU) \times_{\McU \times \McU} (\McU \times \tMcU) = \tMcU
  \times \tMcU \to \McU \times \tMcU$ along
  $\McU \times \tMcU \to \McU \times \McU$, so the result follows by
  \Cref{prop:Pi-Sigma-generic}.
\end{proof}

\subsection{CwR Presentations of Type Theory}
The universe category axiomatisation of various logical constructions also lend
themselves easily to formulation of various type-theories as CwRs.

We first recall the definition of CwRs by \citeauthor{uem23} \cite{uem23} and
the bicocompleteness of the category of CwRs due to \citeauthor{jel25}
\cite{jel25}.

\begin{definition}[{\cite[Definition 2.3.1]{jel25}, \cite[Definition
    3.2.1]{uem23}}]\label{def:cwr}
  A \emph{category with representable maps (CwR)} is a category $\bC$ with
  finite limits equipped with a replete wide subcategory of pullback-stable
  class $\MsR_\bC$ of exponentiable maps called the \emph{representable maps}.

  A map of CwRs is a map between their underlying categories preserving finite
  limits, representable maps, and pushforwards along representable maps.

  Denote by $\CwR$ the 2-category of (small) CwRs, maps of CwRs, and natural
  transformations whose square at representable maps are pullbacks.
\end{definition}

\begin{definition}\label{def:marked-cat}
  Denote by $\Cat_\Msm$ the 2-category of (small) categories with marked maps and squares.

  That is, a marked category is a category $\bC$ equipped with two specific
  choices of replete wide subcategories $\McM_\bC \hookrightarrow \bC$ and
  $\McS_\bC \hookrightarrow \bC^\to$.
  A 1-cell between marked categories is exactly a functor between underlying
  categories sending marked maps and squares to marked maps and squares.
  A 2-cell between 1-cells of marked categories is a natural transformation
  whose naturality square at marked arrows are also marked squares.

  Often, we denote the marked maps by the arrow ${\bulletto}$.
\end{definition}

The representable maps of a CwR give it the structure of a marked category by
taking the marked maps as representable maps and marked squares as pullback
squares.
Conversely, each marked category freely gives rise to a CwR due to the following
result by \citeauthor{jel25}.

\begin{theorem}[{\cite[Corollaries 3.2.16 and 3.2.17]{jel25}}]
  The forgetful 2-functor $\abs{-} \colon \CwR \to \Cat_\Msm$ has a left
  biadjoint $\vbr{-} \colon \Cat_\Msm \to \CwR$ and $\CwR$ has all bicolimits.
  \def\endingmark{\qedsymbol}
\end{theorem}

Bicocompleteness of CwRs give rise to the following construction.
\begin{construction}\label{constr:cwr-free-quotient}
  Let $\bC$ be a CwR.
  Fix two maps $\bI \to \bJ$ and $\bI \to \abs{\bC}$ of
  marked categories.
  Then, we write $\bC \cup_\bI \bJ$ for the following bipushout in $\CwR$
  \begin{equation*}
    \begin{tikzcd}[cramped]
      {\vbr{\bI}} & \bC \\
      {\vbr{\bJ}} & {\bC \cup_\bI \bJ}
      \arrow[from=1-1, to=1-2]
      \arrow[from=1-1, to=2-1]
      \arrow[from=1-2, to=2-2]
      \arrow[from=2-1, to=2-2]
      \arrow["\lrcorner"{anchor=center, pos=0.15, scale=1.5, rotate=180}, draw=none, from=2-2, to=1-1]
    \end{tikzcd}
  \end{equation*}
  where the map $\vbr{\bI} \to \bC$ in the top row is the
  $(\vbr{-} \dashv \abs{-})$-transpose of the map $\bI \to \abs{\bC}$ and map in
  the left row is the image of $\bI \to \bJ$ under $\vbr{-}$, as per
  \cite[Corollaries 3.2.16 and 3.2.17]{jel25}.
\end{construction}

This construction allows one to derive the following constructions.
\begin{construction}\label{constr:cwr-free-basic}
  Let $\bC$ be a CwR.
  \begin{enumerate}
    \item\label{itm:cwr-add-free-map} The CwR obtained by formally adjoining a map between two objects
    $X,Y \in \bC$ is the bipushout
    \begin{equation*}
      \bC[X \to Y] \coloneqq \bC \cup_{\set{X\quad Y}} \set{X \to Y} \in \CwR
    \end{equation*}
    \item\label{itm:cwr-add-free-square} The CwR obtained by formally adjoining a commutative square between
    two maps $f \colon X \to X'$ and $g \colon Y \to Y'$ is the bipushout
    \begin{equation*}
      \bC[f \to g] =
      \bC\left[
        \begin{tikzcd}[cramped, row sep=small, column sep=small]
          X \ar[d] \ar[r] & Y \ar[d] \\
          X' \ar[r] & Y'
        \end{tikzcd}
      \right]
      \coloneqq
      \bC \cup_{\left\{\begin{tikzcd}[cramped, row sep=tiny, column sep=tiny]\scriptscriptstyle X \ar[d] &\scriptscriptstyle Y \ar[d] \\ \scriptscriptstyle X' & \scriptscriptstyle Y' \end{tikzcd} \right\}}
      \left\{
        \begin{tikzcd}[cramped, row sep=small, column sep=small]
          X \ar[d] \ar[r] & Y \ar[d] \\
          X' \ar[r] & Y'
        \end{tikzcd}
      \right\} \in \CwR
    \end{equation*}
    \item\label{itm:cwr-make-free-iso} The CwR obtained by formally inverting a
    wide subcategory of maps $W \hookrightarrow \bC$ is the bipushout along the inclusion of $W$
    into its homotopical category $W^{-1}W$ where all maps are marked as weak
    equivalences.
    \begin{equation*}
      W^{-1}\bC \coloneqq \bC \cup_{W} W^{-1}W
    \end{equation*}
    \item\label{itm:cwr-add-free-pb} The CwR obtained by formally adjoining a pullback square between two
    maps $f \colon X \to X'$ and $g \colon Y \to Y'$ is the bicolimit obtained
    by first formally adjoining a square $f \to g$ and then formally inverting
    the comparison map $X \to X' \times_{Y'} Y$.
    \begin{equation*}
      \bC[f \xrightarrow{\pb} g] =
      \bC\left[
        \begin{tikzcd}[cramped, row sep=small, column sep=small]
          X \ar[d] \ar[r] \ar[rd, draw=none, "\lrcorner"{anchor=center, pos=0.15, scale=1.5}] & Y \ar[d] \\
          X' \ar[r] & Y'
        \end{tikzcd}
      \right]
      \coloneqq (X \to X' \times_{Y'} Y)^{-1}\bC[f \to g]
    \end{equation*}
  \end{enumerate}
\end{construction}

\begin{remark}
  It was pointed out by \citeauthor{jel25} that using marked squares,
  \Cref{constr:cwr-free-basic}\Cref{itm:cwr-add-free-pb} can alternatively be
  accomplished by freely adjoining a \emph{marked} square.
\end{remark}

The constructions above admit the following (1-categorical) universal properties
of the bipushout.
\begin{lemma}\label{lem:cwr-free-constrs}
  Let $\bC$ be a CwR. Referring to \Cref{constr:cwr-free-basic},
  \begin{enumerate}
    \item\label{itm:cwr-free-map} Given objects $X,Y\in\bC$, isomorphism classes
    of maps $\bC[X \to Y] \to \bD \in \CwR$ are in bijective correspondence with
    isomorphism classes of maps $F \colon \bC \to \bD$ and a choice of map
    $FX \to FY \in \bD$.
    \item\label{itm:cwr-free-square} Given two maps $f \colon X \to X'$ and
    $g \colon Y \to Y'$ in $\bC$, isomorphism classes of maps
    $\bC[f \to g] \to \bD \in \CwR$ are in bijective correspondence with
    isomorphism classes of maps $F \colon \bC \to \bD$ and a choice of commuting
    square $Ff \to Fg \in \bD^\to$.
    \item\label{itm:cwr-free-iso} Given a wide subcategory of maps
    $W \hookrightarrow \bC$, isomorphism classes of maps
    $W^{-1}\bC \to \bD \in \CwR$ are in bijective correspondence with
    isomorphism classes of maps $F \colon \bC \to \bD$ sending all maps in $W$
    to isomorphisms.
    \item\label{itm:cwr-free-pb} Given two maps $f \colon X \to X'$ and
    $g \colon Y \to Y'$ in $\bC$, isomorphism classes of maps
    $\bC[f \xrightarrow{\pb} g] \to \bD$ are in bijective correspondence with
    isomorphism classes of maps $F \colon \bC \to \bD$ and a choice of a
    pullback square $Ff \to Fg \in \bD$.
  \end{enumerate}
\end{lemma}
\begin{proof}
  Straightforward by the universal property of the bipushout.
\end{proof}

The various logical structures of intensional type theory can now be directly
translated into the framework of CwRs.
\begin{construction}\label{constr:cwr-frags}
  The various fragments of intensional type theory constructed as CwRs are given
  as follows.
  \begin{enumerate}
    \item[\Cref{def:univ-struct}] \label{itm:cwr-univ} The CwR $\Univ$ of a generic universal map is the
    free CwR generated by the walking marked arrow, which is the same as the one
    from \cite[Notation 4.3.1]{jel25} or \cite[Example 3.2.6]{uem23}.
    \begin{equation*}
      \Univ \coloneqq \vbr{\set{\pi \colon \tMcU \bulletto \McU}}
    \end{equation*}
    \item[\Cref{def:int-univ}] \label{itm:cwr-int-Univ} The CwR $\IntUniv$ of a generic universal map
    equipped with a generic internal (fibrant) universe is the CwR obtained from
    $\Univ$ by formally adjoining a map $\tMcU_0 \to \McU_0$ to $\Univ$ and the
    two pullback squares as from \Cref{def:int-univ} using iterated applications
    of \Cref{constr:cwr-free-basic}.
    \begin{equation*}
      \IntUniv \coloneqq \Univ
      \left[\begin{tikzcd}[cramped, row sep=small, column sep=small]
          \tMcU_0 \ar[d] \\ \McU_0
        \end{tikzcd}
      \right]
      \left[
        \begin{tikzcd}[cramped, row sep=small, column sep=small]
          \tMcU_0 \ar[d] \ar[r] \ar[rd, draw=none, "\lrcorner"{anchor=center, pos=0.15, scale=1.5}] & \McU_0 \ar[d] \\
          \tMcU \ar[r] & \McU
        \end{tikzcd}
      \right]
      \left[
        \begin{tikzcd}[cramped, row sep=small, column sep=small]
          \McU_0 \ar[d] \ar[r] \ar[rd, draw=none, "\lrcorner"{anchor=center, pos=0.15, scale=1.5}] & \tMcU \ar[d] \\
          1 \ar[r] & \McU
        \end{tikzcd}
      \right]
    \end{equation*}
    \item[\Cref{def:Id-type}] \label{itm:cwr-Id} The CwR
    $\Univ_{\text{\sffamily pre-}\Id}$ of a generic universal map equipped with
    a generic pre-$\Id$-type structure is obtained from $\Univ$ by freely
    adjoining a square of the form \Cref{eqn:pre-Id-def} from
    \Cref{def:pre-id-type}.
    \begin{equation*}
      \Univ_{\text{\sffamily pre-}\Id} \coloneqq \Univ
      \left[
        \begin{tikzcd}[cramped, row sep=small, column sep=small]
          \tMcU & \tMcU \\
          {\tMcU \times_\McU \tMcU} & \McU
          \arrow["{\ceil{\refl}}", dashed, from=1-1, to=1-2]
          \arrow[from=1-1, to=2-1, "{\Delta}"']
          \arrow[from=1-2, to=2-2]
          \arrow["\Id"', dashed, from=2-1, to=2-2]
        \end{tikzcd}
      \right]
    \end{equation*}
    The CwR of a generic $\Id$-type structure is constructed from
    $\Univ_{\text{\sffamily pre-}\Id}$ by formally adjoining a local lifting
    structure using \cite[Construction 3.8]{struct-lift}
    \begin{equation*}
      \Univ_{\Id} \coloneqq
      \Univ_{\text{\sffamily pre-}\Id}
      \left[
      \begin{tikzcd}[cramped, row sep=small, column sep=small]
        \tMcU \ar[d,hook,"{\refl}"'] \\ \Id_\McU(\tMcU)
      \end{tikzcd}
      \fracsquareslash{\McU}
      \begin{tikzcd}[cramped, row sep=small, column sep=small]
        \McU\times\tMcU \ar[d] \\ \McU\times\McU
      \end{tikzcd}
      \right]
    \end{equation*}
    \item[\Cref{def:sigma-type}] The CwR $\Univ_{\Sigma}$ of a generic universal
    map equipped with a generic $\Sigma$-type structure is obtained from $\Univ$
    by freely adjoining a pullback square from $\GenComp(\pi,\pi)$ of
    \Cref{constr:gen-comp} to $\pi$
    \begin{equation*}
      \Univ_{\Sigma} \coloneqq
      \Univ
      \left[
        \begin{tikzcd}[cramped, row sep=small, column sep=small]
          \ev^*(\tMcU \times \tMcU) \ar[d,"{\GenComp(\pi,\pi)}"'] \ar[rd, "\lrcorner"{anchor=center, pos=0.15, scale=1.5}, draw=none] & \tMcU \\
          \pi_*(\tMcU \times \McU) & \McU
          \arrow["{\pair}", dashed, from=1-1, to=1-2]
          \arrow[from=1-2, to=2-2]
          \arrow["\Sigma"', dashed, from=2-1, to=2-2]
        \end{tikzcd}
      \right]
    \end{equation*}
    \item[\Cref{def:pi-type}] The CwR $\Univ_{\Pi}$ of a generic universal map
    equipped with a generic $\Pi$-type structure is obtained from $\Univ$ by
    freely adjoining a pullback square from $\pi_*(\tMcU \times \pi)$ of
    \Cref{constr:gen-comp} to $\pi$
    \begin{equation*}
      \Univ_{\Pi} \coloneqq
      \Univ
      \left[
        \begin{tikzcd}[cramped, row sep=small, column sep=small]
          \pi_*(\tMcU \times \tMcU) \ar[d] \ar[rd, "\lrcorner"{anchor=center, pos=0.15, scale=1.5}, draw=none] & \tMcU \\
          \pi_*(\tMcU \times \McU) & \McU
          \arrow["{\lam}", dashed, from=1-1, to=1-2]
          \arrow[from=1-2, to=2-2]
          \arrow["\Pi"', dashed, from=2-1, to=2-2]
        \end{tikzcd}
      \right]
    \end{equation*}
  \end{enumerate}
  Each of these constructions are objects in $\sfrac{\Univ}{\CwR}$, where the
  maps from $\Univ$ are the obvious colimiting maps from their constructions as
  bicolimits.
\end{construction}

The CwRs as above have the following universal properties.

\begin{proposition}\label{prop:cwr-itt}
  Let $\bC$ be a finitely complete category.
  \begin{enumerate}
    \item $\bC$ has the structure of a CwR by taking the representable maps to
    be the exponentiable maps.
    \item Isomorphism classes of maps $\Univ \to \bC \in \CwR$ are in bijective
    correspondence with isomorphism classes of choices of universal maps in
    $\bC$.
    \item Factorisations of $M \colon \Univ \to \bC \in \CwR$ via
    $\Univ \to \bT$ where
    $\bT = \IntUniv,\Univ_\Id,\Univ_\Sigma,\allowbreak\Univ_\Pi$ are
    respectively in bijective correspondence with choices of internal universe,
    $\Id$-type, $\Sigma$-type, $\Pi$-type structures on the universal map
    $M\pi \colon M\tMcU \to M\McU \in \bC$
  \end{enumerate}
\end{proposition}
\begin{proof}
  We show each part sequentially.

  The class of exponentiable maps of $\bC$ is closed under composition by
  composition of right adjoints.
  It is also closed under pullbacks by \cite[Corollary 1.4]{nie82}.
  This shows the first part.

  The CwR $\bC$ where representable maps are taken to be exponentiable maps also
  admits the structure of a category with marked squares where the marked maps
  are the exponentiable maps and the marked squares are taken to be the pullback
  squares.
  Then, because $\Univ = \vbr{\set{\pi \colon \tMcU \bulletto \McU}}$ is the
  image of the walking marked map under the left adjoint
  $\vbr{-} \colon \Cat_\Msm \to \CwR$ of \cite[Collary 3.2.17]{jel25}, the
  result of the second part follows by definition of the biadjunction.

  The final part follows immediately from the universal property of the bipushout for
  $\Pi$- and $\Sigma$-type.
  For the $\Id$-types, we use \cite[Theorem 3.9]{struct-lift}.
\end{proof}

\subsection{Id-type Homotopy Theory}
As $\Id$-types are generic fibred path objects for fibrations, one may phrase a
variety of homotopical concepts and constructions using pre-$\Id$-types.
We recall these ideas.

We first start with the logical concept of a \emph{propositional equality}.
\begin{definition}
  Let $\pi \colon \tMcU \to \McU$ be a universal map in a finitely complete
  category $\bC$ equipped with a pre-$\Id$-type structure
  $\Id \colon \tMcU \times_\McU \tMcU \to \McU$.

  For a fixed object $\Gamma \in \bC$ and two maps
  $\ceil{t_0},\ceil{t_1} \colon \Gamma \rightrightarrows \tMcU$ such that
  $\pi \ceil{t_0} = \pi \ceil{t_1}$, a \emph{propositional equality} between
  $\ceil{t_0}$ and $\ceil{t_1}$ is a factorisation
  $\ceil{H} \colon \Gamma \to \tMcU$ of
  $\Id(\ceil{t_0},\ceil{t_1}) \colon \Gamma \to \tMcU \times_\McU \tMcU \to
  \McU$ as follows.
  \begin{equation*}
    \begin{tikzcd}[cramped]
      && \tMcU \\
      \Gamma & {\tMcU \times_\McU \tMcU} & \McU
      \arrow["\pi", two heads, from=1-3, to=2-3]
      \arrow["{\ceil{H}}", dashed, from=2-1, to=1-3]
      \arrow["{(\ceil{t_0}, \ceil{t_1})}"', from=2-1, to=2-2]
      \arrow["\Id"', from=2-2, to=2-3]
    \end{tikzcd}
  \end{equation*}
\end{definition}
\begin{definition}
  Let $\pi \colon \tMcU \to \McU$ be a universal map in a finitely complete
  category $\bC$ equipped with an $\Id$-type structure
  $\Id \colon \tMcU \times_\McU \tMcU \to \McU$.

  Given $\Gamma \in \bC$, a \emph{propositional automorphism} at
  $A \colon \Gamma \to \McU$ consists of some
  $\ceil{t} \colon \Gamma.A \to \tMcU$ equipped with a propositional equality to
  $\var_A \colon \Gamma.A \to \tMcU$ (so in particular
  $\pi \ceil{t} = A \cdot \pi_A$).
\end{definition}

Taking pullbacks gives rise to the homotopical concept of a \emph{fibred
  homotopy} between parallel maps.
\begin{definition}
  Let $\pi \colon \tMcU \to \McU$ be a universal map in a finitely complete
  category $\bC$ equipped with a pre-$\Id$-type structure
  $\Id \colon \tMcU \times_\McU \tMcU \to \McU$.

  For a fixed $\pi$-fibrant object $E \twoheadrightarrow B \in \sfrac{\bC}{B}$
  and two parallel maps $f,g \colon X \rightrightarrows B \in \sfrac{\bC}{B}$,
  an \emph{$\Id$-homotopy} between $f$ and $g$ is a factorisation of
  $(f,g) \colon X \to E \times_B E$ via
  $\ev_\partial \colon \Id_B(E) \twoheadrightarrow E \times_B E$.
  \begin{equation*}
    \begin{tikzcd}[cramped]
      & {\Id_B(E)} \\
      X & {E \times_B E}
      \arrow["{\ev_\partial}", two heads, from=1-2, to=2-2]
      \arrow["H", dashed, from=2-1, to=1-2]
      \arrow["{(f,g)}"', from=2-1, to=2-2]
    \end{tikzcd}
  \end{equation*}
\end{definition}
\begin{definition}
  Let $\pi \colon \tMcU \to \McU$ be a universal map in a locally cartesian
  closed category $\bC$ equipped with a pre-$\Id$-type structure
  $\Id \colon \tMcU \times_\McU \tMcU \to \McU$.

  For a fixed $\pi$-fibrant object $E \twoheadrightarrow B \in \sfrac{\bC}{B}$,
  an automorphism $t \colon E \to E \in \sfrac{\bC}{B}$ is an
  \emph{$\Id$-homotopy identity} when it is $\Id$-homotopic to the identity
  map at $E$.

  An \emph{$\Id$-homotopy isomorphism} or \emph{$\Id$-homotopy equivalence} is a
  map $f \colon E_0 \to E_1 \in \sfrac{\bC}{B}$ for $\pi$-fibrant
  $E_0,E_1 \twoheadrightarrow B$ equipped with maps
  $s,r \colon E_0 \rightrightarrows E_1 \in \sfrac{\bC}{B}$ such that $fs$ and
  $rf$ are both \emph{$\Id$-homotopy identities}.
\end{definition}

Under the presence of $\Pi$-type and full $\Id$-type structures, pullbacks of
$\refl$ along fibrations also lift on the left against the universal fibration.
As part of the consequence, we obtain the usual $\transport$ map.
\begin{construction}\label{constr:transport}
  Let $\pi \colon \tMcU \to \McU$ be a universal map in a finitely complete
  $\bC$ equipped with a $\Pi$-type structure and a full-$\Id$-type structure
  $\Id \colon \tMcU \times_\McU \tMcU \to \McU$.

  Suppose that we have composable maps
  $X \twoheadrightarrow E \twoheadrightarrow B$ where each map is a
  $\pi$-fibration.
  The goal is to construct a dashed map
  $\transport \colon \ev_0^*X \to \ev_1^*X$ over $\Id_B(E)$ such that its
  pullback along $\refl \colon E \hookrightarrow \Id_B(E)$ is the
  identity at $X$.
  \begin{equation*}
    \begin{tikzcd}
      X && {\ev_0^*X} && X \\
      & X && {\ev_1^*X} \\
      E && {\Id_B(E)} && E
      \arrow[from=1-1, to=1-3, hook]
      \arrow["{=}"{description}, from=1-1, to=2-2]
      \arrow[from=1-1, to=3-1, two heads]
      \arrow[from=1-3, to=1-5]
      \arrow["\transport"{description}, dashed, from=1-3, to=2-4]
      \arrow[from=1-3, to=3-3, two heads]
      \arrow[from=1-5, to=3-5, two heads]
      \arrow[hook, crossing over, from=2-2, to=2-4]
      \arrow[from=2-2, to=3-1, two heads]
      \arrow[from=2-4, to=1-5]
      \arrow[from=2-4, to=3-3, two heads]
      \arrow["\refl"', hook, from=3-1, to=3-3]
      \arrow["{=}"{description}, curve={height=24pt}, from=3-1, to=3-5]
      \arrow["{\ev_1}"', shift right, from=3-3, to=3-5]
      \arrow["{\ev_0}", shift left, from=3-3, to=3-5]
    \end{tikzcd}
  \end{equation*}

  We do so by solving the following lifting problem
  \begin{equation*}
    \begin{tikzcd}[cramped]
      X & {\ev_1^*X} \\
      {\ev_0^*X} & {\Id_B(E)}
      \arrow[hook, from=1-1, to=1-2]
      \arrow[hook, from=1-1, to=2-1]
      \arrow[two heads, from=1-2, to=2-2]
      \arrow["\transport"{description}, dashed, from=2-1, to=1-2]
      \arrow[two heads, from=2-1, to=2-2]
    \end{tikzcd}
  \end{equation*}
  which has a solution because $X \hookrightarrow \ev_0^*X$ is the pullback of
  $E \hookrightarrow \Id_B(E)$ along a $\pi$-fibration
  $\ev_0^*X \twoheadrightarrow \Id_B(E)$ and $\pi$ admits a $\Pi$-type
  structure.
\end{construction}

The $\transport$ map above required the trivial cofibration (i.e. left lifting)
property of the $\refl$ map which is given by the $\MsJ$-structure on a full
$\Id$-type structure.
In the case of pre-$\Id$-types (such as $\Path$-types of CCHM cubical type
theory), we may not necessarily have the $\MsJ$-elimination available.
Nevertheless, we can still require a $\transport$ structure on a pre-$\Id$-type,
which is the map that one would have obtained from applying
\Cref{constr:transport} by taking $X \twoheadrightarrow E \twoheadrightarrow B$
to be the generic pair of composable fibrations
$\ev^*(\tMcU \times \tMcU) \twoheadrightarrow \pi^*\pi_*(\tMcU \times \McU)
\twoheadrightarrow \pi_*(\tMcU \times \McU)$ from \Cref{constr:gen-comp}.

\begin{definition}\label{def:transport}
  Let $\pi \colon \tMcU \to \McU$ be a universal map in a finitely complete
  $\bC$.
  A \emph{$\transport$-structure} on a pre-$\Id$-type structure
  $\Id \colon \tMcU \times_\McU \tMcU \to \McU$ is some dashed map $\transport$
  as below, where the pre-$\Id$-type is taken for the $\pi$-fibration
  $\pi^*\pi\*(\tMcU \times \McU) \twoheadrightarrow \pi_*(\tMcU \times \McU)$,
  as from \Cref{constr:gen-comp}.
  \begin{equation*}
    \begin{tikzcd}[cramped, column sep=small]
      {\ev_0^*\ev^*(\tMcU \times \tMcU)} && {\ev^*(\tMcU \times \tMcU)} & {\tMcU \times \tMcU} \\
      & {\ev_1^*\ev^*(\tMcU \times \tMcU)} \\
      {\Id_{\pi_*(\tMcU \times \McU)}(\pi^*\pi_*(\tMcU \times \McU))} && {\pi^*\pi_*(\tMcU \times \McU)} & {\tMcU \times \McU}
      \arrow[from=1-1, to=1-3]
      \arrow["\transport"{description}, dashed, from=1-1, to=2-2]
      \arrow[from=1-1, to=3-1]
      \arrow[from=1-3, to=1-4]
      \arrow[from=1-3, to=3-3]
      \arrow["\lrcorner"{anchor=center, pos=0.15, scale=1.5}, draw=none, from=1-3, to=3-4]
      \arrow[from=1-4, to=3-4]
      \arrow[from=2-2, to=1-3]
      \arrow[from=2-2, to=3-1]
      \arrow["{\ev_1}"', shift right, from=3-1, to=3-3]
      \arrow["{\ev_0}", shift left, from=3-1, to=3-3]
      \arrow["\ev", from=3-3, to=3-4]
    \end{tikzcd}
  \end{equation*}
\end{definition}

Then immediate by definition, one observes that \Cref{constr:transport} equips
full-$\Id$-type structures with a (canonical) $\transport$-structure under the
presence of $\Pi$-type structures.


\section{Generic Object of Homotopy Isomorphisms}\label{sec:gen-hiso}
For this section, we fix, in a finitely complete category $\bC$, a universal map
$\pi \colon \tMcU \to \McU$ (or more generally any $\pi$-fibration
$E \twoheadrightarrow B$, which admits an inherited choice of pullbacks making
it an universal map).
To even begin stating univalence of $\tMcU \to \McU$, we need to first construct
the object representing the $\Id$-type homotopy equivalences between $\pi$-small
fibrations when $\tMcU \to \McU$ is equipped with a pre-$\Id$-type structure
$\Id \colon \tMcU \times_\McU \tMcU \to \McU$.

As foreshadowing, when stating univalence, the constructions in this section
will be applied to the universal map $\pi_0 \colon \tMcU_0 \to \McU_0$ of an
internal universe structure on $\tMcU \to \McU$.
As such, the construction will end up classifying the $\Id_0$-type homotopy
equivalences between $\pi_0$-small fibrations when $\tMcU_0 \to \McU_0$ is
equipped with a pre-$\Id_0$-type structure
$\Id_0 \colon \tMcU_0 \times_{\McU_0} \tMcU_0 \to \tMcU_0$.
However, to prevent from having subscripts everywhere, we drop the subscript,
since this construction itself can be performed for any universal map.

\subsection{Generic Object of Homotopy Identities}
We first construct the isomorphism class of object representing $\pi$-fibrations
paired with an automorphism homotopic to the identity.
\begin{lemma}\label{lem:gen-auto-id}
  Assume that $\tMcU \to \McU$ has a pre-$\Id$-type structure
  $\Id \colon \tMcU \times_\McU \tMcU \to \McU$.

  In the diagram below,
  \begin{equation*}
    \small
    \begin{tikzcd}[cramped, column sep=tiny]
      {\tMcU \times_\McU \tMcU} && \tMcU & {\pi_*(\tMcU \times_\McU \tMcU)} \\
      & {\tMcU \times (\tMcU \times_\McU \tMcU)} &&& {\pi_*(\tMcU \times (\tMcU \times_\McU \tMcU))} \\
      && {\tMcU \times \McU} &&& {\pi_*(\tMcU \times \McU)} \\
      & \tMcU && \McU
      \arrow[from=1-1, to=1-3]
      \arrow["{(\proj_1, \id)}"{description}, from=1-1, to=2-2]
      \arrow["{\proj_1}"{description}, from=1-1, to=4-2]
      \arrow["\lrcorner"{anchor=center, pos=0, rotate=0, scale=1.5, pos=0.05}, draw=none, from=1-1, to=4-4]
      \arrow[from=1-3, to=4-4]
      \arrow["{\pi_*(\proj_1, \id)}"{description}, from=1-4, to=2-5]
      \arrow[from=1-4, to=4-4]
      \arrow["{\tMcU \times \Id}"{description}, from=2-2, to=3-3]
      \arrow["{\proj_1}"{description}, from=2-2, to=4-2]
      \arrow["{\pi_*(\tMcU \times \Id)}"{description}, from=2-5, to=3-6]
      \arrow[from=2-5, to=4-4]
      \arrow["{\proj_1}"{description}, from=3-3, to=4-2]
      \arrow[from=3-6, to=4-4]
      \arrow[from=4-2, to=4-4]
    \end{tikzcd}
  \end{equation*}
  taking the left row
  $\tMcU \times_\McU \tMcU \xrightarrow{\proj_1,\id} \tMcU \times (\tMcU
  \times_\McU \tMcU) \xrightarrow{\tMcU \times \Id} \tMcU \times \McU$, applying
  the pushforward along $\pi \colon \tMcU \to \McU$ to obtain the composite
  \begin{equation*}
    \pi_*(\tMcU \times_\McU \tMcU) \xrightarrow{\pi_*(\proj_1,\id)} \pi_*(\tMcU
    \times (\tMcU \times_\McU \tMcU)) \xrightarrow{\pi_*(\tMcU \times \Id)}
    \pi_*(\tMcU \times \McU)
  \end{equation*}
  and further applying the Yoneda embedding obtains the natural transformation
  of presheaves
  \newsavebox{\tType}
  \begin{lrbox}{\tType}\begin{tikzcd}
      \Gamma.A \ar[r, "{\ceil{t}}"] \ar[d, "{\var_A}"'] & \tMcU \ar[d] \\
      \tMcU \ar[r] & \McU
    \end{tikzcd}
  \end{lrbox}
  \begin{align*}
    \left(\coprod_{\Gamma \xrightarrow{A} \McU}
    \left\{ \Gamma.A \xrightarrow{\ceil{t}} \tMcU ~\middle|~
    \usebox{\tType} \right\}\right)_{\Gamma \in \bC}
    &\to
      \left(\coprod_{\Gamma \xrightarrow{A} \McU}
      \left\{ \Gamma.A \xrightarrow{B} \McU \right\}\right)_{\Gamma \in \bC}
    \\
    \left(\Gamma \xrightarrow{A} \McU, \Gamma.A \xrightarrow{\ceil{t}} \tMcU\right)
    &\mapsto
      \left(\Gamma \xrightarrow{A} \McU,
      \Gamma.A \xrightarrow{(\var_A,\ceil{t})} \tMcU \times_\McU \tMcU \xrightarrow{\Id} \McU\right)
  \end{align*}
\end{lemma}
\begin{proof}
  We first note that given any map viewed as an object
  $A \colon \Gamma \to \McU \in \sfrac{\bC}{\McU}$ in the slice category, a map
  $\Gamma \to \pi_*(\tMcU \times \tMcU)$ corresponds by transpose to a map
  $\pi^*\Gamma \cong \Gamma.A \to \tMcU \times_\McU \tMcU \cong \pi^*\tMcU \in
  \sfrac{\bC}{\tMcU}$.
  Because pullback is right adjoint to post-composition, such a map is in
  bijective correspondence with a map $\Gamma.A \to \tMcU$ over $\McU$.

  Therefore, for any fixed $\Gamma \in \bC$, we can write the $\Hom$-set
  $\bC(\Gamma, \pi_*(\tMcU \times_\McU \tMcU)$ as
  \begin{align*}
    \bC(\Gamma, \pi_*(\tMcU \times_\McU \tMcU))
    &= \coprod_{A \colon \Gamma \to \McU} \sfrac{\bC}{\McU}(A, \pi_*(\tMcU \times_\McU \tMcU))
    \\
    &\cong
      \coprod_{A \colon \Gamma \to \McU} \sfrac{\bC}{\tMcU}(\pi^*A, \tMcU \times_\McU \tMcU)
    \\
    &\cong
      \coprod_{A \colon \Gamma \to \McU} \sfrac{\bC}{\tMcU}(\Gamma.A, \tMcU \times_\McU \tMcU)
    \\
    &\cong
      \coprod_{A \colon \Gamma \to \McU} \sfrac{\bC}{\tMcU}(\Gamma.A, \pi^*\tMcU)
    \\
    &\cong
      \coprod_{A \colon \Gamma \to \McU} \sfrac{\bC}{\McU}(\Gamma.A, \tMcU)
  \end{align*}
  This shows that $\pi_*(\tMcU \times_\McU \tMcU)$ represents the presheaf with
  components
  \begin{lrbox}{\tType}\begin{tikzcd}
      \Gamma.A \ar[r, "{\ceil{t}}"] \ar[d, "{\var_A}"'] & \tMcU \ar[d] \\
      \tMcU \ar[r] & \McU
    \end{tikzcd}
  \end{lrbox}
  \begin{equation*}
    \left(\coprod_{\Gamma \xrightarrow{A} \McU}
      \left\{ \Gamma.A \xrightarrow{\ceil{t}} \tMcU ~\middle|~
        \usebox{\tType} \right\}\right)_{\Gamma \in \bC}
  \end{equation*}

  Now, given a map $\Gamma \to \pi_*(\tMcU \times_\McU \tMcU)$ corresponding to
  $(A \colon \Gamma \to \McU, \ceil{t} \colon \Gamma.A \to \tMcU)$, taking the
  $(\pi^*\dashv\pi_*)$-transpose of the composite
  $\Gamma \to \pi_*(\tMcU \times_\McU \tMcU) \to \pi_*(\tMcU \times \McU)$ is
  the same as taking the composite of the transpose
  $\Gamma.A \to \tMcU \times_\McU \tMcU \to \tMcU \times \McU$, which precisely
  gives the composite
  \begin{equation*}
    \Gamma.A \xrightarrow{(\var_A, \ceil{t})}
    \tMcU \times_\McU \tMcU \xrightarrow{(\proj_1,\id)}
    \tMcU \times (\tMcU \times_\McU \tMcU) \xrightarrow{\tMcU \times \Id}
    \tMcU \times \McU \in \sfrac{\bC}{\tMcU}
  \end{equation*}
  Further taking the image of this above map under the isomorphism
  \begin{align*}
    \sfrac{\bC}{\tMcU}(\Gamma.A, \tMcU \times \McU) \cong \bC(\Gamma.A, \McU)
  \end{align*}
  gives the map
  \begin{equation*}
    \Gamma.A \xrightarrow{(\var_A, \ceil{t})} \tMcU \times_\McU \tMcU \xrightarrow{\Id} \McU
  \end{equation*}

  Therefore, the following diagram summarises our operations.
  \begin{equation*}\small
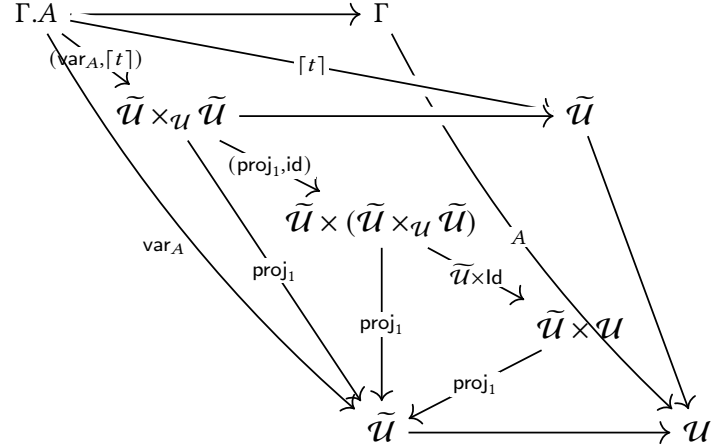
    \begin{tikzcd}[column sep=small]
      {\Gamma.A} && \Gamma \\
      & {\tMcU \times_\McU \tMcU} && \tMcU \\
      && {\tMcU \times (\tMcU \times_\McU \tMcU)} \\
      &&& {\tMcU \times \McU} \\
      && \tMcU && \McU
      \arrow[from=1-1, to=1-3]
      \arrow["{(\var_A, \ceil{t})}"{description}, from=1-1, to=2-2]
      \arrow["{\var_A}"', curve={height=12pt}, from=1-1, to=5-3]
      \arrow["{A}"{description}, curve={height=12pt}, from=1-3, to=5-5]
      \arrow["{(\proj_1, \id)}"{description}, from=2-2, to=3-3]
      \arrow["{\proj_1}"{description}, from=2-2, to=5-3]
      \arrow[from=2-4, to=5-5]
      \arrow["{\tMcU \times \Id}"{description}, from=3-3, to=4-4]
      \arrow["{\proj_1}"{description}, from=3-3, to=5-3]
      \arrow["{\proj_1}"{description}, from=4-4, to=5-3]
      \arrow[from=5-3, to=5-5]
      \arrow["{\ceil{t}}"{description}, from=1-1, to=2-4, crossing over]
      \arrow[crossing over, from=2-2, to=2-4]
    \end{tikzcd}
  \end{equation*}
\end{proof}

As an immediate corollary, one obtains the following object representing
propositional identities.
\begin{construction}\label{constr:gen-auto}
  %
  Assume $\tMcU \to \McU$ has a pre-$\Id$-type structure
  $\Id \colon \tMcU \times_\McU \tMcU \to \McU$.

  We define $\HId_\McU^\Id(\tMcU)$ as (the isomorphism class of) the pullback
  \begin{equation*}\small
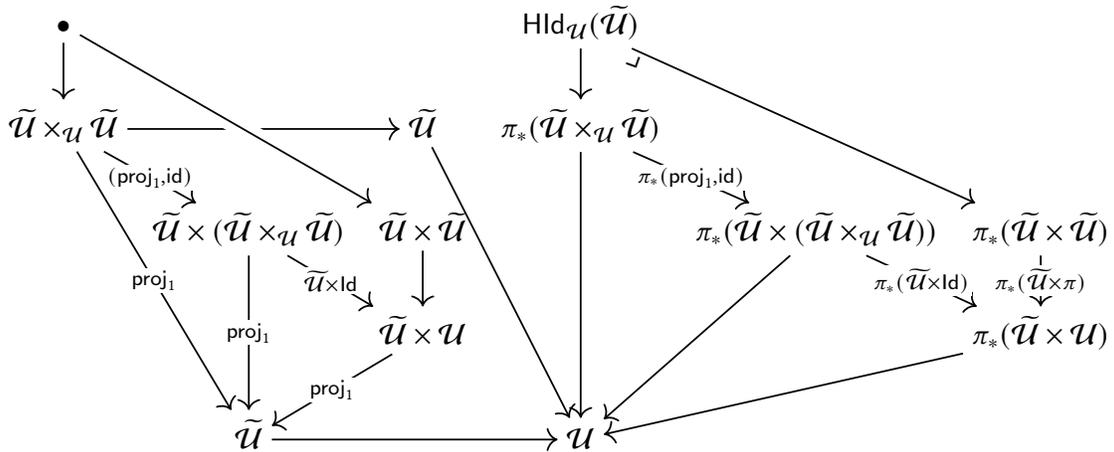
    \begin{tikzcd}[column sep=tiny, cramped]
      \bullet &&& {\HId_\McU(\tMcU)} \\
      {\tMcU \times_\McU \tMcU} && \tMcU & {\pi_*(\tMcU \times_\McU \tMcU)} \\
      & {\tMcU \times (\tMcU \times_\McU \tMcU)} & {\tMcU \times \tMcU} && {\pi_*(\tMcU \times (\tMcU \times_\McU \tMcU))} & {\pi_*(\tMcU \times \tMcU)} \\
      && {\tMcU \times \McU} &&& {\pi_*(\tMcU \times \McU)} \\
      & \tMcU && \McU
      \arrow[from=1-1, to=2-1]
      \arrow[from=1-4, to=2-4]
      \arrow[from=1-4, to=3-6]
      \arrow[from=2-1, to=2-3]
      \arrow["{(\proj_1, \id)}"{description}, from=2-1, to=3-2]
      \arrow["{\proj_1}"{description}, from=2-1, to=5-2]
      \arrow[from=2-3, to=5-4]
      \arrow["{\pi_*(\proj_1, \id)}"{description}, from=2-4, to=3-5]
      \arrow[from=2-4, to=5-4]
      \arrow["{\tMcU \times \Id}"{description}, from=3-2, to=4-3]
      \arrow["{\proj_1}"{description}, from=3-2, to=5-2]
      \arrow[from=3-3, to=4-3]
      \arrow["{\pi_*(\tMcU \times \Id)}"{description}, from=3-5, to=4-6]
      \arrow[from=3-5, to=5-4]
      \arrow["{\pi_*(\tMcU \times \pi)}"{description}, from=3-6, to=4-6]
      \arrow["{\proj_1}"{description}, from=4-3, to=5-2]
      \arrow[from=4-6, to=5-4]
      \arrow[from=5-2, to=5-4]
      \arrow[crossing over, from=1-1, to=3-3]
      \arrow["\lrcorner"{anchor=center, pos=0, rotate=-20, scale=1.5, pos=0.05}, draw=none, from=1-4, to=4-6]
    \end{tikzcd}
  \end{equation*}
  By construction, it is equipped with a map into the internal-Hom
  \begin{equation*}
    \HId_\McU(\tMcU) \xrightarrow{\proj_1} \pi_*(\tMcU \times_\McU \tMcU) \cong [\tMcU,\tMcU]_\McU
  \end{equation*}
\end{construction}
\begin{corollary}\label{cor:gen-auto-syn}
  %
  Assume $\tMcU \to \McU$ has a pre-$\Id$-type structure
  $\Id \colon \tMcU \times_\McU \tMcU \to \McU$.

  Any object of the isomorphism class of $\HId_\McU^\Id(\tMcU)$ represents the
  presheaf of propositional identities.
  \begin{lrbox}{\tType}\begin{tikzcd}[column sep=small]
      \Gamma.A \ar[r, "{\ceil{t}}"] \ar[d, "{\var_A}"'] & \tMcU \ar[d] \\
      \tMcU \ar[r] & \McU
    \end{tikzcd}
  \end{lrbox}
  \newsavebox{\hEndpt}
  \begin{lrbox}{\hEndpt}
    \begin{tikzcd}[]
      && \tMcU \\
      {\Gamma.A} & {\tMcU \times_{\McU} \tMcU} & \McU
      \arrow[from=1-3, to=2-3]
      \arrow["{\ceil{H}}", from=2-1, to=1-3]
      \arrow["{(\var_A, \ceil{t})}"', from=2-1, to=2-2]
      \arrow["\Id"', from=2-2, to=2-3]
    \end{tikzcd}
  \end{lrbox}
  \begin{align*}\small
    \left(
    \coprod_{\Gamma \xrightarrow{A} \McU}
    \left\{
    (\Gamma.A \xrightarrow{\ceil{t}} \tMcU, \Gamma.A \xrightarrow{\ceil{H}} \tMcU)
    ~\middle|~
    \usebox{\tType} \text{ and } \usebox{\hEndpt}
    \right\}
    \right)_\Gamma
  \end{align*}
  and the map $\HId_\McU^\Id(\tMcU) \to \pi_*(\tMcU \times \tMcU)$ represents the
  natural transformation whose components are
  $(A, \ceil{t}, \ceil{H}) \mapsto (A, \ceil{t})$.
\end{corollary}
\begin{proof}
  This presheaf is the pullback of the presheaf natural transformation
  represented by $\pi_*(\tMcU \times \tMcU) \to \pi_*(\tMcU \times \McU)$ along
  the map of presheaves from \Cref{lem:gen-auto-id}.
  By continuity of the Yoneda embedding and \Cref{lem:gen-auto-id},
  representability follows.
\end{proof}

A more semantic description of the object $\HId_\McU^\Id(\tMcU)$ is given as follows.
\begin{proposition}\label{prop:gen-auto}
  %
  Assume $\tMcU \to \McU$ has a pre-$\Id$-type structure
  $\Id \colon \tMcU \times_\McU \tMcU \to \McU$.

  Then any object in the isomorphism class of $\HId_\McU^\Id(\tMcU)$ represents
  the presheaf
  \begin{equation*}
    \left(
      \coprod_{\Gamma \xrightarrow{A} \McU}
      \left\{
        \begin{pmatrix}
          t \in \sfrac{\bC}{\Gamma}(\Gamma.A,\Gamma.A) \\
          H \in \sfrac{\bC}{\Gamma.A}(\Gamma.A, \Id_\Gamma(\Gamma.A))
        \end{pmatrix}
        ~\middle|~
        \begin{tikzcd}[cramped, column sep=small]
          {\Gamma.A} && {\Id_\Gamma(\Gamma.A)} \\
          & {\Gamma.A \times_\Gamma \Gamma.A}
          \arrow["H", from=1-1, to=1-3]
          \arrow["{(t, \id)}"{description}, from=1-1, to=2-2]
          \arrow["{\ev_\partial}"{description}, from=1-3, to=2-2]
        \end{tikzcd}
      \right\}
    \right)_\Gamma
  \end{equation*}
  and the map $\HId_\McU^\Id(\tMcU) \to [\tMcU,\tMcU]_\McU$ represents the map
  into
  $\left( \coprod_{\Gamma \xrightarrow{A} \tMcU}
    \sfrac{\bC}{\Gamma}(\Gamma.A,\Gamma.A) \right)_\Gamma$ with components
  $(t,H) \mapsto t$.
\end{proposition}
\begin{proof}
  We first show that the presheaf above and the presheaf of propositional
  identities of \Cref{cor:gen-auto-syn} are isomorphic.
  For each fixed $\ceil{A} \colon \Gamma \to \tMcU$ the correspondence between
  maps $\ceil{t} \colon \Gamma.A \to \tMcU$ and maps
  $t \colon \Gamma.A \to \Gamma.A$ over $\Gamma$ is by taking
  $\ceil{t} \coloneqq \var \cdot t \colon \Gamma.A \to \Gamma.A \to \tMcU$ and
  $t \coloneqq (\pi_A, \ceil{t}) \colon \Gamma.A \to \Gamma.A \cong \Gamma
  \times_{\McU} \tMcU$.
  The correspondence between $H \colon \Gamma.A \to \Id_\Gamma(\Gamma.A)$ and
  $\ceil{H} \colon \Gamma.A \to \tMcU$ is by taking
  $\ceil{H} \coloneqq \widetilde{\Id} \cdot \Id_{\ceil{A}}(\var) \cdot H$ and
  $H \coloneqq ((\id,t), \ceil{H})$. 
  %
  %
  \begin{equation*}\small
    \begin{tikzcd}[cramped, row sep=small, column sep=small]
      & {\Id_{\Gamma}(\Gamma.A)} && {\Id_\McU(\tMcU)} && \tMcU \\
      {\Gamma.A} \\
      && {\Gamma.A} && \tMcU \\
      & {\Gamma.A \times_\Gamma \Gamma.A} && {\tMcU \times_\McU \tMcU} && \McU \\
      && \Gamma && \McU \\
      & {\Gamma.A} && \tMcU
      \arrow["{\Id_{\ceil{A}}(\var)}", from=1-2, to=1-4]
      \arrow["{\widetilde{\Id}}", from=1-4, to=1-6]
      \arrow[from=1-6, to=4-6]
      \arrow["H"{description}, from=2-1, to=1-2]
      \arrow["t"{description, pos=0.7}, from=2-1, to=3-3]
      \arrow["{\ceil{t}}"{description}, curve={height=-12pt}, from=2-1, to=3-5]
      \arrow["{(\id, t)}"{description}, from=2-1, to=4-2]
      \arrow[from=2-1, to=5-3]
      \arrow["="', curve={height=12pt}, from=2-1, to=6-2]
      \arrow["\var"{description,pos=0.25}, from=3-3, to=3-5]
      \arrow[from=3-3, to=5-3]
      \arrow[from=3-5, to=5-5]
      \arrow[from=4-2, to=3-3]
      \arrow[from=4-2, to=6-2]
      \arrow[from=4-4, to=3-5]
      \arrow["{\ceil{A}}"{description, pos=0.25}, from=5-3, to=5-5]
      \arrow[from=6-2, to=5-3]
      \arrow["\var"{description,pos=0.5}, from=6-2, to=6-4]
      \arrow[from=6-4, to=5-5]
      \arrow[crossing over, from=1-2, to=4-2]
      \arrow[crossing over, from=1-4, to=4-4]
      \arrow["{\var \times_{\ceil{A}} \var}"{description}, crossing over, from=4-2, to=4-4]
      \arrow["\Id"{description, pos=0.3}, crossing over, from=4-4, to=4-6]
      \arrow[crossing over, from=4-4, to=6-4]
      \arrow["{\ceil{H}}"{description}, crossing over, from=2-1, to=1-6]
    \end{tikzcd}
  \end{equation*}

  The above proof also allows one to conclude that by representability and
  \Cref{cor:gen-auto-syn} the map $\HId_\McU^\Id(\tMcU) \to [\tMcU,\tMcU]_\McU$,
  when passing into the presheaf category, sends
  $(t,H) \mapsto (\ceil{t},\ceil{H}) \mapsto \ceil{t} \mapsto t$, as claimed.
\end{proof}

\subsection{Generic Object of Homotopy Isomorphisms}
Now we construct the generic object representing maps between fibrant objects in
each slice equipped separate homotopy inverses for each side.
The idea is to first internalise the composition operation then take a pullback
of the map $\HId_\McU^\Id(\tMcU) \to [\tMcU,\tMcU]_{\tMcU}$.

\begin{construction}\label{constr:lrhinv}
  Assume $\tMcU \to \McU$ is equipped with a pre-$\Id$-type structure
  $\Id \colon \tMcU \times_\McU \tMcU \to \McU$.
  We construct the isomorphism classes of objects $\LHInv_\McU^\Id(\tMcU)$ and
  $\RHInv_\McU^\Id(\tMcU)$ representing maps between objects in slice categories
  spanned by pullbacks of $\tMcU \to \McU$ equipped with a single-sided homotopy
  inverse.

  First, denote by $\CompPair_\McU(\tMcU)$ the following pullback consisting of
  pairs of composable maps
  \begin{equation*}
    \begin{tikzcd}[cramped]
      {\CompPair_\McU(\tMcU)} & {[\tMcU \times \McU, \McU \times \tMcU]_{\McU \times \McU}} \\
      {[\McU \times \tMcU, \tMcU \times \McU]_{\McU \times \McU}} & {\McU \times \McU}
      \arrow[from=1-1, to=1-2]
      \arrow[from=1-1, to=2-1]
      \arrow["\lrcorner"{anchor=center, pos=0.05, scale=1.5, rotate=0}, draw=none, from=1-1, to=2-2]
      \arrow[from=1-2, to=2-2]
      \arrow[from=2-1, to=2-2]
    \end{tikzcd}
  \end{equation*}
  so that it is equipped with two composition maps
  \begin{align*}
    \CompPair_\McU(\tMcU) \to
    [\tMcU, \tMcU]_{\McU} \times \McU
    &&
    \CompPair_\McU(\tMcU) \to
    \McU \times [\tMcU, \tMcU]_{\McU}
  \end{align*}

  Assume $\tMcU \to \McU$ has an $\Id$-type structure.
  The generic isomorphism classes of objects $\LHInv_\McU^\Id(\tMcU)$and
  $\RHInv_\McU^\Id(\tMcU)$ representing left and right homotopically invertible
  maps are given by the pullbacks as follows.
  \begin{center}
    \begin{minipage}{0.45\linewidth}
      \begin{equation*}\small
        \begin{tikzcd}[cramped, row sep=small, column sep=small]
          {\LHInv_\McU^\Id(\tMcU)} & {\McU \times \HId_\McU^\Id(\tMcU)} \\
          {\CompPair_\McU(\tMcU)} & {\McU \times [\tMcU, \tMcU]_\McU} \\
          {[\tMcU \times \McU, \McU \times \tMcU]_{\McU \times \McU}}
          \arrow[from=1-1, to=1-2]
          \arrow[from=1-1, to=2-1]
          \arrow["\lrcorner"{anchor=center, pos=0.15, scale=1.5}, draw=none, from=1-1, to=2-2]
          \arrow[from=1-2, to=2-2]
          \arrow[from=2-1, to=2-2]
          \arrow[from=2-1, to=3-1]
        \end{tikzcd}
      \end{equation*}
    \end{minipage}
    \begin{minipage}{0.45\linewidth}
      \begin{equation*}\small
        \begin{tikzcd}[cramped, row sep=small, column sep=small]
          {\RHInv_\McU^\Id(\tMcU)} & {\HId_\McU^\Id(\tMcU) \times \McU} \\
          {\CompPair_\McU(\tMcU)} & {[\tMcU, \tMcU]_\McU \times \McU} \\
          {[\tMcU \times \McU, \McU \times \tMcU]_{\McU \times \McU}}
          \arrow[from=1-1, to=1-2]
          \arrow[from=1-1, to=2-1]
          \arrow["\lrcorner"{anchor=center, pos=0.15, scale=1.5}, draw=none, from=1-1, to=2-2]
          \arrow[from=1-2, to=2-2]
          \arrow[from=2-1, to=2-2]
          \arrow[from=2-1, to=3-1]
        \end{tikzcd}
      \end{equation*}
    \end{minipage}
  \end{center}
\end{construction}

We now show that $\LHInv_\McU^\Id(\tMcU)$ and $\RHInv_\McU^\Id(\tMcU)$ and
indeed represent the correct presheaves.
\begin{lemma}\label{lem:lrhinv-rep}
  Assume $\tMcU \to \McU$ has a pre-$\Id$-type structure
  $\Id \colon \tMcU \times_\McU \tMcU \to \McU$.

  \begin{enumerate}
    \item The map
    $\RHInv_\McU^\Id(\tMcU) \to [\tMcU\times\McU, \McU\times\tMcU]_{\McU\times\McU}$
    represents the map of presheaves
    \newsavebox{\HRslice}
    \begin{lrbox}{\HRslice}
      \begin{tikzcd}[cramped, column sep=small]
        {\Gamma.A} && {\Id_\Gamma(\Gamma.A)} \\
        & {\Gamma.A \times_\Gamma \Gamma.A}
        \arrow["H_s"{description}, from=1-1, to=1-3]
        \arrow["{(fs, \id)}"{description}, from=1-1, to=2-2]
        \arrow["{\ev_\partial}"{description}, from=1-3, to=2-2]
      \end{tikzcd}
    \end{lrbox}
    \begin{equation*}
      \begin{tikzcd}
        {\displaystyle\left(
            \coprod_{A,B \colon \Gamma \rightrightarrows \McU}
            \left\{
              \begin{pmatrix}
                f \in \sfrac{\bC}{\Gamma}(\Gamma.A,\Gamma.B) \\
                s \in \sfrac{\bC}{\Gamma}(\Gamma.B,\Gamma.A) \\
                H_s \in \sfrac{\bC}{\Gamma}(\Gamma, \Id_\Gamma(\Gamma.A))
              \end{pmatrix}
              ~\middle|~
              \usebox{\HRslice}
            \right\}
          \right)_{\Gamma \in \bC}}
        \ar[d, "{(A,B,f,g,H) \mapsto (A,B,f)}"]
        \\
        {\displaystyle\left(
            \coprod_{A,B \colon \Gamma \rightrightarrows \McU}
            \sfrac{\bC}{\Gamma}(\Gamma.A,\Gamma.B)
          \right)_{\Gamma \in \bC}}
      \end{tikzcd}
    \end{equation*}
    \item Likewise, the map
    $\LHInv_\McU^\Id(\tMcU) \to [\tMcU\times\McU, \McU\times\tMcU]_{\McU\times\McU}$
    represents the map of presheaves
    \newsavebox{\HLslice}
    \begin{lrbox}{\HLslice}
      \begin{tikzcd}[cramped, column sep=small]
        {\Gamma.B} && {\Id_\Gamma(\Gamma.B)} \\
        & {\Gamma.B \times_\Gamma \Gamma.B}
        \arrow["{H_r}"{description}, from=1-1, to=1-3]
        \arrow["{(rf, \id)}"{description}, from=1-1, to=2-2]
        \arrow["{\ev_\partial}"{description}, from=1-3, to=2-2]
      \end{tikzcd}
    \end{lrbox}
    \begin{equation*}
      \begin{tikzcd}
        {\displaystyle\left(
            \coprod_{A,B \colon \Gamma \rightrightarrows \McU}
            \left\{
              \begin{pmatrix}
                f \in \sfrac{\bC}{\Gamma}(\Gamma.A,\Gamma.B) \\
                r \in \sfrac{\bC}{\Gamma}(\Gamma.B,\Gamma.A) \\
                H_r \in \sfrac{\bC}{\Gamma}(\Gamma, \Id_\Gamma(\Gamma.B))
              \end{pmatrix}
              ~\middle|~
              \usebox{\HLslice}
            \right\}
          \right)_{\Gamma \in \bC}}
        \ar[d, "{(A,B,f,g,H) \mapsto (A,B,f)}"]
        \\
        {\displaystyle\left(
            \coprod_{A,B \colon \Gamma \rightrightarrows \McU}
            \sfrac{\bC}{\Gamma}(\Gamma.A,\Gamma.B)
          \right)_{\Gamma \in \bC}}
      \end{tikzcd}
    \end{equation*}
  \end{enumerate}
\end{lemma}
\begin{proof}
  By the definition of the internal composition map and \Cref{prop:gen-auto}.
\end{proof}

Therefore, the following construction indeed yields the general object
representing fibred $\Id$-type homotopy equivalences.
\begin{construction}\label{constr:hiso}
  Assume $\tMcU \to \McU$ has a pre-$\Id$-type structure
  $\Id \colon \tMcU \times_\McU \tMcU \to \McU$.
  Construct the isomorphism class $\HIso_\McU^\Id(\tMcU)$ and the maps
  $\src,\dest \colon \HIso_\McU^\Id(\tMcU) \rightrightarrows \McU$ as follows.
  \begin{equation*}
    \begin{tikzcd}[cramped, row sep=small, column sep=small]
      & {\HIso_\McU^\Id(\tMcU)} \\
      {\LHInv_\McU^\Id(\tMcU)} && {\RHInv_\McU^\Id(\tMcU)} \\
      & {[\tMcU \times \McU, \McU \times \tMcU]_{\McU \times \McU}} \\
      & {\McU \times \McU} \\
      \McU && \McU
      \arrow[from=1-2, to=2-1]
      \arrow[from=1-2, to=2-3]
      \arrow["\lrcorner"{anchor=center, pos=0.15, scale=1.5, rotate=-45}, draw=none, from=1-2, to=3-2]
      \arrow[from=2-1, to=3-2]
      \arrow[from=2-3, to=3-2]
      \arrow[from=3-2, to=4-2]
      \arrow["{\proj_1}", from=4-2, to=5-1]
      \arrow["{\proj_2}"', from=4-2, to=5-3]
      \arrow["\src"'{pos=0.8}, from=1-2, to=5-1, crossing over, bend right=10]
      \arrow["\dest"{pos=0.8}, from=1-2, to=5-3, crossing over, bend left=10]
    \end{tikzcd}
  \end{equation*}
\end{construction}

\begin{lemma}\label{lem:hinv-rep}
  %
  Assume $\tMcU \to \McU$ has a pre-$\Id$-type structure.
  Then, the maps
  \begin{equation*}
    \src,\dest \colon \HIso_\McU^\Id(\tMcU) \rightrightarrows \McU
  \end{equation*}
  represent the following maps of presheaves.
  \begin{lrbox}{\HRslice}\scriptsize
    \begin{tikzcd}[cramped, column sep=tiny]
      {\Gamma.A} && {\Id_\Gamma(\Gamma.A)} \\
      & {\Gamma.A \times_\Gamma \Gamma.A}
      \arrow["H_s"{description}, from=1-1, to=1-3]
      \arrow["{(fs, \id)}"{description}, from=1-1, to=2-2]
      \arrow["{\ev_\partial}"{description}, from=1-3, to=2-2]
    \end{tikzcd}
  \end{lrbox}
  \begin{lrbox}{\HLslice}\scriptsize
    \begin{tikzcd}[cramped, column sep=tiny]
      {\Gamma.B} && {\Id_\Gamma(\Gamma.B)} \\
      & {\Gamma.B \times_\Gamma \Gamma.B}
      \arrow["{H_r}"{description}, from=1-1, to=1-3]
      \arrow["{(rf, \id)}"{description}, from=1-1, to=2-2]
      \arrow["{\ev_\partial}"{description}, from=1-3, to=2-2]
    \end{tikzcd}
  \end{lrbox}
  \begin{equation*}\scriptsize
    \begin{tikzcd}
      {\displaystyle
        \left(
          \coprod_{\substack{A \colon \Gamma \to \McU \\ B \colon \Gamma \to \McU}}
          \coprod_{\substack{f \in \sfrac{\bC}{\Gamma}(\Gamma.A,\Gamma.B) \\ s,r \in \sfrac{\bC}{\Gamma}(\Gamma.B,\Gamma.A)}}
          \left\{\scriptsize
            \begin{pmatrix}
              H_s \in \sfrac{\bC}{\Gamma}(\Gamma.A, \Id_\Gamma(\Gamma.A)) \\
              H_r \in \sfrac{\bC}{\Gamma}(\Gamma.B, \Id_\Gamma(\Gamma.B)) \\
            \end{pmatrix}
            ~\middle|~
            \usebox{\HRslice} \text{\scriptsize and}
            \usebox{\HLslice}
          \right\}
        \right)_{\Gamma \in \bC}
      }
      \ar[d, shift left=2, "{(A,B,f,s,r,H_s,H_r) \mapsto B}"]
      \ar[d, shift left=-2, "{(A,B,f,s,r,H_s,H_r) \mapsto A}"']
      \\
      {\scriptsize (\bC(\Gamma,\McU))_{\Gamma \in \bC}}
    \end{tikzcd}
  \end{equation*}
\end{lemma}
\begin{proof}
  By \Cref{lem:lrhinv-rep}.
\end{proof}

This description of $\HIso_\McU^\Id(\tMcU)$ by representability also allows one to
construct the map of trivial homotopy equivalences, in the following sense.
\begin{construction}\label{constr:trv-htpy}
  Assume $\tMcU \to \McU$ has a pre-$\Id$-type structure
  $\Id \colon \tMcU \times_\McU \tMcU \to \McU$.
  The map
  \begin{equation*}
    \trv \colon \McU \hookrightarrow \HIso_\McU^\Id(\tMcU)
  \end{equation*}
  serving as a common section
  $\src,\dest\colon \HIso_\McU^\Id(\tMcU) \rightrightarrows \McU$ is constructed, by
  \Cref{lem:hinv-rep}, as the unique map representing the map of presheaves
  \begin{lrbox}{\HRslice}\scriptsize
    \begin{tikzcd}[cramped, column sep=tiny]
      {\Gamma.A} && {\Id_\Gamma(\Gamma.A)} \\
      & {\Gamma.A \times_\Gamma \Gamma.A}
      \arrow["H"{description}, from=1-1, to=1-3]
      \arrow["{(fs, \id)}"{description}, from=1-1, to=2-2]
      \arrow["{\ev_\partial}"{description}, from=1-3, to=2-2]
    \end{tikzcd}
  \end{lrbox}
  \begin{lrbox}{\HLslice}\scriptsize
    \begin{tikzcd}[cramped, column sep=tiny]
      {\Gamma.B} && {\Id_\Gamma(\Gamma.B)} \\
      & {\Gamma.B \times_\Gamma \Gamma.B}
      \arrow["{H_r}"{description}, from=1-1, to=1-3]
      \arrow["{(rf, \id)}"{description}, from=1-1, to=2-2]
      \arrow["{\ev_\partial}"{description}, from=1-3, to=2-2]
    \end{tikzcd}
  \end{lrbox}
  \begin{equation*}\scriptsize
    \begin{tikzcd}
      {\displaystyle
        \left(
          \coprod_{\substack{A \colon \Gamma \to \McU \\ B \colon \Gamma \to \McU}}
          \coprod_{\substack{f \in \sfrac{\bC}{\Gamma}(\Gamma.A,\Gamma.B) \\ s,r \in \sfrac{\bC}{\Gamma}(\Gamma.B,\Gamma.A)}}
          \left\{\scriptsize
            \begin{pmatrix}
              H_s \in \sfrac{\bC}{\Gamma}(\Gamma.A, \Id_\Gamma(\Gamma.A)) \\
              H_r \in \sfrac{\bC}{\Gamma}(\Gamma.B, \Id_\Gamma(\Gamma.B)) \\
            \end{pmatrix}
            ~\middle|~
            \usebox{\HRslice} \text{\scriptsize and}
            \usebox{\HLslice}
          \right\}
        \right)_{\Gamma \in \bC}
      }
      \\
      {\scriptsize (\bC(\Gamma,\McU))_{\Gamma \in \bC}}
      \ar[u, "{A \mapsto (A,A,\id_{\Gamma.A},\id_{\Gamma.A},\id_{\Gamma.A},\refl_A,\refl_A)}", hook]
    \end{tikzcd}
  \end{equation*}
\end{construction}


\section{Univalence}\label{sec:univalence}
Having constructed the generic object representing fibred $\Id$-type fibred
homotopy equivalences in \Cref{constr:hiso} of an universal map
$\tMcU \to \McU$, we can now begin to state two forms of the univalence axiom:
one directly corresponding to a lifting condition that turns out to be stronger
than univalence from \cite{hottbook} and another corresponding to traditional
univalence from the HoTT book \cite{hottbook}.

Roughly, the stronger version says that given an universal map $\tMcU \to \McU$,
the generic object of $\Id$-type fibred homotopy equivalences serves as a very
good path object.
To give a uniform treatment of book univalence, however, we need some
preparation.

\begin{definition}
  A \emph{lower-half lift} of a map $K \to L$ against a map $E \to B$ for a
  specified lifting is a filler only making the bottom triangle commute.
  \begin{equation*}
    \begin{tikzcd}[cramped]
      K & E \\
      L & B
      \arrow[dashed, from=1-1, to=1-2]
      \arrow[from=1-1, to=2-1]
      \arrow[from=1-2, to=2-2]
      \arrow[""{name=0, anchor=center, inner sep=0}, dotted, from=2-1, to=1-2]
      \arrow[dashed, from=2-1, to=2-2]
      \arrow["{?}"{description}, draw=none, from=1-1, to=0]
    \end{tikzcd}
  \end{equation*}
\end{definition}

\begin{remark}
  Note that the map $K \to L$ in the definition of a lower-half lift \emph{does}
  serve a purpose and it is not the same as saying that all maps $L \to B$
  factors via $E \to B$.
  This is because one cannot simply replace $K \to L$ with the map $0 \to L$ from
  the initial object (or the empty type), if one exists.
  Instead, the lower-half lift property specifies the factorisation property via
  $E \to B$ only for those $L \to B$ such that its restriction $K \to L \to B$
  factors via $E \to B$ (as opposed to for arbitrary $L \to B$ as given by the
  $0 \to L$ case).
  Therefore, if $E \to B$ is such that $0 \to L$ has the lower-half left lift
  property then this would imply that $K \to L$ also has the lower-half left
  lifting property, but the converse is not necessarily true.
\end{remark}

Similar to \cite[Definition 1.4]{struct-lift}, we can also have a uniformity
condition on the choice of lower-half lifts.
\begin{definition}\label{def:lower-half-lift}
  Fix maps $i \colon K \to L$ and $p \colon E \to B$.

  A \emph{family of lower-half lifts} is an association taking each object
  $X \in \bC$ and lifting problem $(u, v)$ of $X \times i$ against $p$ to a
  lower-half lift $F_X(u,v)$ as below.
  \begin{equation*}
    \begin{tikzcd}[cramped]
      {X \times K} & E \\
      {X \times L} & B
      \arrow["u", color=red0, dashed, from=1-1, to=1-2]
      \arrow[color=darkblue0, from=1-1, to=2-1]
      \arrow[color=darkblue0, from=1-2, to=2-2]
      \arrow["{?}"{description}, draw=none, from=0, to=1-1]
      \arrow[""{name=0, anchor=center, inner sep=0}, "{F_X(u,v)}"{description}, color=yellow0, dotted, from=2-1, to=1-2]
      \arrow["v"', color=red0, dashed, from=2-1, to=2-2]
    \end{tikzcd}
  \end{equation*}

  This family $F$ is said to be \emph{uniform} when one has one has
  $F_Y(u \cdot (t \times U), v' \cdot (t \times V')) = F_X(u,v') \cdot (t \times
  V)$ for any $t \colon Y \to X$, as below.
  \begin{equation*}
    \begin{tikzcd}[cramped]
      {Y \times K} & {X \times K} & E \\
      {Y \times L} & {X \times L} & B
      \arrow["{t \times K}", from=1-1, to=1-2]
      \arrow[color=darkblue0, from=1-1, to=2-1]
      \arrow["u", color=red0, dashed, from=1-2, to=1-3]
      \arrow[color=darkblue0, from=1-2, to=2-2]
      \arrow[color=darkblue0, from=1-3, to=2-3]
      \arrow[draw=none, from=2-1, to=1-2]
      \arrow[color=yellow0, dotted, from=2-1, to=1-3]
      \arrow["{t \times L}"', from=2-1, to=2-2]
      \arrow[color=yellow0, dotted, from=2-2, to=1-3]
      \arrow["v"', color=red0, dashed, from=2-2, to=2-3]
    \end{tikzcd}
  \end{equation*}

  The set of all such uniform family of lower-half lifts, which
  we refer to as \emph{lower-half lifting structures} is denoted by
  \begin{equation*}
    {\begin{tikzcd}[cramped, column sep=small, row sep=small]
        K \ar[d] \\ L
      \end{tikzcd}}
    \lowerslash
    {\begin{tikzcd}[cramped, column sep=small, row sep=small]
        E \ar[d] \\ B
      \end{tikzcd}}
  \end{equation*}

  When working in the slice over an object $C$, the set of all such uniform
  family of lower-half lifts in the slice $\sfrac{\bC}{C}$ is denoted by
  \begin{equation*}
    {\begin{tikzcd}[cramped, column sep=small, row sep=small]
        K \ar[d] \\ L
      \end{tikzcd}}
    \fraclowerslash{C}
    {\begin{tikzcd}[cramped, column sep=small, row sep=small]
        E \ar[d] \\ B
      \end{tikzcd}}
  \end{equation*}
\end{definition}

\begin{definition}\label{def:axm-univalence}
  Let $\pi_0 \colon \tMcU_0 \to \McU_0$ and $\pi \colon \tMcU \to \McU$ be two
  universal maps such that $\pi_0$ has a pre-$\Id$-type structure
  $\Id_0 \colon \tMcU_0 \times_{\McU_0} \tMcU_0 \to \McU_0$.

  A \emph{pointed} (respectively, \emph{book}) \emph{$\pi$-univalence} structure
  on the (universe, pre-$\Id$-type)-pair $(\pi_0,\Id_0)$ is a $\pi$-fibrancy
  structure on the map $(\src, \dest)$ along with a choice of uniform
  (respectively, lower-half) lifts $\PtdUA$ (respectively $\UA$) of $\trv$ on
  the against $\pi$, in the sense of \cite[Definition 1.4]{struct-lift}
  (respectively, \Cref{def:lower-half-lift})
  \begin{center}
    \begin{minipage}{0.45\linewidth}
      \begin{equation*}
        \PtdUA \in
        \left(
          {\begin{tikzcd}[cramped]
              \tMcU_0 \ar[d, hook, "{\trv}"'] \\ \HIso_{\tMcU_0}^{\Id_0}(\tMcU_0)
            \end{tikzcd}}
          \squareslash
          {\begin{tikzcd}[cramped]
              \tMcU \ar[d, two heads] \\ \McU
            \end{tikzcd}}\right)
      \end{equation*}
    \end{minipage}
    \begin{minipage}{0.45\linewidth}
      \begin{equation*}
        \UA \in
        \left(
          {\begin{tikzcd}[cramped]
              \tMcU_0 \ar[d, hook, "{\trv}"'] \\ \HIso_{\tMcU_0}^{\Id_0}(\tMcU_0)
            \end{tikzcd}}
          \lowerslash
          {\begin{tikzcd}[cramped]
              \tMcU \ar[d, two heads] \\ \McU
            \end{tikzcd}}\right)
      \end{equation*}
    \end{minipage}
  \end{center}
  where $(\src, \dest)$ and $\trv$ are from the diagonal factorisation in
  \Cref{constr:hiso,constr:trv-htpy} via $\HIso_{\McU_0}^{\Id_0}(\tMcU_0)$
  \begin{equation*}
    \begin{tikzcd}[cramped]
      {\McU_0} & {\HIso_{\McU_0}^{\Id_0}(\tMcU_0)} &[3em] {\McU_0 \times \McU_0}
      \arrow["{\trv}", hook, from=1-1, to=1-2]
      \arrow["{(\src, \dest)}", two heads, from=1-2, to=1-3]
    \end{tikzcd}
  \end{equation*}
  The (universe, pre-$\Id$-type)-pair $(\pi_0,\Id_0)$ is \emph{pointed}
  (respectively \emph{book}) \emph{$\pi$-univalent} when it admits a pointed
  (respectively book) $\pi$-univalence structure.
\end{definition}
\begin{remark}
  In the above definitions, we have just required $\pi$-fibrancy, but to remain
  syntactically faithful, one should work instead with \emph{canonical}
  $\pi$-fibrancy.
  But because $\HIso_{\tMcU_0}^{\Id_0}(\tMcU_0)$ was constructed by way of
  representability, by taking canonical $\pi$-fibrant replacements, we may
  equivalently express univalence as a \emph{choice of an object} in the
  isomorphism class of $\HIso_{\McU_0}^{\Id_0}(\tMcU_0)$ such that the
  associated
  $(\src,\dest) \colon \HIso_{\McU_0}^{\Id_0}(\tMcU_0) \twoheadrightarrow
  \tMcU_0 \times \tMcU_0$ is \emph{canonically} $\pi$-fibrant.
  Hence, the above definition does indeed reflect syntax.
\end{remark}

\subsection{Fibrancy of the Endpoint Maps}
When working with models of type theory, in many cases, checking for univalence
amounts to checking the trivial cofibrancy of the $\trv$ map, since the $\Pi$-,
$\Sigma$-type structures gives automatically fibrancy of the $(\src,\dest)$ map.

\begin{proposition}\label{prop:endpt-fib}
  Suppose $\pi \colon \tMcU \to \McU$ is a universal map equipped with a pre-$\Id$-type
  structure $\Id \colon \tMcU \times_\McU \tMcU \to \McU$.
  When $\pi$ is also equipped with $\Pi$-, $\Sigma$-type structures, the map
  \begin{equation*}
    (\src,\dest) \colon \HIso_\McU^\Id(\tMcU) \twoheadrightarrow \McU \times \McU
  \end{equation*}
  is $\pi$-fibrant.
\end{proposition}
\begin{proof}
  In the context of \Cref{constr:gen-auto}, by the definition of $\Pi$-type
  structures, the map
  \begin{equation*}
    \HId_\McU(\tMcU) \twoheadrightarrow \pi_*(\tMcU \times_\McU \tMcU)
  \end{equation*}
  is $\pi$-fibrant.

  As a result, in the context of \Cref{constr:lrhinv}
  the maps $\LHInv_\McU^\Id(\tMcU) \to \CompPair_\McU(\tMcU)$ and
  $\RHInv_\McU^\Id(\tMcU) \to \CompPair_\McU(\tMcU)$ from
  are $\pi$-fibrant.
  Moreover, by \Cref{cor:local-hom-fib}, the objects
  $[\tMcU \times \McU, \McU \times \tMcU]_{\tMcU \times \McU}$ and
  $[\McU \times \tMcU, \tMcU \times \McU]_{\tMcU \times \McU}$ are also
  $\pi$-fibrant object in the slice over $\McU \times \McU$, so the maps
  \begin{align*}
    \begin{array}{c}
      \CompPair_\McU(\tMcU) \twoheadrightarrow [\tMcU \times \McU, \McU \times \tMcU]_{\tMcU \times \McU} \\
      \CompPair_\McU(\tMcU) \twoheadrightarrow [\McU \times \tMcU, \tMcU \times \McU]_{\tMcU \times \McU}
    \end{array}
  \end{align*}
  are $\pi$-fibrant.
  Hence, by the $\Sigma$-type structure and \Cref{prop:Pi-Sigma-generic}, the composites
  \begin{align*}
    \begin{array}{c}
      \LHInv_\McU^\Id(\tMcU) \twoheadrightarrow \CompPair_\McU(\tMcU) \twoheadrightarrow [\tMcU \times \McU, \McU \times \tMcU] \\
      \RHInv_\McU^\Id(\tMcU) \twoheadrightarrow \CompPair_\McU(\tMcU) \twoheadrightarrow [\tMcU \times \McU, \McU \times \tMcU]
    \end{array}
  \end{align*}
  are $\pi$-fibrant.

  Therefore, in the context of \Cref{constr:hiso}, by the above reasoning and
  the $\Sigma$-type structure, the composite map
  \begin{equation*}
    \HIso_\McU^\Id(\tMcU) \twoheadrightarrow [\tMcU \times \McU, \McU \times \tMcU]_{\McU\times\McU}
  \end{equation*}
  is $\pi$-fibrant.
  And again by \Cref{cor:local-hom-fib}, the map
  $\tMcU \times \McU, \McU \times \tMcU]_{\tMcU \times \McU} \twoheadrightarrow \McU \times \McU$ is $\pi$-fibrant.
  Hence, by the $\Sigma$-type structure and \Cref{prop:Pi-Sigma-generic}, the map
  \begin{equation*}
    (\src,\dest) \colon \HIso_\McU^\Id(\tMcU) \twoheadrightarrow \McU \times \McU
  \end{equation*}
  is $\pi$-fibrant.
\end{proof}

Therefore, when stating univalence for internal universes with $\Id$-, $\Pi$-,
$\Sigma$-structures, fibrancy comes for free.
\begin{corollary}\label{cor:endpt-int-uni-fib}
  Let $\pi \colon \tMcU \to \McU$ be a universal map with an internal universe
  $\pi_0 \colon \tMcU_0 \to \McU_0$ such that $\pi_0$ is equipped with
  pre-$\Id$-, $\Pi$-, $\Sigma$-structures respectively denoted
  $\Id_0,\Pi_0,\Sigma_0$.
  Then, the map
  \begin{equation*}
    (\src,\dest) \colon \HIso_{\McU_0}^{\Id_0}(\tMcU_0) \twoheadrightarrow \McU_0 \times \McU_0
  \end{equation*}
  is $\pi$-fibrant.
\end{corollary}
\begin{proof}
  By \Cref{prop:endpt-fib}, the map
  $(\src,\dest) \colon \HIso_{\McU_0}^{\Id_0}(\tMcU_0) \twoheadrightarrow
  \tMcU_0 \times \tMcU_0$ is $\pi_0$-fibrant.
  But $\pi_0$ is (canonically) $\pi$-fibrant, so the result follows.
\end{proof}

Furthermore, when discussing internal universes where the ambient external
universe is itself equipped with a $\Sigma$-type structure, the individual
$\src,\dest$ maps are fibrant.

\begin{corollary}\label{cor:src-dest-int-uni-fib}
  Let $\pi \colon \tMcU \to \McU$ be a universal map with an internal universe
  $\pi_0 \colon \tMcU_0 \to \McU_0$ such that $\pi_0$ is equipped with
  pre-$\Id$-, $\Pi$-, $\Sigma$-structures respectively denoted
  $\Id_0,\Pi_0,\Sigma_0$ and $\pi$ is equipped with a $\Sigma$-type structure.
  Then, the maps
  \begin{equation*}
    \src,\dest \colon \HIso_{\McU_0}^{\Id_0}(\tMcU_0) \twoheadrightarrow \McU_0
  \end{equation*}
  are both $\pi$-fibrant.
\end{corollary}
\begin{proof}
  By \Cref{cor:endpt-int-uni-fib}, the map
  $(\src,\dest) \colon \HIso_{\McU_0}^{\Id_0}(\tMcU_0) \twoheadrightarrow
  \McU_0 \times \McU_0$ is $\pi$-fibrant.
  Because $\tMcU_0 \twoheadrightarrow 1$ is (canonically) $\pi$-fibrant, so are
  the two projections
  $\proj_1,\proj_2 \colon \McU_0 \times \McU_0 \twoheadrightarrow \McU_0$.
  Referring to \Cref{constr:hiso}, the result then follows by
  \Cref{prop:Pi-Sigma-generic} on the $\Sigma$-structure.
\end{proof}


\subsection{Univalence via a Homotopy Retract Argument}
Primarily, people are interested in having an internal univalent universe.
In view of \Cref{cor:endpt-int-uni-fib}, all the work amounts to showing the
(lower-half) left lifting property of the map
$\trv \colon \tMcU_0 \to \HIso_{\McU_0}^{\Id_0}(\tMcU_0)$.
In this section, we further reduce the need to solve all (lower-half) lifting
problems to just one specific lifting problem.

\begin{definition}\label{def:const-res-htpy}
  Suppose $\pi \colon \tMcU \to \McU$ is an universal map equipped with an
  $\Id$-type structure $\Id \colon \tMcU \times_\McU \tMcU \to \McU$.

  Given $B \in \bC$ and two maps
  $f,g \colon X \rightrightarrows E \in \sfrac{\bC}{B}$ where
  $E \twoheadrightarrow B$ is $\pi$-fibrant, an $\Id$-homotopy between $f$ and
  $g$ \emph{constant along $i \colon U \to X$} such that
  $fi = gi \eqqcolon h \colon U \to X \rightrightarrows E$ is a solution $H$ to
  the following lifting problem.
  \begin{equation*}
    \begin{tikzcd}
      U & E & {\Id_B(E)} \\
      X && {E \times_B E}
      \arrow["h", from=1-1, to=1-2]
      \arrow["i"', from=1-1, to=2-1]
      \arrow["\refl", hook, from=1-2, to=1-3]
      \arrow[two heads, "{\ev_\partial}", from=1-3, to=2-3]
      \arrow["H"{description}, dashed, from=2-1, to=1-3]
      \arrow["{(f, g)}"', from=2-1, to=2-3]
    \end{tikzcd}
  \end{equation*}
\end{definition}

The reason we are interested in the above lifting problem is because left
lifting property against $\tMcU \to \McU$ transfer along homotopy retracts that
restrict to give the constant homotopy.
\begin{lemma}\label{lem:retract-lift}
  Suppose $\pi \colon \tMcU \to \McU$ has a pre-$\Id$-type structure $\Id \colon \tMcU \times_\McU \tMcU \to \McU$.

  Let there be the following data over an object $B \in \bC$ where
  $E \twoheadrightarrow B$ is a $\pi$-fibration
  \begin{equation*}
    \begin{tikzcd}[cramped]
      & U \\
      E & {E'} & E
      \arrow["i"', from=1-2, to=2-1]
      \arrow["{j}"{description}, from=1-2, to=2-2]
      \arrow["i", from=1-2, to=2-3]
      \arrow["s"', from=2-1, to=2-2]
      \arrow["r"', from=2-2, to=2-3]
    \end{tikzcd} \in \sfrac{\bC}{B}
  \end{equation*}
  where the composite $sr$ is $\Id$-homotopic to the identity by some $\Id$-type homotopy $H$ over $B$.
  \begin{enumerate}
    \item\label{itm:retract-lift-lower} Suppose the pre-$\Id$-type structure has a $\transport$-structure.
    Then, one has a map
    \begin{equation*}
      \left(
        {\begin{tikzcd}[cramped]
            U \ar[d, "{j}"'] \\ E'
          \end{tikzcd}}
        \fraclowerslash{B}
        {\begin{tikzcd}[cramped]
            B \times \tMcU \ar[d, two heads] \\ B \times \McU
          \end{tikzcd}}\right)
      \xrightarrow{\qquad}
      \left(
        {\begin{tikzcd}[cramped]
            U \ar[d, "{i}"'] \\ E
          \end{tikzcd}}
        \fraclowerslash{B}
        {\begin{tikzcd}[cramped]
            B \times \tMcU \ar[d, two heads] \\ B \times \McU
          \end{tikzcd}}\right)
    \end{equation*}
    \item\label{itm:retract-lift-strong} Suppose the pre-$\Id$-type structure is
    an $\Id$-type structure and $\pi$ has a $\Pi$-type structure.
    If the homotopy $H$ is constant when restricted along
    $i \colon U \to E$, then one has a map
    \begin{equation*}
      \left(
        {\begin{tikzcd}[cramped]
            U \ar[d, "{j}"'] \\ E'
          \end{tikzcd}}
        \fracsquareslash{B}
        {\begin{tikzcd}[cramped]
            B \times \tMcU \ar[d, two heads] \\ B \times \McU
          \end{tikzcd}}\right)
      \xrightarrow{\qquad}
      \left(
        {\begin{tikzcd}[cramped]
            U \ar[d, "{i}"'] \\ E
          \end{tikzcd}}
        \fracsquareslash{B}
        {\begin{tikzcd}[cramped]
            B \times \tMcU \ar[d, two heads] \\ B \times \McU
          \end{tikzcd}}\right)
    \end{equation*}
  \end{enumerate}
\end{lemma}
\begin{proof}
  We first examine the goals.
  By representability, for \Cref{itm:retract-lift-lower}, to show $i$ admits
  lower-half lifting structures against $\pi$ is to show that for all maps
  $f \colon C \to B$, if $p \colon X \twoheadrightarrow f^*E$ is a
  $\pi$-fibration and $x \colon f^*U \hookrightarrow (f^*i)^*X$ is a section of
  $(f^*i)^*p \colon (f^*i)^*X \twoheadrightarrow U$ then one can find some
  section $\overline{x} \colon f^*E \hookrightarrow X$ section of $p$, as
  depicted below.
  \begin{equation*}
    \begin{tikzcd}[cramped]
      {(f^*i)^*X} & X \\
      {f^*U} & {f^*E}
      \arrow[from=1-1, to=1-2]
      \arrow[shift left, two heads, from=1-1, to=2-1]
      \arrow["\lrcorner"{anchor=center, pos=0.15, scale=1.5}, draw=none, from=1-1, to=2-2]
      \arrow[shift left, two heads, from=1-2, to=2-2]
      \arrow["{x}", shift left, hook', from=2-1, to=1-1]
      \arrow["{f^*i}"', from=2-1, to=2-2]
      \arrow["{\overline{x}}", shift left, dashed, hook', from=2-2, to=1-2]
    \end{tikzcd} \in \sfrac{\bC}{C}
  \end{equation*}
  The uniformity property for the lift means that for any $g \colon D \to C$, one must additionally have
  \begin{equation*}
    \ol{g^*x} = g^*\ol{x}
  \end{equation*}
  For \Cref{itm:retract-lift-strong}, one additionally needs to ensure that the
  pullback of $\overline{x}$ along $f^*i$ is $(f^*i)^*\overline{x} = x$.

  Similarly, by representability, we have the respective extension properties for $j$.
  So we must show that if $j$ has the corresponding extension properties for
  cases \Cref{itm:retract-lift-lower} and \Cref{itm:retract-lift-strong} then so
  does $i$.

  To this end, fix some $f \colon C \to B$ and let there be a $\pi$-fibration
  $X \twoheadrightarrow f^*E$ along with a section
  $x \colon f^*U \hookrightarrow (f^*i)^*X$ of the pullback
  $(f^*i)^*p \colon (f^*i)^*X \twoheadrightarrow f^*U$.
  In both cases \Cref{itm:retract-lift-lower} and
  \Cref{itm:retract-lift-strong}, we are equipped with a homotopy $H$ from $sr$
  to the identity.
  Pulling back the homotopy $H$ then gives a homotopy $f^*H$ from $(f^*s)(f^*r)$
  to the identity.
  Using $\transport$, we can apply the following procedure.
  \begin{enumerate}
    \item We first pullback $X \twoheadrightarrow f^*E$ along the two endpoint
    evaluation maps $\ev_0,\ev_1 \colon \Id_C(f^*E) \twoheadrightarrow f^*E$ to obtain
    $\ev_0^*X, \ev_1^*X \twoheadrightarrow \Id_C(f^*E)$.
    \begin{equation*}\small
      \begin{tikzcd}[cramped, row sep=small, column sep=small]
        {\ev_1^*X} \\
        && X \\
        & {\ev_0^*X} \\
        {\Id_C(f^*E)} && f^*E
        \arrow[from=1-1, to=2-3]
        \arrow[two heads, from=1-1, to=4-1]
        \arrow[two heads, from=2-3, to=4-3]
        \arrow[from=3-2, to=2-3]
        \arrow[two heads, from=3-2, to=4-1]
        \arrow["{\ev_1}", shift left, from=4-1, to=4-3]
        \arrow["{\ev_0}"', shift right, from=4-1, to=4-3]
      \end{tikzcd}
    \end{equation*}
    \item By assumption in case \Cref{itm:retract-lift-lower} and
    \Cref{constr:transport} in case \Cref{itm:retract-lift-strong} , we have a
    transport map $\transport \colon \ev_0^*X \to \ev_1^*X$ over $\Id_C(f^*E)$.
    \begin{equation*}\small
      \begin{tikzcd}[cramped, row sep=small, column sep=small]
        {\ev_1^*X} \\
        && X \\
        & {\ev_0^*X} \\
        {\Id_C(f^*E)} && f^*E
        \arrow[from=1-1, to=2-3]
        \arrow[two heads, from=1-1, to=4-1]
        \arrow[two heads, from=2-3, to=4-3]
        \arrow["\transport"{description}, from=3-2, to=1-1]
        \arrow[from=3-2, to=2-3]
        \arrow[two heads, from=3-2, to=4-1]
        \arrow["{\ev_1}", shift left, from=4-1, to=4-3]
        \arrow["{\ev_0}"', shift right, from=4-1, to=4-3]
      \end{tikzcd}
    \end{equation*}
    \item Further pulling back the $\transport$ map along
    $f^*H \colon f^*E \to \Id_C(f^*E)$ gives
    \begin{equation*}
      (f^*r)^*(f^*s)^*X \xrightarrow{(f^*H)^*\transport} X
    \end{equation*}
    over $f^*E$.
    \begin{equation*}\small
      \begin{tikzcd}[cramped, row sep=small, column sep=small]
        X && {\ev_1^*X} \\
        &&&& X \\
        & {(f^*r)^*(f^*s)^*X} && {\ev_0^*X} \\
        f^*E && {\Id_C(f^*E)} && f^*E
        \arrow[from=1-1, to=1-3]
        \arrow[two heads, from=1-1, to=4-1]
        \arrow[from=1-3, to=2-5]
        \arrow[two heads, from=1-3, to=4-3]
        \arrow[two heads, from=2-5, to=4-5]
        \arrow["{(f^*H)^*\transport}"{description}, from=3-2, to=1-1]
        \arrow[two heads, from=3-2, to=4-1]
        \arrow["\transport"{description}, from=3-4, to=1-3]
        \arrow[from=3-4, to=2-5]
        \arrow[two heads, from=3-4, to=4-3]
        \arrow["{f^*H}"', from=4-1, to=4-3]
        \arrow["{\ev_1}", shift left, from=4-3, to=4-5]
        \arrow["{\ev_0}"', shift right, from=4-3, to=4-5]
        \arrow[crossing over, from=3-2, to=3-4]
      \end{tikzcd}
    \end{equation*}
    \item Again pulling back along $f^*i \colon f^*U \to f^*E$ gives a map
    $(f^*i)^*(f^*H)^*\transport$ over $U$ from $(f^*i)^*X$ to itself because
    $rsi = i$ by assumption.
    \begin{equation*}\label{eqn:star}\tag{$\star$}\small
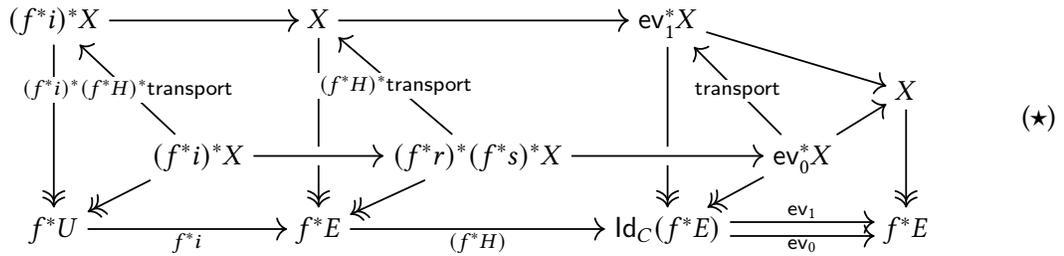
      \begin{tikzcd}[cramped, row sep=small, column sep=small]
        {(f^*i)^*X} && X && {\ev_1^*X} \\
        &&&&&& X \\
        & {(f^*i)^*X} && {(f^*r)^*(f^*s)^*X} && {\ev_0^*X} \\
        f^*U && f^*E && {\Id_C(f^*E)} && f^*E
        \arrow[from=1-1, to=1-3]
        \arrow[two heads, from=1-1, to=4-1]
        \arrow[from=1-3, to=1-5]
        \arrow[two heads, from=1-3, to=4-3]
        \arrow[from=1-5, to=2-7]
        \arrow[two heads, from=1-5, to=4-5]
        \arrow[two heads, from=2-7, to=4-7]
        \arrow["{(f^*i)^*(f^*H)^*\transport}"{description}, from=3-2, to=1-1]
        \arrow[two heads, from=3-2, to=4-1]
        \arrow["{(f^*H)^*\transport}"{description}, from=3-4, to=1-3]
        \arrow[two heads, from=3-4, to=4-3]
        \arrow["\transport"{description}, from=3-6, to=1-5]
        \arrow[from=3-6, to=2-7]
        \arrow[two heads, from=3-6, to=4-5]
        \arrow["{f^*i}"', from=4-1, to=4-3]
        \arrow["{(f^*H)}"', from=4-3, to=4-5]
        \arrow["{\ev_1}", shift left, from=4-5, to=4-7]
        \arrow["{\ev_0}"', shift right, from=4-5, to=4-7]
        \arrow[crossing over, from=3-2, to=3-4]
        \arrow[crossing over, from=3-4, to=3-6]
      \end{tikzcd}
    \end{equation*}
    Moreover in the case of \Cref{itm:retract-lift-strong}, by strength, we have
    $H|_i = \refl \cdot i$, so
    \begin{equation*}
      (f^*i)^*(f^*H)^*\transport
      = f^*(i^*H^*\transport)
      = f^*(i^*\refl^*\transport) = f^*(i^*\id) = \id
    \end{equation*}
  \end{enumerate}

  Therefore, we can obtain the required extension by taking
  \begin{equation*}
   \overline{x} \coloneqq (f^*H)^*\transport \cdot (f^*s)^*x'
  \end{equation*}
  where $x'$ is the extension of $x$ along $f^*j$.
  Step by step, starting with a fibration $X \twoheadrightarrow E$ and a section
  $x \colon U \hookrightarrow i^*X$ of the pullback,
  
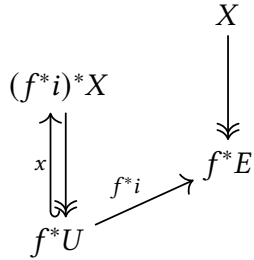
\begin{equation*}
    \begin{tikzcd}[cramped, row sep=small, column sep=small]
      && X \\
      {(f^*i)^*X} \\
      && f^*E \\
      f^*U
      \arrow[two heads, from=1-3, to=3-3]
      \arrow[shift left, two heads, from=2-1, to=4-1]
      \arrow["x", shift left, hook', from=4-1, to=2-1]
      \arrow["{f^*i}", from=4-1, to=3-3]
    \end{tikzcd}
  \end{equation*}
  one applies the following steps:
  \begin{enumerate}
    \item We first pullback $X \twoheadrightarrow f^*E$ along $f^*r$ to get a
    $\pi$-fibration $(f^*r)^*X \twoheadrightarrow f^*E'$, as seen below
    
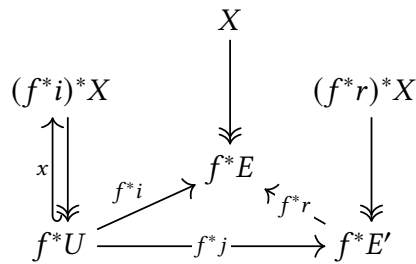
\begin{equation*}
      \begin{tikzcd}[cramped, row sep=small, column sep=small]
        && X \\
        {(f^*i)^*X} &&& {(f^*r)^*X} \\
        && f^*E \\
        f^*U &&& {f^*E'}
        \arrow[two heads, from=1-3, to=3-3]
        \arrow[shift left, two heads, from=2-1, to=4-1]
        \arrow[shift left, two heads, from=2-4, to=4-4]
        \arrow["x", shift left, hook', from=4-1, to=2-1]
        \arrow["{f^*i}", from=4-1, to=3-3]
        \arrow["{f^*j}"{description}, from=4-1, to=4-4]
        \arrow["{f^*r}"{description}, from=4-4, to=3-3]
      \end{tikzcd}
    \end{equation*}
    \item Because $r$ is under $U$ and $j$ uniformly lifts against $\pi$ by
    assumption, we can extend $x \colon f^*U \hookrightarrow (f^*i)^*X = (f^*j)^*(f^*r)^*X$ along
    $f^*j$ to some $x' \colon f^*E' \hookrightarrow (f^*r)^*X$ as seen below
    \begin{equation*}
      \begin{tikzcd}[cramped, row sep=small, column sep=small]
        && X \\
        {(f^*i)^*X} &&& {(f^*r)^*X} \\
        && f^*E \\
        f^*U &&& {f^*E'}
        \arrow[two heads, from=1-3, to=3-3]
        \arrow[shift left, two heads, from=2-1, to=4-1]
        \arrow[shift left, two heads, from=2-4, to=4-4]
        \arrow["x", shift left, hook', from=4-1, to=2-1]
        \arrow["{f^*i}", from=4-1, to=3-3]
        \arrow["{f^*j}"{description}, from=4-1, to=4-4]
        \arrow["{x'}", shift left, hook', from=4-4, to=2-4]
        \arrow["{f^*r}"{description}, from=4-4, to=3-3]
      \end{tikzcd}
    \end{equation*}
    where in \Cref{itm:retract-lift-strong}, we additionally have
    \begin{equation*}
      (f^*j)^*x' = x
    \end{equation*}
    \item We then pullback $x' \colon f^*E' \hookrightarrow (f^*r)^*X$ along $f^*s \colon f^*E \to f^*E'$
    to obtain a section $(f^*s)^*x' \colon f^*E \to (f^*s)^*(f^*r)^*X$, as seen below
    \begin{equation*}
      \begin{tikzcd}[cramped, row sep=small, column sep=small]
        && X \\
        {(f^*i)^*X} &&& {(f^*r)^*X} \\
        && E && {(f^*s)^*(f^*r)^*X} \\
        f^*U &&& {f^*E'} \\
        &&&& f^*E
        \arrow[two heads, from=1-3, to=3-3]
        \arrow[shift left, two heads, from=2-1, to=4-1]
        \arrow[shift left, two heads, from=2-4, to=4-4]
        \arrow[shift left, two heads, from=3-5, to=5-5]
        \arrow["x", shift left, hook', from=4-1, to=2-1]
        \arrow["{f^*i}", from=4-1, to=3-3]
        \arrow["{f^*j}"{description}, from=4-1, to=4-4]
        \arrow["{f^*i}"', from=4-1, to=5-5]
        \arrow["{x'}", shift left, hook', from=4-4, to=2-4]
        \arrow["{f^*r}"{description}, from=4-4, to=3-3]
        \arrow["{(f^*s)^*x'}", shift left, hook', from=5-5, to=3-5]
        \arrow["{f^*s}"{description}, from=5-5, to=4-4]
      \end{tikzcd}
    \end{equation*}
    \item But our goal is to have something that ends up in $X$, so we make a
    final adjustment using $(f^*H)^*\transport \colon (f^*s)^*(f^*r)^*X \to X$
    over $f^*E$ from \Cref{eqn:star} and put
    $\overline{x} \coloneqq (f^*H)^*\transport \cdot (f^*s)^*x'$, as seen below
    \begin{equation*}
      \begin{tikzcd}[cramped, row sep=small, column sep=small]
        && X \\
        {(f^*i)^*X} &&& {(f^*r)^*X} \\
        && f^*E && {(f^*s)^*(f^*r)^*X} \\
        f^*U &&& {f^*E'} && X \\
        &&&& f^*E
        \arrow[two heads, from=1-3, to=3-3]
        \arrow[shift left, two heads, from=2-1, to=4-1]
        \arrow[shift left, two heads, from=2-4, to=4-4]
        \arrow["{(f^*H)^*\transport}"{}, from=3-5, to=4-6]
        \arrow[shift left, two heads, from=3-5, to=5-5]
        \arrow["x", shift left, hook', from=4-1, to=2-1]
        \arrow["{f^*i}", from=4-1, to=3-3]
        \arrow["{f^*j}"{description}, from=4-1, to=4-4]
        \arrow["{f^*i}"', from=4-1, to=5-5]
        \arrow["{x'}", shift left, hook', from=4-4, to=2-4]
        \arrow["{f^*r}"{description}, from=4-4, to=3-3]
        \arrow[shift left, two heads, from=4-6, to=5-5]
        \arrow["{(f^*s)^*x'}", shift left, hook', from=5-5, to=3-5]
        \arrow["{f^*s}"{description}, from=5-5, to=4-4]
        \arrow["{\overline{x}}", shift left, hook', from=5-5, to=4-6]
      \end{tikzcd}
    \end{equation*}
  \end{enumerate}
  In the case of \Cref{itm:retract-lift-lower}, we are now done as we only need
  to provide some section of $X \twoheadrightarrow f^*E$.

  In the case of \Cref{itm:retract-lift-strong}, this procedure gives a required
  extension $\overline{x}$ because the adjustment factor by $(f^*H)^*\transport$ is
  killed off by pullback along $f^*i$.
  \begin{equation*}
    (f^*i)^*\overline{x} = (f^*i)^*(f^*s)^*x' \cdot (f^*i)^*(f^*H)^*\transport = (f^*j)^*x' \cdot \id = x
  \end{equation*}

  In both cases, due to the uniformity of the $\transport$ structure and
  uniformity lift of $j$ against $\pi$ and the uniformity of the $\transport$
  structure, for any $g \colon D \to C$ we have $g^*x' = (g^*x)'$ so
  \begin{align*}
    g^*\overline{x}
    &=
      g^*((f^*H)^*\transport) \cdot g^*((f^*s)^*x')
    \\
    &=
      (g^*f^*H)^*(g^*\transport) \cdot (g^*f^*s)^*(g^*x')
    \\
    &=
      (g^*f^*H)^*\transport \cdot (g^*f^*s)^*(g^*x)'
    \\
    &=
      ((gf)^*H)^*\transport \cdot ((gf)^*s)^*(g^*x)'
    \\
    g^*\overline{x}
    &=
      \overline{g^*x}
  \end{align*}
  This proves uniformity of the lift.
\end{proof}

A special case of constant-restricted homotopy retracts are duals of strong
deformation retracts.
\begin{definition}[{\cite[Definition 2.4.16]{cis19}}]\label{def:sdr}
  Suppose $\pi \colon \tMcU \to \McU$ is a universal map equipped with a
  pre-$\Id$-type structure $\Id \colon \tMcU \times_\McU \tMcU \to \McU$.

  Fix an object $B \in \bC$.
  Given a map $f \colon E \to Y \in \sfrac{\bC}{B}$ where
  $E \twoheadrightarrow B$ is a $\pi$-fibration, another map
  $g \colon Y \to E \in \sfrac{\bC}{B}$ exhibits $f$ as the \emph{dual of a
    strong $\Id$-deformation retract} when $(g,f)$ is a (precise)
  section-retraction pair (i.e. $fg = \id_Y$) and $gf$ is equipped with an
  $\Id$-homotopy to the identity that is constant when restricted along $g$.
  \begin{equation*}
    \begin{tikzcd}
      Y & E & {\Id_B(E)} \\
      E && {E \times_B E}
      \arrow["g", hook, from=1-1, to=1-2]
      \arrow["g"', hook, from=1-1, to=2-1]
      \arrow["\refl", hook, from=1-2, to=1-3]
      \arrow[two heads, "{\ev_\partial}", from=1-3, to=2-3]
      \arrow["H"{description}, dashed, from=2-1, to=1-3]
      \arrow["{(gf,\id)}"', from=2-1, to=2-3]
    \end{tikzcd}
  \end{equation*}
  %
  %
\end{definition}

\begin{proposition}\label{prop:dual-df-lift}
  Suppose $\pi \colon \tMcU \to \McU$ is equipped with a pre-$\Id$-type
  structure $\Id \colon \tMcU \times_\McU \tMcU \to \McU$.

  Let $f \colon E \to Y$ be a map over $E$ where $E$ is $\pi$-fibrant.
  \begin{enumerate}
    \item\label{itm:dual-df-half-lift}
    Suppose that the pre-$\Id$-type structure is equipped with a $\transport$-structure.
    Then, a section $g \colon Y \hookrightarrow E$ is an $\Id$-retraction of
    $f \colon E \to Y$ if and only if
    \begin{equation*}
      \left(
        {\begin{tikzcd}[cramped]
            Y \ar[d, "{g}"', hook] \\ E
          \end{tikzcd}}
        \fraclowerslash{B}
        {\begin{tikzcd}[cramped]
            B \times \tMcU \ar[d, two heads] \\ B \times \McU
          \end{tikzcd}}\right)
      \neq
      \emptyset
    \end{equation*}
    \item\label{itm:dual-df-lift}
    Suppose the pre-$\Id$-type structure is an $\Id$-type structure and $\pi$
    has a $\Pi$-type structure.
    Then, a section $g \colon Y \to E$ exhibits $f \colon E \to Y$ as the dual
    of a strong $\Id$-deformation retraction if and only if
    \begin{equation*}
      \left(
        {\begin{tikzcd}[cramped]
            Y \ar[d, "{g}"', hook] \\ E
          \end{tikzcd}}
        \fracsquareslash{B}
        {\begin{tikzcd}[cramped]
            B \times \tMcU \ar[d, two heads] \\ B \times \McU
          \end{tikzcd}}\right)
      \neq
      \emptyset
    \end{equation*}
  \end{enumerate}
\end{proposition}
\begin{proof}
  We first show the $\Rightarrow$ directions for both parts.
  Because $g$ is a precise section of $f$, we have the following
  \begin{equation*}
    \begin{tikzcd}[cramped]
      & Y \\
      E & Y & E
      \arrow["g"', from=1-2, to=2-1]
      \arrow["{=}"{description}, from=1-2, to=2-2]
      \arrow["g", from=1-2, to=2-3]
      \arrow["f"', from=2-1, to=2-2]
      \arrow["g"', from=2-2, to=2-3]
    \end{tikzcd} \in \sfrac{\bC}{B}
  \end{equation*}
  In both cases \Cref{itm:dual-df-half-lift} and \Cref{itm:dual-df-lift}, there
  is some $\Id$-homotopy $H$ between the bottom row composite $gf$ and $\id$ and
  further in case \Cref{itm:dual-df-lift} one has that $H|_s = \refl \cdot s$ by
  strength.
  Because the identity map $Y \to Y$ lifts on the left against $\tMcU \to \McU$,
  the result for \Cref{itm:dual-df-half-lift} follows by
  \Cref{lem:retract-lift}\Cref{itm:retract-lift-lower} and the result for
  \Cref{itm:dual-df-lift} follows by
  \Cref{lem:retract-lift}\Cref{itm:retract-lift-strong}.

  The $\Leftarrow$ direction follows because
  $\ev_\partial \colon \Id_B(X) \twoheadrightarrow X \times_B X$ is itself a
  $\pi$-fibration by the $\Id$-type structure.
  Thus, in both the $\Leftarrow$ directions for cases
  \Cref{itm:dual-df-half-lift} and \Cref{itm:dual-df-lift}, one has a diagonal
  filler as follows, where the upper triangle does not necessarily commute for
  \Cref{itm:dual-df-half-lift}, but does commute for \Cref{itm:dual-df-lift}.
  \begin{equation*}
    \begin{tikzcd}
      Y & E & {\Id_B(E)} \\
      E && {E \times_B E}
      \arrow["g", hook, from=1-1, to=1-2]
      \arrow["g"', hook, from=1-1, to=2-1]
      \arrow["\refl", hook, from=1-2, to=1-3]
      \arrow[two heads, "{\ev_\partial}", from=1-3, to=2-3]
      \arrow[""{name=0, anchor=center, inner sep=0}, "H"{description}, dashed, from=2-1, to=1-3]
      \arrow["{(?)}"{description}, draw=none, from=1-1, to=0]
      \arrow["{(gf,\id)}"', from=2-1, to=2-3]
    \end{tikzcd}
  \end{equation*}
  In other words, in case \Cref{itm:dual-df-half-lift}, $(f,g)$ is an
  $\Id$-homotopy section-retraction pair and furthermore in
  \Cref{itm:dual-df-lift}, $g$ exhibits $f$ as the dual of a strong
  $\Id$-deformation retract.
\end{proof}

In particular, applying \Cref{prop:dual-df-lift} in the context of univalence
for internal universes, we get the following characterisation of internal
univalent universes.

\begin{theorem}\label{thm:univalence-sdf}
  Let $\pi \colon \tMcU \to \McU$ be a universal map equipped with pre-$\Id$-,
  $\Pi$-, $\Sigma$-structures and an internal universe
  \begin{equation*}
    \pi_0 \colon \tMcU_0 \to \McU_0
  \end{equation*}
  Further assume the internal universe $\pi_0$ is equipped with pre-$\Id$-,
  $\Pi$-, $\Sigma$-structures respectively denoted $\Id_0,\Pi_0,\Sigma_0$ and
  $\pi$ is equipped with a $\Sigma$-type structure.
  \begin{enumerate}
    \item\label{itm:univalence-sdf-book}
    Suppose the pre-$\Id$-structure on $\pi$ has a $\transport$-structure.
    Then, internal (universe, pre-$\Id$-type)-pair $(\pi_0,\Id_0)$ is book
    $\pi$-univalent exactly when
    $\trv \cdot p \colon\HIso_{\McU_0}^{\Id_0}(\tMcU_0) \to \McU_0 \to
    \HIso_{\McU_0}^{\Id_0}(\tMcU_0)$ is $\Id$-homotopic to the identity, where $p$ is
    either $p=\src$ or $p=\dest$.
    \item\label{itm:univalence-sdf-pointed}
    Suppose the pre-$\Id$-structure on $\pi$ is a full $\Id$-structure.
    Then, the internal (universe, pre-$\Id$-type)-pair $(\pi_0,\Id_0)$ is
    pointed $\pi$-univalent exactly when
    $\trv \cdot p \colon\HIso_{\McU_0}^{\Id_0}(\tMcU_0) \to \McU_0 \to
    \HIso_{\McU_0}^{\Id_0}(\tMcU_0)$ is $\Id$-homotopic to the identity via some
    homotopy that is constant when restricted along $\trv$, where $p$ is either
    $p=\src$ or $p=\dest$.
  \end{enumerate}

  In other words, book univalence is equivalent to the following lifting problem
  admitting a lower-half solution and pointed univalence is equivalent to the
  following lifting problem admitting a full solution.
  \begin{equation*}
    \begin{tikzcd}
      \McU_0 & \HIso_{\McU_0}^{\Id_0}(\tMcU_0) & {\Id_{\McU_0}(\HIso_{\McU_0}^{\Id_0}(\tMcU_0))} \\
      \HIso_{\McU_0}^{\Id_0}(\tMcU_0) && {\HIso_{\McU_0}^{\Id_0}(\tMcU_0) \times_{\McU_0} \HIso_{\McU_0}^{\Id_0}(\tMcU_0)}
      \arrow["{\trv}", hook, from=1-1, to=1-2]
      \arrow["{\trv}"', hook, from=1-1, to=2-1]
      \arrow["\refl", hook, from=1-2, to=1-3]
      \arrow[two heads, "{\ev_\partial}", from=1-3, to=2-3]
      \arrow[""{name=0, anchor=center, inner sep=0}, "H"{description}, dashed, from=2-1, to=1-3]
      \arrow["{(\trv\cdot p,\id)}"', from=2-1, to=2-3]
      \arrow["{(?)}"{description}, draw=none, from=1-1, to=0]
    \end{tikzcd}
  \end{equation*}
  %
  %
\end{theorem}
\begin{proof}
  Because $\pi$ admits a $\Sigma$-type structure, the maps $\src,\dest$ are
  $\pi$-fibrant by \Cref{cor:src-dest-int-uni-fib} and $\trv$ is by construction a
  section of both of these maps.
  Consequently, $\trv$ is a map between $\pi$-fibrant objects in $\bC$.
  The first and second parts now follow by the respective first and second parts
  of \Cref{prop:dual-df-lift} along with an unfolding of the definition of an
  $\Id$-section-retraction pair and the definition of the dual of a strong
  $\Id$-deformation retract.
\end{proof}


\subsection{Univalence Type-Theoretically}
The formulation of univalence in \Cref{def:axm-univalence} was motivated by
lifting conditions from Quillen model category models of homotopy type theory.
However, when working type-theoretically, univalence is often stated by saying
that the map obtained by $\MsJ$-elimination between $\HIso$ and the $\Id$-type
path object is itself an $\Id$-type homotopy equivalence.
We check our sanity by giving an analogue result of \cite[Theorem 3.3.7]{kl21}
showing that these two versions of univalence are equivalent.

\begin{definition}\label{def:tt-univalence}
  Let $\pi \colon \tMcU \to \McU$ be a universal map equipped with
  $\Sigma$-,$\Pi$-,$\Id$-type structures and an internal universe
  \begin{equation*}
    \pi_0 \colon \tMcU_0 \to \McU_0
  \end{equation*}
  Further assume the internal universe $\pi_0$ is equipped with $\Id$-, $\Pi$-,
  $\Sigma$-structures respectively denoted $\Id_0,\Pi_0,\Sigma_0$ and the
  external universe $\pi$ is equipped with a $\Sigma$-type structure.

  A \emph{type-theoretic pointed univalence} structure on the internal universe
  relative to the ambient universe consists of a pair of maps
  $s,r \colon \HIso_{\McU_0}(\tMcU_0) \rightrightarrows \Id(\McU_0)$ over
  $\McU_0\times\McU_0$ and under $\McU_0$
  \begin{equation*}
    \begin{tikzcd}[cramped, row sep=small, column sep=small]
      & {\McU_0} \\
      {\HIso_{\McU_0}(\tMcU_0)} && {\Id(\McU_0)} \\
      & {\McU_0 \times \McU_0}
      \arrow["\trv"', hook', from=1-2, to=2-1]
      \arrow["\refl", hook, from=1-2, to=2-3]
      \arrow["s"{description}, shift left=3, dashed, from=2-1, to=2-3]
      \arrow["{(\src,\dest)}"', two heads, from=2-1, to=3-2]
      \arrow["J"{description}, from=2-3, to=2-1]
      \arrow["r"', shift left=-3, dashed, from=2-1, to=2-3]
      \arrow["{\ev_\partial}", two heads, from=2-3, to=3-2]
    \end{tikzcd}
  \end{equation*}
  along with $\Id$-homotopies $H_s$ and $H_r$ over $\McU_0 \times \McU_0$
  respectively from $Js$ to $\id$ \emph{constant along $\trv$} and $rJ$ to $\id$
  \emph{constant along $\refl$} in the sense of \Cref{def:const-res-htpy}.

  A \emph{type-theoretic book univalence} structure is just a choice maps $s,r$
  with $\Id$-homotopies $H_s,H_r$ as above, except $H_s$ and $H_r$ need not be
  constant when restricted.

  The internal univalence is pointed (respectively, book) univalent relative to
  the ambient universe when it can be equipped with such a pointed
  (respectively, book) $\pi$-univalence structure.
\end{definition}

Thanks to the homotopy retract stability result from \Cref{lem:retract-lift},
type-theoretic pointed univalence immediately implies lifting-form pointed
univalence as defined in \Cref{def:axm-univalence}.
The other direction is basically the same proof of homotopical uniqueness of
factorisations in model categories, which we recall in the setting of universe
categories.

\begin{lemma}\label{lem:path-htpy}
  Assume that $\pi\colon\tMcU \to \McU$ is equipped with $\Id$-types.
  Then, for every commutative diagram below where $E_i \twoheadrightarrow B$ are
  $\pi$-fibrant and $X \to E_i$ lifts against $\pi$, whenever
  $J \colon E_0 \to E_1$ is a filler, one can find fillers
  $s,r \colon E_1 \rightrightarrows E_0$ in the other direction
  \begin{equation*}
    \begin{tikzcd}[cramped]
      & X \\
      {E_0} && {E_1} \\
      & B
      \arrow[from=1-2, to=2-1]
      \arrow[from=1-2, to=2-3]
      \arrow["J"{description}, from=2-1, to=2-3]
      \arrow[two heads, from=2-1, to=3-2]
      \arrow["r", shift left=3, dashed, from=2-3, to=2-1]
      \arrow["s"', shift right=3, dashed, from=2-3, to=2-1]
      \arrow[two heads, from=2-3, to=3-2]
    \end{tikzcd}
  \end{equation*}
  along with $\Id$-homotopies $H_s$ and $H_r$ over $B$ respectively from $Js$ to
  $\id$ constant along $X \to E_1$ and $rJ$ to $\id$ constant along $X \to E_0$.
\end{lemma}
\begin{proof}
  By the lifting property of $X \to E_1$, we can find a lift
  $\ell \colon E_1 \to E_0$, which we take to be both $s$ and $r$.
  Then, the homotopies from $J\ell$ to $\id$ constant along $X \to E_1$ and
  $\ell J$ to $\id$ constant along $X \to E_0$ are provided by solving the two
  lifting problems.
  \begin{center}
    \begin{minipage}{0.45\linewidth}
      \begin{equation*}
        \begin{tikzcd}[cramped]
          X & {E_1} & {\Id_B(E_1)} \\
          {E_1} && {E_1 \times_B E_1}
          \arrow[from=1-1, to=1-2]
          \arrow[from=1-1, to=2-1]
          \arrow["\refl", hook, from=1-2, to=1-3]
          \arrow["\Delta"{description, pos=0.7}, from=1-2, to=2-3]
          \arrow["{\ev_\partial}", two heads, from=1-3, to=2-3]
          \arrow["{(J\ell,\id)}"', from=2-1, to=2-3]
          \arrow["{H_s}"{description}, dashed, from=2-1, to=1-3, crossing over]
        \end{tikzcd}
      \end{equation*}
    \end{minipage}
    \begin{minipage}{0.45\linewidth}
      \begin{equation*}
        \begin{tikzcd}[cramped]
          X & {E_0} & {\Id_B(E_0)} \\
          {E_0} && {E_0 \times_B E_0}
          \arrow[from=1-1, to=1-2]
          \arrow[from=1-1, to=2-1]
          \arrow["\refl", hook, from=1-2, to=1-3]
          \arrow["\Delta"{description, pos=0.7}, from=1-2, to=2-3]
          \arrow["{\ev_\partial}", two heads, from=1-3, to=2-3]
          \arrow["{(\ell J,\id)}"', from=2-1, to=2-3]
          \arrow["{H_r}"{description}, dashed, from=2-1, to=1-3, crossing over]
        \end{tikzcd}
      \end{equation*}
    \end{minipage}
  \end{center}
\end{proof}

With this, we can observe the expected equivalence result between univalence
formulated in terms of lifting and univalence formulated type-theoretically.

\begin{theorem}[{\cite[Theorem 3.3.7]{kl21}}]\label{thm:univalence-tt}
  Let $\pi \colon \tMcU \to \McU$ be a universal map equipped with
  $\Sigma$-,$\Pi$-,$\Id$-type structures and an internal universe
  \begin{equation*}
    \pi_0 \colon \tMcU_0 \to \McU_0
  \end{equation*}
  Further assume the internal universe $\pi_0$ is equipped with $\Id$-, $\Pi$-,
  $\Sigma$-structures respectively denoted $\Id_0,\Pi_0,\Sigma_0$ and $\pi$ is
  equipped with a $\Sigma$-type structure.

  Then, the internal univalence is type-theoretically pointed (respectively,
  book) univalent relative to the ambient universe as formulated in
  \Cref{def:tt-univalence} if and only if $(\pi_0,\Id_0)$ is pointed
  (respectively, book) $\pi$-univalent as phrased in terms of lifting in
  \Cref{def:axm-univalence}.
\end{theorem}
\begin{proof}
  For the $\Rightarrow$ direction, we must show that pointed (respectively,
  book) univalence formulated type-theoretically implies pointed (respectively,
  book) univalence formulated as a lifting condition.

  For the $\Rightarrow$ direction of pointed univalence, we note that as part of
  the definition of type-theoretic pointed univalence, we have
  \begin{equation*}
    \begin{tikzcd}[cramped]
      & {\McU_0} \\
      {\HIso_{\McU_0}^{\Id_0}(\tMcU_0)} & {\Id(\McU_0)} & {\HIso_{\McU_0}^{\Id_0}(\tMcU_0)}
      \arrow["\trv"', hook, from=1-2, to=2-1]
      \arrow["\refl"{description}, hook, from=1-2, to=2-2]
      \arrow["\trv", hook, from=1-2, to=2-3]
      \arrow["s"', from=2-1, to=2-2]
      \arrow["J"', from=2-2, to=2-3]
    \end{tikzcd}
  \end{equation*}
  where there is some $\Id$-homotopy $H_s$ between the bottom row composite $Js$
  and $\id$ such that $H_s|_\trv = \refl \cdot \trv$.
  Because $\MsJ$-elimination ensures that
  $\refl \colon \tMcU_0 \hookrightarrow \Id(\McU_0)$ lifts against
  $\tMcU \to \McU$, we use \Cref{lem:retract-lift}\Cref{itm:retract-lift-strong}
  to get pointed univalence in lifting form as from \Cref{def:axm-univalence}.

  For the $\Rightarrow$ direction of book univalence, we recall that it still
  gives some $s \colon \HIso_{\McU_0}(\tMcU_0) \to \Id(\McU_0)$ with a homotopy
  $H_s \colon Js \simeq \id$.
  Then, again, because $\MsJ$-elimination ensures that
  $\refl \colon \tMcU_0 \hookrightarrow \Id(\McU_0)$ lifts against
  $\tMcU \to \McU$, we use \Cref{lem:retract-lift}\Cref{itm:retract-lift-lower}
  to get book univalence in lifting form as from \Cref{def:axm-univalence}.

  For the $\Leftarrow$ direction, we must show that pointed (respectively, book)
  univalence formulated in terms of lifting implies pointed (respectively, book)
  univalence formulated type-theoretically.
  The $\Leftarrow$ direction for pointed univalence is by \Cref{lem:path-htpy}.
  To show the $\Leftarrow$ direction for book univalence, we note that by
  \Cref{thm:univalence-sdf}, one has an $\Id$-homotopy section-retraction pair
  \begin{equation*}
    \HIso_{\McU_0}^{\Id_0}(\tMcU_0)
    \xrightarrow{\src}
    \McU_0
    \xrightarrow{\trv}
    \HIso_{\McU_0}^{\Id_0}(\tMcU_0) \in \sfrac{\bC}{\McU_0}
  \end{equation*}
  But also one has a precise section-retraction pair
  \begin{equation*}
    \McU_0
    \xrightarrow{\trv}
    \HIso_{\McU_0}^{\Id_0}(\tMcU_0)
    \xrightarrow{\src}
    \McU_0 \in \sfrac{\bC}{\McU_0}
  \end{equation*}
  Thus, book univalence via lifting shows that
  $\HIso_{\McU_0}^{\Id_0}(\tMcU_0) \twoheadrightarrow \McU_0$ as a fibrant
  object over $\McU_0$ is contractible.
  But so is $\Id(\McU_0) \twoheadrightarrow \McU_0$, so one may find an
  equivalence $\HIso_{\McU_0}^{\Id_0}(\tMcU_0) \simeq \Id(\McU_0)$ adjusted to
  be fibrewise over $\McU_0 \times \McU_0$ using \cite[Theorem 4.7.7]{hottbook}.
\end{proof}


\subsection{The CwR of Univalent Type Theory}
An immediate categorical metatheory application of
\Cref{thm:univalence-sdf,thm:univalence-tt} is a relatively lightweight
construction of the CwR of univalent type theory via a colimiting construction.
This is because these theorems together imply that pointed (respectively, book)
univalence is equivalent to just requiring one map that is a filler
(respectively, lower-half filler) to a square.

We first define intensional type theory, as a CwR, to be the CwR freely
generated by a representable arrow equipped with generic
$\Sigma$-,$\Pi$-,$\Id$-structures.
\begin{construction}
  Define $\Univ_\ITT \in \CwR$ as the following bicolimiting cocone
  \begin{equation*}
    \begin{tikzcd}[cramped, row sep=small, column sep=small]
      & \Univ \\
      {\Univ_\Sigma} & {\Univ_\Pi} & {\Univ_\Id} \\
      & {\Univ_\ITT}
      \arrow[from=1-2, to=2-1]
      \arrow[from=1-2, to=2-2]
      \arrow[from=1-2, to=2-3]
      \arrow[dashed, from=2-1, to=3-2]
      \arrow[dashed, from=2-2, to=3-2]
      \arrow[dashed, from=2-3, to=3-2]
    \end{tikzcd}
  \end{equation*}
  where $\Univ_\Sigma,\Univ_\Pi,\Univ_\Id$ are from \Cref{constr:cwr-frags}
\end{construction}

Then, the type theory consisting of an internal universe and an external
universe each equipped with the structure of intensional type theory is
constructed as the following bicolimit.
\begin{construction}
  Define $\IntUniv_\ITT \in \CwR$ as the following bicolimiting cocone
  \begin{equation*}
    \begin{tikzcd}[cramped]
      \Univ & {\Univ_\ITT} \\
      \IntUniv & {\IntUniv_\ITT}
      \arrow["\pi", from=1-1, to=1-2]
      \arrow["\pi"', shift right, from=1-1, to=2-1]
      \arrow["{\pi_0}", shift left, from=1-1, to=2-1]
      \arrow[dashed, from=1-2, to=2-2]
      \arrow[dashed, from=2-1, to=2-2]
    \end{tikzcd}
  \end{equation*}
  where the maps $\pi$ and $\pi_0$ respectively select the correspondingly named
  representable map in the codomain CwRs.
\end{construction}

Informed by \Cref{thm:univalence-sdf}, one can now construct the CwRs of book
and pointed univalent type theory respectively.
\begin{construction}\label{cwr:ua}
  Define $\Univ_\BookUA \in \CwR$ as the bipushout
  \begin{equation*}
    \Univ_\BookUA \coloneqq \IntUniv_\ITT
    \cup_{\left\{{\scriptsize\begin{tikzcd}[cramped, row sep=small, column sep=small]& \bullet \\ \bullet & \bullet \arrow[from=1-2, to=2-2]	\arrow[from=2-1, to=2-2] \end{tikzcd}}\right\}}
    \left\{{\scriptsize
      \begin{tikzcd}[cramped, row sep=small, column sep=3em]
        & {\Id_{\McU_0}(\HIso_{\McU_0}^{\Id_0}(\tMcU_0))} \\
        {\HIso_{\McU_0}^{\Id_0}(\tMcU_0)} & {\HIso_{\McU_0}^{\Id_0}(\tMcU_0) \times_{\McU_0} \HIso_{\McU_0}^{\Id_0}(\tMcU_0)}
        \arrow[from=1-2, to=2-2]
        \arrow[dashed, from=2-1, to=1-2]
        \arrow["{(\trv \cdot \src, \id)}"', from=2-1, to=2-2]
      \end{tikzcd}}
    \right\}
  \end{equation*}
  and similarly define $\Univ_\PointedUA \in \CwR$ as the bipushout
  \begin{equation*}
    \Univ_\PointedUA \coloneqq \IntUniv_\ITT
    \cup_{\left\{{\scriptsize\begin{tikzcd}[cramped, row sep=small, column sep=small] \bullet & \bullet \\ \bullet & \bullet \arrow[from=1-1, to=1-2] \arrow[from=1-1, to=2-1]  \arrow[from=1-2, to=2-2] \arrow[from=2-1, to=2-2] \end{tikzcd}}\right\}}
    \left\{{\scriptsize
        \begin{tikzcd}[cramped, row sep=small, column sep=3em]
          {\McU_0} & {\Id_{\McU_0}(\HIso_{\McU_0}^{\Id_0}(\tMcU_0))} \\
          {\HIso_{\McU_0}^{\Id_0}(\tMcU_0)} & {\HIso_{\McU_0}^{\Id_0}(\tMcU_0) \times_{\McU_0} \HIso_{\McU_0}^{\Id_0}(\tMcU_0)}
          \arrow["{\refl(\trv)}", from=1-1, to=1-2]
          \arrow["\trv"', from=1-1, to=2-1]
          \arrow[from=1-2, to=2-2]
          \arrow[dashed, from=2-1, to=1-2]
          \arrow["{(\trv \cdot \src, \id)}"', from=2-1, to=2-2]
        \end{tikzcd}}
    \right\}
  \end{equation*}
\end{construction}

Correctness of these constructions are readily observed as follows.
\begin{theorem}
  Let $\bC$ be a finitely complete category.
  \begin{enumerate}
    \item Isomorphism classes of maps $M \colon \Univ_\ITT \to \bC \in \CwR$ are in
    bijective correspondence with isomorphism classes of choices of a universal
    map $M\pi \colon M\tMcU \to M\McU \in \bC$ equipped with a choice of
    $\Sigma$-,$\Pi$-,$\Id$-type structures.
    \item Isomorphism classes $M \colon \IntUniv_\ITT \to \bC \in \CwR$ are in
    bijective correspondence with isomorphism classes of choices of a universal
    map $M\pi \colon M\tMcU \to M\McU \in \bC$ along with a choice of
    $\Sigma$-,$\Pi$-,$\Id$-type structures and a choice of an internal universe
    $M\pi_0 \colon M\tMcU_0 \to M\McU_0$ of $M\pi$ also equipped with a choice
    of $\Sigma$-,$\Pi$-,$\Id$-type structures.
    \item Factorisations of a map $M \colon \IntUniv_\ITT \to \bC \in \CwR$ via
    $\IntUniv_\ITT \to \Univ_\PointedUA$ (respectively,
    $\IntUniv_\ITT \to \Univ_\BookUA$) exist precisely when the internal
    universe is type-theoretically pointed (respectively, book) univalent.
  \end{enumerate}
\end{theorem}
\begin{proof}
  The first two parts follow by \Cref{prop:cwr-itt} and the universal property
  of the bicolimit and the biadjunction from \cite[Corollaries 3.2.16 and
  3.2.17]{jel25}.
  The final part follows from the equivalences given by
  \Cref{thm:univalence-tt,thm:univalence-sdf}.
\end{proof}


\subsection{Remarks on Pointed Univalence}
\label{sec:reexam}
\newcommand{\app}{\textsf{app}}
\newcommand{\ua}{\textsf{ua}}
In this paper, we presented three formulations of the univalence axiom in a universe category:
\begin{itemize}
    \item \cref{def:axm-univalence} gives a formulation that is easy to state in a universe category;
    \item \cref{thm:univalence-sdf} provides an equivalent condition that is not only much easier to verify in many universe categories that serve as models of homotopy type theory (e.g., those coming from model categories), but also much easier to state in the framework of CwRs
    \item \cref{thm:univalence-tt} provides an equivalent condition phrased in a more syntactic way, thus showing that the two conditions above indeed model univalence.
\end{itemize}
As we indicated before, our formulation of pointed univalence is stronger than how the axiom is typically stated.
We explain this comparison now.

\subsubsection{Pointed Univalence Versus Book Univalence}
In \Cref{def:tt-univalence}, we have formulated type-theoretic pointed
univalence by requiring that the homotopies must restrict to the constant
homotopy along $\trv$ and $\refl$ and showed in \Cref{thm:univalence-tt} that
our lifting formulation of pointed univalence in \Cref{def:axm-univalence} is
equivalent to it.
In the HoTT book presentation of univalence \cite[Axiom 2.10.3]{hottbook},
however, one requires only that $s,r$ are homotopy sections and retractions of
$J$, without the extra restriction to constant homotopy condition.

Informally, the extra requirement that $H_s \colon Js \simeq \id$ restricts to
$\refl$ along $\trv \colon \tMcU_0 \to \HIso_{\McU_0}(\tMcU_0)$ translates in
type theory notation to mean that there is a family of maps $s$, indexed by
internal types $A,B : \McU_0$
\begin{equation*}\label{eqn:J-sec}\tag{\textsc{$J$-sec}}
  s_{(A,B)} : \HIso(\El(A),\El(B)) \to \Id_{\McU_0}(A,B)
\end{equation*}
along with homotopies indexed by $e : \HIso(\El(A),\El(B))$
\begin{equation*}\label{eqn:Hse}\tag{$H_s(e)$}
  H_s(e) : \Id_{\HIso(\El(A),\El(B))}(e, \app(J \circ s_{(A,B)}, e))
\end{equation*}
such that for each $A : \McU_0$
\begin{align*}\label{eqn:res-const}\tag{\textsc{res-const}}
  s_{(A,B)}(\id_A) = \refl : \Id_{\McU_0}(A,A)
  &&
  H_s(\id_A) = \refl : \Id_{\HIso(\El(A),\El(A))}(\id_A, \id_A)
\end{align*}
In other words, the section must definitionally map the trivial homotopy
isomorphism to the trivial path and the chosen proof of homotopy section must
return the trivial path at trivial homotopy isomorphisms.

Given that our condition is easy to state and natural to verify in a variety of
settings, we believe that giving such a strengthening is justified.  We do not
know whether our formulation of pointed univalence can be deduced from book
univalence.

\subsubsection{Comparison with known proofs of Univalence}
Distilling the proof of univalence from the simplicial model \cite{kl21}, we
have only originally only formulated univalence in the pointed form.
It was only after completion of the proofs of the second parts of
\Cref{thm:univalence-sdf,thm:univalence-tt} regarding pointed univalence that we
discovered pointed univalence gives something stronger than book univalence.
This then led to the first parts of \Cref{thm:univalence-sdf,thm:univalence-tt}
treating book univalence.

Throughout the proofs of both pointed and book univalence, we have heavily
relied on the homotopy retract closure property of left classes to fibrations as
given by \Cref{lem:retract-lift}.
In particular, a closely related mechanised proof of type-theoretic book
univalence via an alternative condition was posted by Dan Licata and others due
to an observation by Mart\'{i}n Escard\'{o} also using retract properties in a
Google Groups
discussion.\footnote{\url{https://groups.google.com/g/homotopytypetheory/c/j2KBIvDw53s/m/YTDK4D0NFQAJ}}

We briefly describe Licata's observation using type-theoretic notation.
Licata observed that type-theoretic book univalence can be obtained by requiring
a map $s_{(A,B)}$ like from \Cref{eqn:J-sec}, which he calls $\ua$, along with a
family of homotopies $H_s$ like from \Cref{eqn:Hse}.
However, the conditions \Cref{eqn:res-const} is replaced with the condition that
for each homotopy isomorphism $e \colon \El(A) \to \El(B)$, transporting along
$\app(\ua_{(A,B)}, e) \colon \Id_{\McU_0}(A,B)$ gives parallel maps
\begin{equation*}
  \transport_{\app(\ua_{(A,B)}, e)}, e : \El(A) \rightrightarrows \El(B)
\end{equation*}
that are homotopic.

In contrast, we have obtained type-theoretic book univalence in
\Cref{thm:univalence-tt} by showing that under book univalence in terms of
lifting, for each $A : \McU_0$, the type
\begin{equation*}
  \Sigma(B:\McU_0).\textsf{HIso}(\El(A),\El(B))
\end{equation*}
is contractible, and then appealing to \cite[Theorem 4.7.7]{hottbook}.
This is also the same argument employed in \cite[Corollary 11]{cchm15} for
cubical type theory.
There, univalence was defined in terms of contractibility of
$\Sigma(B:\McU_0).\textsf{HIso}(\El(A),\El(B))$ and the equivalence of this
formulation to type-theoretic book univalence was also credited to Mart\'{i}n
Escard\'{o} made in a Google Groups
discussion.\footnote{\url{https://groups.google.com/g/homotopytypetheory/c/HfCB_b-PNEU/m/Ibb48LvUMeUJ}}

Conversely, the implication that type-theoretic book univalence implies book
univalence as a lifting condition can be seen as a homotopy isomorphism
induction principle.
This was observed and mechanised by \citeauthor{bl11} in \cite[Theorem
{\sffamily weq\_induction}]{bl11}.




\section{Pointed Functional Extensionality}\label{sec:funext}
In the rest of this paper, we build up to showing that pointed univalence is
preserved by formation of inverse diagrams.
For this, one requires a pointed version of functional extensionality, much like
how \cite{shu15} required functional extensionality to show book univalence is
closed under formation of inverse diagrams.
We define pointed functional extensionality in this section and study some of
its properties.

For the rest of this section, fix a universe category $\bC$ with universal map
$\pi \colon \tMcU \to \McU$.
\begin{definition}\label{def:funext}
  Let $\pi \colon \tMcU \to \McU$ be a universal map an $\Id$-type structure
  $\Id \colon \tMcU \times_\McU \tMcU \to \McU$.

  A pointed $\Id$-functional extensionality structure is a choice of a
  structured lift
  \begin{equation*}
    \PtdFunExt \in
    \left(\begin{tikzcd}[cramped, sep=small]
      \bP_\pi(\tMcU) \ar[d, "{\bP_\pi(\refl)}"'] \\ \bP_\pi(\Id_\McU(\tMcU))
    \end{tikzcd}
    \fracsquareslash{\McU}
    \begin{tikzcd}[cramped, sep=small]
      \tMcU \times \McU \ar[d] \\ \McU \times \McU
    \end{tikzcd}\right)
  \end{equation*}
  in the sense of \cite[Definition 3.1]{struct-lift}.
  %
  %
\end{definition}

Similar to \Cref{thm:univalence-tt}, we get the following more type-theoretic
characterisation of pointed functional extensionality.
\begin{definition}\label{def:funext-tt}
  Let $\pi \colon \tMcU \to \McU$ be a universal map an $\Id$-type structure
  $\Id \colon \tMcU \times_\McU \tMcU \to \McU$ and a $\Pi$-type structure
  $\Pi \colon \bP_\pi(\McU) \to \McU$.

  A \emph{type-theoretic} pointed $(\Id,\Pi)$-functional extensionality
  structure is an assignment to each $\pi$-fibration
  $q \colon B \twoheadrightarrow A$ and $\pi$-fibrant object
  $E \twoheadrightarrow B$ over $B$ two maps
  $s^{q,E},r^{q,E} \colon q_*(\Id_BE) \rightrightarrows \Id_A(q_*E)$
  over $q_*E \times_A q_*E$ and under $q_*E$
  \begin{equation*}
    \begin{tikzcd}[cramped, sep=small]
      & {q_*E} \\
      {\Id_A(q_*E)} && {q_*(\Id_BE)} \\
      & {q_*E \times_A q_*E}
      \arrow["\refl"', from=1-2, to=2-1, hook]
      \arrow["{q_*(\refl)}", from=1-2, to=2-3, hook]
      \arrow["J"{description}, from=2-1, to=2-3]
      \arrow["{\ev_\partial}"', two heads, from=2-1, to=3-2]
      \arrow["r", shift left=3, dashed, from=2-3, to=2-1]
      \arrow["s"', shift right=3, dashed, from=2-3, to=2-1]
      \arrow["{q_*(\ev_\partial)}", two heads, from=2-3, to=3-2]
    \end{tikzcd}
  \end{equation*}
  along with $\Id$-homotopies $H^{q,E}_s$ and $H^{q,E}_r$ over
  $q_*E \times_A q_*E$ respectively from $Js^{q,E}$ to $\id$ \emph{constant
    along $q_*(\refl)$} and $r^{q,E}J$ to $\id$ \emph{constant along $\refl$} in
  the sense of \Cref{def:const-res-htpy}.

  The maps $s^{q,E},r^{q,E}H^{q,E}_s,H^{q,E}_r$ are also chosen such that they
  are stable under pullback, in that for any pullback selected by the universe
  structure as follows
  \begin{equation*}
    \begin{tikzcd}[cramped]
      {B'} & B \\
      {A'} & A
      \arrow["", from=1-1, to=1-2]
      \arrow["{q'}"', two heads, from=1-1, to=2-1]
      \arrow["\lrcorner"{anchor=center, pos=0.15, scale=1.5}, draw=none, from=1-1, to=2-2]
      \arrow["q", two heads, from=1-2, to=2-2]
      \arrow["f"', from=2-1, to=2-2]
    \end{tikzcd}
  \end{equation*}
  one has
  \begin{align*}
    f^*(s^{q,E}) = s^{q',f^*E}
    &&
       f^*(r^{q,E}) = r^{q',f^*E}
    &&
       f^*(H_s^{q,E}) = H_s^{q',f^*E}
    &&
       f^*(H_r^{q,E}) = H_r^{q',f^*E}
  \end{align*}
\end{definition}
\begin{theorem}\label{thm:ptd-funext-tt}
  Let $\pi \colon \tMcU \to \McU$ be a universal map with a $\Pi$-structure
  $\Pi \colon \bP_\pi(\McU) \to \McU$ and an $\Id$-type structure
  $\Id \colon \tMcU \times_\McU \tMcU \to \McU$.

  Then, a $\pi$ supports a pointed $\Id$-functional extensionality structure (as
  in \Cref{def:funext}) precisely when it supports a type-theoretic pointed
  $(\Id,\Pi)$-functional extensionality structure (as in \Cref{def:funext-tt}).
\end{theorem}
\begin{proof}
  The $\Leftarrow$ direction is by
  \Cref{lem:retract-lift}\Cref{itm:retract-lift-strong} and the $\Rightarrow$
  direction is by \Cref{lem:path-htpy}, whose constructions are all stable under
  rebasing by pullback.
\end{proof}

\begin{remark}
  By a result of Voevodsky, it is known that book univalence implies the usual
  functional extensionality \cite[Theorem 4.9.4]{hottbook}, where one does not
  have the basepoint preservation condition of \Cref{def:funext}.
  However, it is unknown whether pointed univalence also implies pointed
  functional extensionality.
\end{remark}

We also get the following pointwise characterisation of pointed functional
extensionality that avoids mentioning the $\Id$-type.
So in some categories where the $\Id$-type construction is painful, we can use
it to avoid some amount of pain to provide pointed functional extensionality
structures if we just know how the pushforwards behave.
Namely, it reduces pointed functional extensionality to checking that
pushforwards along fibrations preserve the left class to fibrations.
The primary use case of it will be in \Cref{thm:gl-ptd-fe}, where we show
pointed functional extensionality in gluing categories.
\begin{lemma}\label{lem:ptd-funext-sdf}
  Let $\pi \colon \tMcU \to \McU$ be a universal map with a $\Pi$-structure
  $\Pi \colon \bP_\pi(\McU) \to \McU$ and an $\Id$-type structure
  $\Id \colon \tMcU \times_\McU \tMcU \to \McU$.

  Each pointed $\Id$-functional extensionality structure
  $\PtdFunExt \in \bP_\pi(\refl) \fracsquareslash{\McU} (\pi \times \McU)$
  gives rise to a family of maps as below indexed by $\pi$-fibrations
  $q \colon B \twoheadrightarrow A$ along with $s \colon Y \to E$ in the
  $\pi$-fibrant slice over $B$
  \begin{equation*}
    ((s \fracsquareslash{B} (B \times \pi)) \xrightarrow{\ptdfunext_{q,s}} (q_*s \fracsquareslash{A} (B \times \pi)))_{q,s}
  \end{equation*}
  subject to the condition that for each $f \colon A' \to A$, the following
  diagram commutes.
  \newsavebox{\EYbox}
  \begin{lrbox}{\EYbox}\scriptsize
    \begin{tikzcd}[sep=small, cramped]
      Y \ar[d, "s"] \\ E
    \end{tikzcd}
  \end{lrbox}
  \newsavebox{\BUbox}
  \begin{lrbox}{\BUbox}\scriptsize
    \begin{tikzcd}[sep=small, cramped]
      B \times \tMcU \ar[d] \\ B \times \McU
    \end{tikzcd}
  \end{lrbox}
  \newsavebox{\qEYbox}
  \begin{lrbox}{\qEYbox}\scriptsize
    \begin{tikzcd}[sep=small, cramped]
      q_*Y \ar[d, "q_*s"] \\ q_*E
    \end{tikzcd}
  \end{lrbox}
  \newsavebox{\AUbox}
  \begin{lrbox}{\AUbox}\scriptsize
    \begin{tikzcd}[sep=small, cramped]
      A \times \tMcU \ar[d] \\ A \times \McU
    \end{tikzcd}
  \end{lrbox}
  \newsavebox{\fEYbox}
  \begin{lrbox}{\fEYbox}\scriptsize
    \begin{tikzcd}[sep=small, cramped]
      f^*Y \ar[d, "f^*s"] \\ f^*E
    \end{tikzcd}
  \end{lrbox}
  \newsavebox{\fBUbox}
  \begin{lrbox}{\fBUbox}\scriptsize
    \begin{tikzcd}[sep=small, cramped]
      f^*B \times \tMcU \ar[d] \\ f^*B \times \McU
    \end{tikzcd}
  \end{lrbox}
  \newsavebox{\ApUbox}
  \begin{lrbox}{\ApUbox}\scriptsize
    \begin{tikzcd}[sep=small, cramped]
      A' \times \tMcU \ar[d] \\ A' \times \McU
    \end{tikzcd}
  \end{lrbox}
  \newsavebox{\fqEYbox}
  \begin{lrbox}{\fqEYbox}\scriptsize
    \begin{tikzcd}[sep=small, cramped]
      f^*q_*Y \ar[d, "f^*q_*s"] \\ f^*q_*E
    \end{tikzcd}
  \end{lrbox}
  \begin{equation*}
    \begin{tikzcd}[cramped, row sep=small, column sep=huge]
      \left(\usebox{\EYbox} \fracsquareslash{B} \usebox{\BUbox}\right)
      \ar[r, "{\ptdfunext_{q,s}}"]
      \ar[d]
      &
      \left(\usebox{\qEYbox} \fracsquareslash{A} \usebox{\AUbox}\right)
      \ar[d]
      \\
      \left(\usebox{\fEYbox} \fracsquareslash{B} \usebox{\fBUbox}\right)
      \ar[r, "{\ptdfunext_{f^*q,f^*s}}"]
      &
      \left(\usebox{\fqEYbox} \fracsquareslash{B} \usebox{\ApUbox}\right)
    \end{tikzcd}
  \end{equation*}
  Conversely, each such family restricts to a pointed $\Id$-functional
  extensionality structure.
\end{lemma}
\begin{proof}
  It is clear that if one has such a family $(\ptdfunext_{q,s})_{q,s}$ then by
  letting $q$ vary and taking $s = \refl$ so that one may apply each
  $\ptdfunext_{q,\refl}$ at the lifting structure provided by $\MsJ$-elimination
  gives some
  $\PtdFunExt \in \bP_\pi(\refl) \fracsquareslash{\McU} (\pi \times \McU)$.

  Conversely, suppose one has some
  $\PtdFunExt \in \bP_\pi(\refl) \fracsquareslash{\McU} (\pi \times \McU)$.
  Let $s \colon Y \to E$ be a map in the $\pi$-fibrant slice over $B$ along with
  a $\pi$-fibration $q \colon B \twoheadrightarrow A$ so that the goal is to
  define the map
  $\ptdfunext_{q,s} \colon (s \fracsquareslash{B} (B \times \pi)) \to (q_*s
  \fracsquareslash{A} (A \times \pi))$.
  One may assume that $(s \fracsquareslash{B} (B \times \pi))$ is not empty so
  that because $Y \twoheadrightarrow B$ is fibrant, $s$ admits a retraction
  $r \colon E \to Y$ over $B$.
  Thus, the idea is to show that duals to strong deformation retracts are
  preserved by pushforwards.

  Specifically, by \Cref{prop:dual-df-lift}\Cref{itm:dual-df-lift}, it suffices
  to construct from a family of maps
  $(H_f \colon f^*E \to \Id_{A'}(f^*E))_{f \colon A' \to A}$ such that the
  diagram on the left commutes for each $f$ and $g^*H_f = H_{gf}$ for each
  $g \colon A'' \to A'$, a similar compatible family of maps
  $(H'_f \colon f^*E \to \Id_{A'}(f^*E))_{f \colon A' \to A}$ such that the
  diagram on the right commutes.
  \begin{center}\small
    \begin{minipage}{0.45\linewidth}
      \begin{equation*}
        \begin{tikzcd}[cramped, sep=small]
          {f^*Y} & {f^*E} & {\Id_{A'}(f^*E)} \\
          {f^*E} && {f^*E \times_{A'} f^*E}
          \arrow["{f^*s}", from=1-1, to=1-2]
          \arrow["{f^*s}"', from=1-1, to=2-1]
          \arrow["\refl", from=1-2, to=1-3]
          \arrow["{\ev_\partial}", two heads, from=1-3, to=2-3]
          \arrow["{H_f}"{description}, dashed, from=2-1, to=1-3]
          \arrow["{(f^*s \cdot f^*r,\id)}"', from=2-1, to=2-3]
        \end{tikzcd}
      \end{equation*}
    \end{minipage}
    \begin{minipage}{0.45\linewidth}
      \begin{equation*}
        \begin{tikzcd}[cramped, sep=small]
          {(f^*q)_*Y} & {(f^*q)_*E} & {\Id_{A'}((f^*q)_*E)} \\
          {(f^*q)_*E} && {(f^*q)_*E \times_{A'} (f^*q)_*E}
          \arrow["{(f^*q)_*s}", from=1-1, to=1-2]
          \arrow["{(f^*q)_*s}"', from=1-1, to=2-1]
          \arrow["\refl", from=1-2, to=1-3]
          \arrow["{\ev_\partial}", two heads, from=1-3, to=2-3]
          \arrow["{H'_f}"{description}, dashed, from=2-1, to=1-3]
          \arrow["{((f^*q)_*s \cdot (f^*q)_*r,\id)}"', from=2-1, to=2-3]
        \end{tikzcd}
      \end{equation*}
    \end{minipage}
  \end{center}

  This is obtained by first taking the image of the diagram on the left under
  the pushforward along $f^*q \colon f^*B \to A'$ to obtain the left slanted
  rectangle as below then applying the pointed functional extensionality
  structure $\PtdFunExt \colon (f^*q)_*(\Id_{A'}(f^*E)) \to \Id_{A'}((f^*q)_*E)$
  like in the pointy triangle on the right below.
  \begin{equation*}\small
    \begin{tikzcd}[cramped, row sep=small, column sep=large]
      {(f^*q)_*Y} \\
      & {(f^*q)s} && {\Id_{A'}((f^*q)_*E)} \\
      {(f^*q)_*E} && {(f^*q)_*\Id_{A'}(f^*E)} \\
      \\
      && {(f^*q)_*E \times_{A'} (f^*q)_*E}
      \arrow["{(f^*q)_*s}", from=1-1, to=2-2]
      \arrow["{(f^*q)_*s}"', from=1-1, to=3-1]
      \arrow["\refl", from=2-2, to=2-4]
      \arrow["{(f^*q)_*\refl}"{description}, from=2-2, to=3-3]
      \arrow["{\ev_\partial}", from=2-4, to=5-3, two heads]
      \arrow["{(f^*q)_*H_f}"{description}, dashed, from=3-1, to=3-3]
      \arrow["{((f^*q)_*s \cdot (f^*q)_*r, \id)}"', from=3-1, to=5-3]
      \arrow["\PtdFunExt"{description}, dashed, from=3-3, to=2-4]
      \arrow["{(f^*q)_*\ev_\partial}"{description}, two heads, from=3-3, to=5-3]
    \end{tikzcd}
  \end{equation*}
  The necessary compatibility conditions are straightforwardly verified.
\end{proof}


\section{Artin–Wraith Gluing of Universe Categories}\label{sec:glue-universe}
The goal for the remaining of this paper is to show that pointed univalence and
function extensionality are preserved formation of inverse diagrams.
In order to do so, the main technical device will be Artin--Wraith gluing
categories, which we recall is defined as a comma category.
\begin{definition}
  The \emph{Artin--Wraith gluing category} of a functor
  $M \colon \bC_0 \to \bC_1$ is the comma category $\bC_1 \downarrow M$.
\end{definition}

By \cite[Theorem 4.5]{shu-reedy}, if $\McI$ is an inverse category of degree
$n$ then it arises as the gluing category along the $n$-th matching object
functor.
The goal of this section is to work abstractly with gluing categories with an
aim towards applying them to inverse diagram categories in
\Cref{sec:inverse-universe}.

%
%
For this purpose, in this section we fix a lex map
\begin{equation*}
  \bC_0 \xrightarrow{M} \bC_1
\end{equation*}
evoking imagery of the matching object functor.
We also assume that $\bC_0$ and $\bC_1$ are finitely complete categories
respectively equipped with universal maps
\begin{align*}
  \pi_0 \colon \tMcU_0 \to \McU_0 \in \bC_0 && \pi_1 \colon \tMcU_1 \to \McU_1 \in \bC_1
\end{align*}
respectively with internal universes
\begin{align*}
  \pi_0^\ds \colon \tMcU_0^\ds \to \McU_0^\ds \in \bC_0 && \pi_1^\ds \colon \tMcU_1^\ds \to \McU_1^\ds \in \bC_1
\end{align*}

Moreover, we work under the following assumptions.
\begin{assumption}\label{asm:mrq}
  We suppose that $M$ has the following ``right Quillen'' properties and
  $\pi_i,\pi_i^\ds$ carry the structure of intensional type theory:
  \begin{enumerate}
    \item\label{itm:mrq-fib} $M \colon \bC_0 \to \bC_1$ takes $\pi_0$-fibrations to
    $\pi_1$-fibrations.
    \item\label{itm:mrq-tc} One has a family of maps as below indexed by objects $B_0 \in \bC_0$
    and maps $s_0 \colon Y_0 \to E_0$ in the $\pi$-fibrant slice over $B_0$
    \begin{equation*}
      \left((s_0 \fracsquareslash{B_0} (B_0 \times \pi_0))
      \xrightarrow{\quad}
      (Ms_0 \fracsquareslash{MB_0} (MB_0 \times \pi_1))\right)_{s_0,B_0}
    \end{equation*}
    subject to the condition that for each $f_0 \colon B_0' \to B_0$, the following
    diagram commutes.
    \newsavebox{\szbox}
    \begin{lrbox}{\szbox}\scriptsize\begin{tikzcd}[sep=small, cramped] Y_0 \ar[d, "s_0"] \\ E_0 \end{tikzcd}\end{lrbox}
    \newsavebox{\uzbox}
    \begin{lrbox}{\uzbox}\scriptsize\begin{tikzcd}[sep=small, cramped] B_0 \times \tMcU_0 \ar[d, ""] \\ B_0 \times \McU_0 \end{tikzcd}\end{lrbox}
    \newsavebox{\mszbox}
    \begin{lrbox}{\mszbox}\scriptsize\begin{tikzcd}[sep=small, cramped] MY_0 \ar[d, "Ms_0"] \\ ME_0 \end{tikzcd}\end{lrbox}
    \newsavebox{\uobox}
    \begin{lrbox}{\uobox}\scriptsize\begin{tikzcd}[sep=small, cramped] MB_0 \times \tMcU_1 \ar[d, ""] \\ MB_0 \times \McU_1 \end{tikzcd}\end{lrbox}
    \newsavebox{\szpbox}
    \begin{lrbox}{\szpbox}\scriptsize\begin{tikzcd}[sep=small, cramped] f_0^*Y_0 \ar[d, "{f_0^*s_0}"] \\ f_0^*E_0 \end{tikzcd}\end{lrbox}
    \newsavebox{\uzpbox}
    \begin{lrbox}{\uzpbox}\scriptsize\begin{tikzcd}[sep=small, cramped] B_0' \times \tMcU_0 \ar[d, ""] \\ B_0' \times \McU_0 \end{tikzcd}\end{lrbox}
    \newsavebox{\mszpbox}
    \begin{lrbox}{\mszpbox}\scriptsize\begin{tikzcd}[sep=small, cramped] M(f_0^*Y_0) \ar[d, "{M(f_0^*s_0)}"] \\ M(f_0^*E_0) \end{tikzcd}\end{lrbox}
    \newsavebox{\uopbox}
    \begin{lrbox}{\uopbox}\scriptsize\begin{tikzcd}[sep=small, cramped] MB_0' \times \tMcU_1 \ar[d, ""] \\ MB_0' \times \McU_1 \end{tikzcd}\end{lrbox}
    \begin{equation*}
      \begin{tikzcd}[cramped, row sep=small, column sep=huge]
        \left(\usebox{\szbox} \fracsquareslash{B_0} \usebox{\uzbox}\right)
        \ar[r, "{}"]
        \ar[d, "{\text{\cite[Construction 3.2]{struct-lift}}}"]
        &
        \left(\usebox{\mszbox} \fracsquareslash{MB_0} \usebox{\uobox}\right)
        \ar[d, "{\text{\cite[Construction 3.2]{struct-lift}}}"]
        \\
        \left(\usebox{\szpbox} \fracsquareslash{B_0'} \usebox{\uzpbox}\right)
        \ar[r, "{}"]
        &
        \left(\usebox{\mszpbox} \fracsquareslash{MB_0'} \usebox{\uopbox}\right)
      \end{tikzcd}
    \end{equation*}
    \item\label{itm:mrq-pi} $\pi_i$ and $\pi_i^\ds$ are equipped with
    $\Unit,\Sigma,\Pi,\Id$-type structures denoted
    $\Unit_i,\Sigma_i,\Pi_i,\Id_i$ and
    $\Unit_i^\ds,\Sigma_i^\ds,\Pi_i^\ds,\Id_i^\ds$ respectively.
  \end{enumerate}
\end{assumption}

We also translate the following definitions from the theory of Reedy categories
into the setting of gluing categories.
\begin{definition}
  For each object $X \in \Gl(M)$, we denote the its image under the projection
  maps $\Gl(M) \to \bC_i$ as $X_i$ so that $X$ is a map
  \begin{equation*}
    X \colon X_1 \to MX_0 \in \bC_1
  \end{equation*}
  We say that $X_i$ is the \emph{$i$-th component} of $X$ and that the codomain
  object $MX_0$ is the \emph{absolute matching object} of $X$.
  When viewed as a map of $\bC_1$, we also say that $X$ is its own
  \emph{absolute matching map}.

  For a map $f \colon Y \to X$, the associated \emph{relative matching map}
  $\Whm(f) \colon Y_1 \to \WhM(f)$ into the \emph{relative matching object}
  $\WhM(f)$, sometimes denoted by $\Whm_X(Y) \colon Y_1 \to \WhM_X(Y)$ when $f$
  is clear from context, is the connecting map into the pullback, as labelled in
  the following diagram.
  \begin{equation*}
    \begin{tikzcd}[cramped, column sep=small]
      {Y_1} \\
      & {\WhM_X(Y)} & {MY_0} \\
      & {X_1} & {MX_0}
      \arrow["{\Whm_X(Y)}"{description}, from=1-1, to=2-2, dotted]
      \arrow["Y", curve={height=-12pt}, from=1-1, to=2-3]
      \arrow["{f_1}"', curve={height=12pt}, from=1-1, to=3-2, dashed]
      \arrow[from=2-2, to=2-3]
      \arrow[from=2-2, to=3-2, dashed]
      \arrow["\lrcorner"{anchor=center, pos=0.15, scale=1.5}, draw=none, from=2-2, to=3-3]
      \arrow["{Mf_0}", from=2-3, to=3-3, dashed]
      \arrow["X"', from=3-2, to=3-3]
    \end{tikzcd}
  \end{equation*}
  We also denote sometimes objects using solid lines when viewed as maps in
  $\bC_1$ and components of maps using dashed lines.
\end{definition}

\begin{definition}\label{def:reedy-fib}
  A map $E \to B \in \Gl(M)$ is a \emph{Reedy fibration} when its 0-component
  $E_0 \to B_0 \in \bC_0$ and the relative matching map
  $\Whm_B(E) \colon E_1 \to \WhM_BE \in \bC_1$ are $\pi_0$- and
  $\pi_1$-fibrations respectively.

  If $E_0 \to B_0$ and $\Whm_B(E) \colon E_1 \to \WhM_BE$ are $\pi_0^\ds$- and
  $\pi_1^\ds$-fibrations respectively then it is an \emph{internal Reedy
    fibration}.

  We also say that an object $X \colon X_1 \to MX_0 \in \Gl(M)$ is a (internal)
  Reedy fibrant object when $X \to 1$ is a (internal) Reedy fibration.
\end{definition}
Immediate from definition, we see that a (internal) Reedy fibrant object
$X \colon X_1 \to MX_0$ is exactly the same as a $\pi_1$-fibration with codomain
object in the image of $M$.

In the rest of the section:
\begin{enumerate}
  \item We show in \Cref{subsec:glue-universe-univ} that
  $\Gl(M)$ yet again has a universal
  Reedy fibration equipped with a universe of internal Reedy fibrations.
  \item Then we show that the $\Unit,\Sigma,\Pi,\Id$-structures from the
  universes of $\bC_0$ and $\bC_1$ give rise to corresponding structures on
  $\Gl(M)$ in
  \Cref{subsec:glue-universe-univ,subsec:glue-universe-unit,subsec:glue-universe-sigma,subsec:glue-universe-pi,subsec:glue-universe-id}.
  \item Finally, we show in
  \Cref{subsec:glue-universe-funext,subsec:glue-universe-univalence} the main
  technical results that the pointed functional extensionality and pointed
  univalence structures of $\bC_0$ and $\bC_1$ also induce corresponding
  structures on $\Gl(M)$.
\end{enumerate}

\subsection{Universal Reedy Fibrations}\label{subsec:glue-universe-univ}
The goal of this part is to construct a universal Reedy fibration equipped with
an internal universal Reedy fibration.

These universal Reedy fibrations are constructed as follows.
\begin{construction}\label{constr:urf}
  We construct the map $\tau \colon \tMcV \to \McV \in \Gl(M)$ whose
  0-component is $\pi_0$ and whose 1-component is $\GenComp(\pi_1,M\pi_0)$.
  \begin{equation*}
    \begin{tikzcd}[cramped, column sep=small]
      {\tMcV_1 = \ev^*\tMcU_1} & {\tMcU_1} \\
      {\WhM_\McV(\tMcV) = (M\pi_0)^*\bP_{M\pi_0}(\McU_1)} & {\McU_1} \\
      {\McV_1 = \bP_{M\pi_0}(\McU_1)} & {M\tMcV_0 = M\tMcU_0} \\
      & {M\McV_0 = M\McU_0}
      \arrow[from=1-1, to=1-2]
      \arrow["{\Whm_\McV(\tMcV) = \ev^*\pi_1}"', two heads, from=1-1, to=2-1]
      \arrow["{\pi_1}", two heads, from=1-2, to=2-2]
      \arrow["\ev"', from=2-1, to=2-2]
      \arrow[two heads, from=2-1, to=3-1, dashed]
      \arrow[from=2-1, to=3-2]
      \arrow[from=3-1, to=4-2]
      \arrow["{M\tau_0 = M\pi_0}", from=3-2, to=4-2, dashed]
      \arrow["\lrcorner"{anchor=center, pos=0.15, scale=1.5}, draw=none, from=1-1, to=2-2]
      \arrow["\lrcorner"{anchor=center, pos=0.15, scale=1.5}, draw=none, from=2-1, to=4-2]
    \end{tikzcd}
  \end{equation*}

  Applying the same construction on the internal universes, we get the map
  $\tau^\ds \colon \tMcV^\ds \to \tMcV^\ds \in \Gl(M)$.
\end{construction}

We must check that the above construction indeed gives a universal Reedy
fibration.
We start with showing that both $\tau$ and $\tau^\ds$ are universal maps.
\begin{lemma}\label{lem:urf-expn}
  The maps $\tau \colon \tMcV \to \McV \in \Gl(M)$ and
  $\tau^\ds \colon \tMcV^\ds \to \McV^\ds \in \Gl(M)$ from \Cref{constr:urf} both has
  the structure of a universal map.
\end{lemma}
\begin{proof}
  The proofs for $\tau$ and $\tau_0$ are identical, so we show the result only
  for $\tau$.
  To show it has the structure of a universal map, we need to choose a specific
  right adjoint of the post-composition functor
  $\tau_! \colon \sfrac{\Gl(M)}{\tMcV} \to \sfrac{\Gl(M)}{\McV}$ and show that
  $\tau$ itself is exponentiable.
  We show these two goals sequentially.

  Because $M\pi_0$ is a $\pi_1$-fibration, $\GenComp(\pi_1,M\pi_0)$ is a
  pullback of $\GenComp(\pi_1,\pi_1)$ by \Cref{lem:gen-comp}.
  By the $\Sigma$-type structure on $\pi_1$, one has that
  $\GenComp(\pi_1,\pi_1)$ is $\pi_1$-fibrant.
  Therefore, the universal structures on $\pi_0$ and $\pi_1$ respectively give
  rise to a choice of pullback along the 0-component $\pi_0$ and 1-component
  $\GenComp(\pi_1,M\pi_0)$ of $\tau$.
  Combined together, this gives rise to a choice of pullbacks along $\tau$,
  since pullbacks in $\Gl(M)$ is calculated pointwise.

  To see that $\tau$ is exponentiable, by \cite[Theorem 2.19]{fkl24}, it
  suffices to show that the three maps
  \begin{align*}
    \pi_0 && M\pi_0 && \GenComp(\pi_1,M\pi_0)
  \end{align*}
  are exponentiable.
  By \Cref{asm:mrq}, $M\pi_0$ is $\pi_1$-fibrant and $\GenComp(\pi_1,M\pi_0)$
  was observed to be $\pi_1$-fibrant above.
  Exponentiable maps are closed under pullbacks by \cite[Corollary 1.4]{nie82},
  so the result follows by the exponentiability of $\pi_0$ and $\pi_1$.
\end{proof}

Next, we check that they indeed have the claimed classification properties.
\begin{proposition}\label{prop:urf-up}
  The map $\tau \colon \tMcV \to \McV \in \Gl(M)$ (respectively,
  $\tau^\ds \colon \tMcV^\ds \to \McV^\ds$) from \Cref{constr:urf} is a
  universal Reedy fibration (respectively, internal Reedy fibration).
  This means that if $E \to B \in \Gl(M)$ is a Reedy fibration (respectively,
  internal Reedy fibration) if and only if $E \to B$ is a $\tau$-fibration
  (respectively, $\tau^\ds$-fibration).
\end{proposition}
\begin{proof}
  The proofs for the $\tau$ case and the $\tau^\ds$ case are identical, so we
  only show the $\tau$ case.

  The universe structure on $\tau$ is already provided by \Cref{lem:urf-expn}.
  It is also clear from construction that $\tau$ is a Reedy fibration so all
  $\tau$-fibrations are also Reedy fibrations.
  It remains to show that all Reedy fibrations $p \colon E \to B$ can be
  equipped with the structure of a $\tau$-fibration.

  We do so with reference to the diagram at the end of the proof.
  By the definition of a Reedy fibration, $p_0 \colon E_0 \to B_0$ occurs the
  pullback of a map $\ceil{E_0} \colon B_0 \to \McU_0$ in $\bC_0$.
  Therefore, $\WhM_BE \to B_1$ occurs as the pullback of
  $M\pi_0 \colon M\tMcU_0 \to M\McU_0$ along the map
  \begin{equation*}
    B_1 \xrightarrow{B} MB_0 \xrightarrow{M\ceil{E_0}} M\McU_0
  \end{equation*}
  By the $\pi_1$-fibrancy of $\Whm_BE \colon E_1 \to \WhM_BE$, one has a map
  $\ceil{E_1} \colon \WhM_BE \to \McU_1$ along which $\Whm_BE$ occurs as the
  pullback of $\pi_1$.
  Using \Cref{lem:gen-comp}, one obtains a map
  $B^*(M\ceil{E_0}).\ceil{E_1} \colon B_1 \to \bP_{M\pi_0}(\McU_1)$
  such that pulling back $\GenComp(\pi_1,M\pi_0)$ along it gives back the 1-component of $p$
  \begin{equation*}
    p_1 \colon E_1 \xrightarrow{\Whm_BE} \WhM_BE \to B_1
  \end{equation*}
  Furthermore, the composite
  \begin{equation*}
   B_1 \xrightarrow{B^*(M\ceil{E_0}).\ceil{E_1}} \bP_{M\pi_0}(\McU_1) \to M\McU_0
  \end{equation*}
  is the composite $B_1 \xrightarrow{B} MB_0 \xrightarrow{M\ceil{E_0}} M\McU_0$.

  Therefore, we have the following diagram in $\Gl(M)$, where the front and back
  faces are pullbacks, as required.
  \begin{equation*}
    \begin{tikzcd}[cramped, row sep=small]
      {E_1} && {\ev^*\tMcU_1} & {\tMcU_1} \\
      {\WhM_BE} && {(M\pi_0)^*\bP_{M\pi_0}(\McU_1)} & {\McU_1} \\
      & {ME_0} && {M\tMcU_0} \\
      {B_1} && {\bP_{M\pi_0}(\McU_1)} \\
      & {MB_0} && {M\McU_0}
      \arrow[from=1-1, to=1-3]
      \arrow[from=1-1, to=2-1]
      \arrow[from=1-3, to=1-4]
      \arrow[from=1-3, to=2-3]
      \arrow[from=1-4, to=2-4]
      \arrow[from=2-1, to=2-3]
      \arrow[from=2-1, to=3-2]
      \arrow[from=2-1, to=4-1]
      \arrow["\ev"', from=2-3, to=2-4]
      \arrow[from=2-3, to=3-4]
      \arrow[from=2-3, to=4-3]
      \arrow[from=3-2, to=3-4]
      \arrow[from=3-2, to=5-2]
      \arrow[from=3-4, to=5-4]
      \arrow["{B^*(M\ceil{E_0}).\ceil{E_1}}"', from=4-1, to=4-3]
      \arrow["B"', from=4-1, to=5-2]
      \arrow[from=4-3, to=5-4]
      \arrow["{M\ceil{E_0}}"', from=5-2, to=5-4]
      \arrow["\lrcorner"{anchor=center, pos=0.15, scale=1.5, rotate=0}, draw=none, from=1-1, to=2-3]
      \arrow["\lrcorner"{anchor=center, pos=0.15, scale=1.5}, draw=none, from=1-3, to=2-4]
      \arrow["\lrcorner"{anchor=center, pos=0.15, scale=1.5}, draw=none, from=2-1, to=4-3]
      \arrow["\lrcorner"{anchor=center, pos=0.15, scale=1.5, rotate=-45}, draw=none, from=2-1, to=5-2]
      \arrow["\lrcorner"{anchor=center, pos=0.15, scale=1.5}, draw=none, from=2-3, to=5-4]
      \arrow["\lrcorner"{anchor=center, pos=0.15, scale=1.5}, draw=none, from=3-2, to=5-4]
    \end{tikzcd}
  \end{equation*}
\end{proof}

We have now established that $\tau$ and $\tau^\ds$ are both universal maps with
the desired classification properties.
We finish this part by checking that $\tau^\ds$ is an internal universe of
$\tau$.
\begin{theorem}\label{thm:glue-universe-univ}
  From \Cref{constr:urf}, the universe of Reedy fibrations
  $\tau \colon \tMcV \to \McV \in \Gl(M)$ is equipped with an internal universe
  $\tau^\ds \colon \tMcV^\ds \to \McV^\ds \in \Gl(M)$ of internal Reedy
  fibrations.
\end{theorem}
\begin{proof}
  Universality of $\tau$ and $\tau^\ds$ from \Cref{prop:urf-up}.
  It remains to check that $\McV^\ds$ (as an object) and $\tau^\ds$ (as a map) are $\tau$-fibrant.
  By \Cref{prop:urf-up}, this is the same as checking they are both Reedy
  fibrant.

  Reedy fibrancy of $\McV^\ds$ follows because $\McU_1^\ds$ and $M\pi_0^\ds$ are
  $\pi_1$-fibrant and $\pi_1$-fibrant maps are closed under pushforwards along
  $\pi_1$-fibrant maps by the $\Pi$-type structure of $\pi_1$.
  Reedy fibrancy of $\tau^\ds$ follows by $\pi_i$-fibrancy of $\pi_i^\ds$.
\end{proof}


\subsection{Unit-types}\label{subsec:glue-universe-unit}
The unit type of $\Gl(M)$ are constructed straightforwardly.
\begin{proposition}\label{prop:glue-universe-unit}
  $\tau \colon \tMcV \to \McV \in \Gl(M)$ has a $\Unit$-type structure preserved
  by both projections $\Gl(M) \to \bC_i$ for $i=0,1$ and likewise does
  $\tau^\ds$.
\end{proposition}
\begin{proof}
  Because $M$ is lex, the $\Unit$-type map $1 \to \McV$ is given by the map
  where the $i$-th component are $\Unit_i \colon 1 \to \McU_i$ for $\tau$ and
  $\Unit^\ds_i \colon 1 \to \McU_i$ for $\tau^\ds$.
\end{proof}


\subsection{$\Sigma$-types}\label{subsec:glue-universe-sigma}
For $\Sigma$-types, we need to show that if $\pi_i$-fibrations are closed under
composition and so are $\tau$-fibrations and likewise for $\pi_i^\ds$-fibrations
and $\tau^\ds$-fibrations.
\begin{lemma}\label{lem:reedy-fib-comp}
  The Reedy fibrations and internal Reedy fibrations of $\Gl(M)$ are closed
  under composition.
\end{lemma}
\begin{proof}
  We only show the case for Reedy fibrations as the internal case is identical.

  Let $Y \twoheadrightarrow E \twoheadrightarrow B$ be a composable pair of
  Reedy fibration.
  Then, its 0-component
  $Y_0 \twoheadrightarrow E_0 \twoheadrightarrow B_0 \in \bC_0$ is a
  $\pi_0$-fibration by the $\Sigma$-type structure on $\pi_0$ .
  By taking iterated pullbacks, we can see that the relative matching map
  $\Whm_BY$ is the composite of the relative matching map $\Whm_EY$ with a
  pullback of the relative matching map $\Whm_BE$.
  \begin{equation*}
    \begin{tikzcd}[cramped, row sep=small, column sep=large]
      {Y_1} & {\WhM_EY} & {\WhM_BY} & {MY_0} \\
      & {E_1} & {\WhM_BE} & {ME_0} \\
      && {B_1} & {MB_0}
      \arrow["{\Whm_EY}"', two heads, from=1-1, to=1-2]
      \arrow["{\Whm_BY}", curve={height=-12pt}, from=1-1, to=1-3]
      \arrow[from=1-1, to=2-2]
      \arrow[two heads, from=1-2, to=1-3]
      \arrow[from=1-2, to=2-2]
      \arrow[from=1-3, to=1-4]
      \arrow[from=1-3, to=2-3]
      \arrow[from=1-4, to=2-4]
      \arrow["{\Whm_BE}", two heads, from=2-2, to=2-3]
      \arrow[from=2-2, to=3-3]
      \arrow[from=2-3, to=2-4]
      \arrow[from=2-3, to=3-3]
      \arrow["\lrcorner"{anchor=center, pos=0.15, scale=1.5}, draw=none, from=1-2, to=2-3]
      \arrow["\lrcorner"{anchor=center, pos=0.15, scale=1.5}, draw=none, from=1-3, to=2-4]
      \arrow["\lrcorner"{anchor=center, pos=0.15, scale=1.5}, draw=none, from=2-3, to=3-4]
      \arrow[from=2-4, to=3-4]
      \arrow[from=3-3, to=3-4]
    \end{tikzcd}
  \end{equation*}
  By pullback stability of $\pi_1$-fibrations and the $\Sigma$-type structure on
  $\pi_1$, the result follows.
\end{proof}

\begin{proposition}\label{prop:glue-universe-sigma}
  Both $\tau \colon \tMcV \to \McV \in \Gl(M)$ and
  $\tau^\ds \colon \tMcV^\ds \to \McV^\ds \in \Gl(M)$ admit a $\Sigma$-type
  structure preserved by the 0-component projection $\Gl(M) \to \bC_0$.
\end{proposition}
\begin{proof}
  By construction, $\GenComp(\tau,\tau)$ and $\GenComp(\tau^\ds,\tau^\ds)$ are
  the composite of two $\tau$- and $\tau^\ds$-fibrations respectively, which are
  Reedy fibrations by \Cref{prop:urf-up}.
  Thus, the result follows by \Cref{lem:reedy-fib-comp}.
\end{proof}


\subsection{$\Pi$-types}\label{subsec:glue-universe-pi}
For $\Pi$-types, we need to show that pushforwards of (internal) Reedy
fibrations along (internal) Reedy fibrations are again (internal) Reedy
fibrations.

\begin{lemma}\label{lem:reedy-fib-pshfw}
  (Internal) Reedy fibrations are closed under pushforwards along (internal)
  Reedy fibrations.
\end{lemma}
\begin{proof}
  Once again, the internal case is identical to the case for ambient Reedy
  fibrations, so we just handle the case of ambient Reedy fibrations.

  Let $E \twoheadrightarrow B \in \Gl(M)$ be a Reedy fibrant object over $B$ and
  $q \colon B \twoheadrightarrow A \in \Gl(M)$ be a Reedy fibration.
  Because the map $\Gl(M) \to \bC_0$ preserves pushforwards by \cite[Theorem
  2.19]{fkl24}, one has that $(q_*E)_0 = (q_0)_*E_0$, which is $\pi_0$-fibrant
  over $B_0$ by the $\Pi$-type structure on $\pi_0$.

  For the 1-component, \cite[Construction 2.12]{fkl24} shows that
  $(q_*E)_1 \to A_1$ as an object over $A_1$ is constructed as a pullback along
  a map $\WhM_A(q_*E) \to (q_1)_*(\WhM_BE)$ of the pushforward along $q_1$ of the
  relative matching map $\Whm_BE \colon E_1 \to \WhM_BE \in \sfrac{\bC_1}{B_1}$, as follows.
  \begin{equation*}
    \begin{tikzcd}[cramped, sep=small]
      {E_1} & {(q_*E)_1} & {(q_1)_*E_1} \\
      {\WhM_BE} & {\WhM_A(q_*E)} & {(q_1)_*(\WhM_BE)} \\
      {B_1} & {A_1}
      \arrow["{\Whm_BE}"', two heads, from=1-1, to=2-1]
      \arrow[from=1-2, to=1-3]
      \arrow["{\Whm_A(q_*E)}"', two heads, from=1-2, to=2-2]
      \arrow["\lrcorner"{anchor=center, pos=0.15, scale=1.5}, draw=none, from=1-2, to=2-3]
      \arrow["{(q_1)_*(\Whm_BE)}", two heads, from=1-3, to=2-3]
      \arrow[two heads, from=2-1, to=3-1]
      \arrow[from=2-2, to=2-3]
      \arrow[two heads, from=2-2, to=3-2]
      \arrow[two heads, from=2-3, to=3-2]
      \arrow["{q_1}"', two heads, from=3-1, to=3-2]
    \end{tikzcd}
  \end{equation*}
  By the distributivity law for pushforwards (\cite[Lemma 2.3]{inv-psfw} or
  \cite[Paragraph 2.3]{gk13}), it thus follows that the map $(p_1)_*(\Whm_BE)$
  is a $\pi_1$-fibration and thus that $(q_*E)_1 \to \WhM_pE$ is a
  $\pi_1$-fibration.
\end{proof}

\begin{proposition}\label{prop:glue-universe-pi}
  If $\pi_i$ has $\Pi$-type structures for $i=0,1$ then
  $\tau \colon \tMcV \to \McV \in \Gl(M)$ admits a $\Pi$-type structure
  preserved by the 0-component projection $\Gl(M) \to \bC_0$.
\end{proposition}
\begin{proof}
  Immediate by \Cref{lem:reedy-fib-pshfw,lem:gen-comp-psfw}.
\end{proof}


\subsection{Id-types}\label{subsec:glue-universe-id}
For $\Id$-types, we need to show that diagonals of (internal) Reedy fibrations
admit a (pointwise trivial cofibration, (internal) Reedy
fibration)-factorisation.
The construction is slight pain but is a direct generalisation of
\cite[Proposition 4.7]{kl21}.

\subsubsection{Id-type for composable pair of fibrations}
For now we first temporarily exit ourselves from the setting of gluing
categories and recall the $\transport$-factorisation of the $\Id$-type for the
$\Sigma$-type in general universe category models of intensional type theory.
\begin{construction}\label{constr:sigma-fact-fib}
  Let $\bC$ be any finitely complete category equipped with a universal map
  $\pi \colon \tMcU \to \McU$ which has an $\Id$-type structure
  $\Id \colon \tMcU \times_\McU \tMcU \to \McU$ and a $\Pi$-type structure.

  Suppose one has a composable pair of $\pi$-fibrations
  $E_1 \twoheadrightarrow E_0 \twoheadrightarrow B$ along with a factorisation
  of the diagonal $E_0 \to E_0 \times_B E_0$ as
  \begin{equation*}
    \begin{tikzcd}[cramped]
      {E_0} & P & {E_0 \times_B E_0}
      \arrow["r", from=1-1, to=1-2]
      \arrow["p", two heads, from=1-2, to=1-3]
    \end{tikzcd} \in \sfrac{\bC}{B}
  \end{equation*}
  where $r$ is equipped with a lifting structure $r \fracsquareslash{B} E \times \pi$.
  The goal is to construct a fibrant object $\ol{\Id}_B(E_0,P,E_1)$ over the
  pullback $P \times_{E_0 \times_B E_0} (E_1 \times_B E_1)$ as follows.
  \begin{equation*}
    \begin{tikzcd}[cramped, column sep=small]
      {\ol{\Id}_B(E_0,P,E_1)} \\
      {P \times_{E_0 \times_B E_0} (E_1 \times_B E_1)} & {E_1 \times_B E_1} \\
      P & {E_0 \times_B E_0}
      \arrow[two heads, from=1-1, to=2-1]
      \arrow[two heads, from=2-1, to=2-2]
      \arrow[two heads, from=2-1, to=3-1]
      \arrow["\lrcorner"{anchor=center, pos=0.15, scale=1.5}, draw=none, from=2-1, to=3-2]
      \arrow[two heads, from=2-2, to=3-2]
      \arrow[two heads, from=3-1, to=3-2]
    \end{tikzcd}
  \end{equation*}

  Denote by
  $p_0,p_1 \colon P \xrightarrow{p} E_0 \times_B E_0 \rightrightarrows E_0$ the
  two projections so that by the lifting structure on $r$, one obtains a
  transport map $\transport \colon p_0^*E_1 \to p_1^*E_1$ as on the left as follows.
  \begin{center}
    \begin{minipage}{0.45\linewidth}
      \begin{equation*}
        \begin{tikzcd}[cramped, column sep=small]
          {p_0^*E_1} \\
          && {p_1^*E_1} && {E_1} \\
          & P & {E_0 \times_B E_0} & {E_0}
          \arrow[from=1-1, to=2-5]
          \arrow[from=1-1, to=3-2]
          \arrow[from=2-3, to=2-5]
          \arrow[two heads, from=2-3, to=3-2]
          \arrow[two heads, from=2-5, to=3-4]
          \arrow[two heads, from=3-2, to=3-3]
          \arrow[shift left, two heads, from=3-3, to=3-4]
          \arrow[shift right, two heads, from=3-3, to=3-4]
          \arrow["\lrcorner"{anchor=center, pos=0.15, scale=1.5}, draw=none, from=1-1, to=3-4]
          \arrow["\lrcorner"{anchor=center, pos=0.15, scale=1.5}, draw=none, from=2-3, to=3-4]
          \arrow["\transport"{description}, from=1-1, to=2-3]
        \end{tikzcd}
      \end{equation*}
    \end{minipage}
    \begin{minipage}{0.45\linewidth}
      \begin{equation*}
        \begin{tikzcd}[cramped, column sep=small]
          {P \times_{E_0 \times_B E_0} (E_1 \times_B E_1)} & {E_1 \times_B E_1} & {E_1} \\
          P & {E_0 \times_B E_0} & {E_0}
          \arrow[two heads, from=1-1, to=1-2]
          \arrow[two heads, from=1-1, to=2-1]
          \arrow["\lrcorner"{anchor=center, pos=0.15, scale=1.5}, draw=none, from=1-1, to=2-2]
          \arrow[from=1-2, to=1-3]
          \arrow[two heads, from=1-2, to=2-2]
          \arrow[two heads, from=1-3, to=2-3]
          \arrow[two heads, from=2-1, to=2-2]
          \arrow[shift left, two heads, from=2-2, to=2-3]
          \arrow[shift right, two heads, from=2-2, to=2-3]
        \end{tikzcd}
      \end{equation*}
    \end{minipage}
  \end{center}
  But note that one also has the commutativity squares on the right.
  Therefore, the universal property of the pullback gives
  maps
  \begin{equation*}
    P \times_{E_0 \times_B E_0} (E_1 \times_B E_1)
    \xrightarrow{\quad}
    p_i^*E_1
    \in \sfrac{\bC}{P}
  \end{equation*}
  where $i=0,1$.
  The 1-component
  $P \times_{E_1 \times_B E_1} (E_1 \times_B E_1) \xrightarrow{\quad} p_1^*E_1$
  together with the composition of the 0-component with the $\transport$ map
  $P \times_{E_0 \times_B E_0} (E_1 \times_B E_1) \xrightarrow{\quad} p_0^*E_1
  \xrightarrow{\transport} p_1^*E_1$ then gives a map
  $\begin{tikzcd}[cramped]
    {P \times_{E_0 \times_B E_0} (E_1 \times_B E_1)} & {p_1^*E_1 \times_P p_1^*E_1}
    \arrow[dashed, from=1-1, to=1-2]
  \end{tikzcd} \in \sfrac{\bC}{P}$ as follows.
  \begin{equation*}\label{eqn:transport-pair}\tag{$\transport\text{-}\textsf{pair}$}
    \begin{tikzcd}[cramped, row sep=small]
      & {p_0^*E_1} & {p_1^*E_1} \\
      {P \times_{E_0 \times_B E_0} (E_1 \times_B E_1)} & {p_1^*E_1 \times_P p_1^*E_1} \\
      && {p_1^*E_1}
      \arrow["\transport", from=1-2, to=1-3]
      \arrow[dotted, from=2-1, to=1-2]
      \arrow[dashed, from=2-1, to=2-2]
      \arrow[dotted, from=2-1, to=3-3]
      \arrow[from=2-2, to=1-3]
      \arrow[from=2-2, to=3-3]
    \end{tikzcd} \in \sfrac{\bC}{P}
  \end{equation*}

  Therefore, the required fibration
  $\ol{\Id}_B(E_0,P,E_1) \xrightarrow{\quad} P \times_{E_0 \times_B E_0} (E_1
  \times_B E_1)$ is constructed by pulling back along the map
  $\begin{tikzcd}[cramped]
    {P \times_{E_0 \times_B E_0} (E_1 \times_B E_1)} & {p_1^*E_1 \times_P p_1^*E_1}
    \arrow[dashed, from=1-1, to=1-2]
  \end{tikzcd}$ constructed above the map
  $\Id_P(p_1^*E_1) \twoheadrightarrow p_1^*E_1 \times_P p_1^*E_1$.
  \begin{equation*}
    \begin{tikzcd}[cramped, column sep=small]
      {\ol{\Id}_B(E_0,P,E_1)} & {\Id_P(p_1^*E_1)} \\
      {P \times_{E_0 \times_B E_0} (E_1 \times_B E_1)} & {p_1^*E_1 \times_P p_1^*E_1}
      \arrow[from=1-1, to=1-2]
      \arrow[from=1-1, to=2-1, two heads]
      \arrow["\lrcorner"{anchor=center, pos=0.15, scale=1.5}, draw=none, from=1-1, to=2-2]
      \arrow["{\ev_\partial}", two heads, from=1-2, to=2-2]
      \arrow[dashed, from=2-1, to=2-2]
    \end{tikzcd} \in \sfrac{\bC}{P}
  \end{equation*}
\end{construction}

The goal is to show that the composite
$\ol{\Id}_B(E_0,P,E_1) \twoheadrightarrow P \times_{E_0 \times_B E_0} (E_1
\times_B E_1) \twoheadrightarrow E_1 \times_B E_1$ is the fibration part of the
diagonal factorisation for $E_1$ over $P$.
To do so, we need to construct the trivial cofibration
$E_1 \to \ol{\Id}_B(E_0,P,E_1)$, whose construction is facilitated by the
following calculation.
\begin{lemma}\label{lem:transport-iter-id-pb}
  In the setting of \Cref{constr:sigma-fact-fib}, the image of
  $\ol{\Id}_B(E_0,P,E_1) \twoheadrightarrow P \times_{E_0 \times_B E_0} (E_1
  \times_B E_1) \in \sfrac{\bC}{P}$ under the pullback along
  $r \colon E_0 \to P$ gives the map
  $\ev_\partial \colon \Id_{E_0}(E_1) \twoheadrightarrow E_1 \times_{E_0} E_1$,
  so that one has the following iterated pullbacks.
  \begin{equation*}
    \begin{tikzcd}[cramped]
      {\Id_{E_0}(E_1)} & {\ol{\Id}_B(E_0,P,E_1)} \\
      {E_1 \times_{E_0} E_1} & {P \times_{E_0 \times_B E_0} (E_1 \times_B E_1)} \\
      {E_0} & P
      \arrow[from=1-1, to=1-2]
      \arrow[two heads, from=1-1, to=2-1]
      \arrow[two heads, from=1-2, to=2-2]
      \arrow[from=2-1, to=2-2]
      \arrow[two heads, from=2-1, to=3-1]
      \arrow["\lrcorner"{anchor=center, pos=0.15, scale=1.5}, draw=none, from=1-1, to=2-2]
      \arrow["\lrcorner"{anchor=center, pos=0.15, scale=1.5}, draw=none, from=2-1, to=3-2]
      \arrow[two heads, from=2-2, to=3-2]
      \arrow["r"', from=3-1, to=3-2]
    \end{tikzcd}
  \end{equation*}
\end{lemma}
\begin{proof}
  Because $\ol{\Id}_B(E_0,P,E_1) \xrightarrow{\quad} P \times_{E_0 \times_B E_0} (E_1
  \times_B E_1)$ is constructed by taking the pullback along the map
  $\begin{tikzcd}[cramped]
    {P \times_{E_0 \times_B E_0} (E_1 \times_B E_1)} & {p_1^*E_1 \times_P p_1^*E_1}
    \arrow[dashed, from=1-1, to=1-2]
  \end{tikzcd}$
  of
  $\Id_P(p_1^*E_1) \twoheadrightarrow p_1^*E_1 \times_P p_1^*E_1$,
  it suffices to prove that rebasing the map
  \begin{equation*}
    \begin{tikzcd}[cramped]
      {P \times_{E_0 \times_B E_0} (E_1 \times_B E_1)} & {p_1^*E_1 \times_P p_1^*E_1}
      \arrow[dashed, from=1-1, to=1-2]
    \end{tikzcd} \in \sfrac{\bC}{P}
  \end{equation*}
  along $r \colon E_0 \to P$ gives the identity at $E_1 \times_{E_0} E_1$.

  To do so, we first note that
  $E_1 \times_B E_1 \twoheadrightarrow E_0 \times_B E_0$ pulls back along the
  diagonal $E_0 \to E_0 \times_B E_0$ to give $E_1 \times_{E_0} E_1$.
  Therefore, $P \times_{E_0 \times_B E_0} (E_1 \times_B E_1) \twoheadrightarrow P$
  fits into the following iterated pullback
  \begin{equation*}
    \begin{tikzcd}[cramped, column sep=small]
      {E_1 \times_{E_0} E_1} & {P \times_{E_0 \times_B E_0} (E_1 \times_B E_1)} & {E_1 \times_B E_1} \\
      {E_0} & P & {E_0 \times_B E_0}
      \arrow[from=1-1, to=1-2]
      \arrow[two heads, from=1-1, to=2-1]
      \arrow[from=1-2, to=1-3]
      \arrow[two heads, from=1-2, to=2-2]
      \arrow[two heads, from=1-3, to=2-3]
      \arrow["r"', from=2-1, to=2-2]
      \arrow["p"', from=2-2, to=2-3]
      \arrow["\lrcorner"{anchor=center, pos=0.15, scale=1.5}, draw=none, from=1-1, to=2-2]
      \arrow["\lrcorner"{anchor=center, pos=0.15, scale=1.5}, draw=none, from=1-2, to=2-3]
    \end{tikzcd}
  \end{equation*}
  In other words, the dashed map
  $P \times_{E_0 \times_B E_0} (E_1 \times_B E_1) \to p_1^*E_1 \times_P
  p_1^*E_1$ of \Cref{eqn:transport-pair} pulls back along $r \colon E_0 \to P$
  to a map of $E_1 \times_{E_0} E_1$ to itself over $E_0$.
  We finish by noting that it is the identity because
  $\transport \colon p_0^*E_1 \to p_1^*E_1 \in \sfrac{\bC}{P}$ pulls back along
  $r \colon E_0 \to P$ to give $\id \colon E_1 \to E_1$.
\end{proof}

We are now able to complete the composite of fibrations
$\ol{\Id}_B(E_0,P,E_1) \twoheadrightarrow P \times_{E_0 \times_B E_0} (E_1
\times_B E_1) \twoheadrightarrow E_1 \times_B E_1$ from
\Cref{constr:sigma-fact-fib} to a diagonal factorisation for $E_1$ over $P$ by
supplying the trivial cofibration part.
\begin{construction}\label{constr:sigma-fact-tc}
  In the setting of \Cref{constr:sigma-fact-fib}, we construct a map
  $\ol{\refl}_B(E_0,P,E_1)$ as the composite
  \begin{equation*}
    \ol{\refl}_B(E_0,P,E_1) \coloneqq
    \left(
      E_1 \xrightarrow{\refl} \Id_{E_0}(E_1) \xrightarrow{\text{\Cref{lem:transport-iter-id-pb}}} \ol{\Id}_B(E_0,P,E_1)
    \right)
  \end{equation*}
  where the final map is the connecting map into $\ol{\Id}_B(E_0,P,E_1)$ from
  its pullback under $r$, which is $\Id_{E_0}(E_1)$ by
  \Cref{lem:transport-iter-id-pb}.
\end{construction}

\begin{lemma}\label{lem:sigma-fact-tc} $ $
  The map $\ol{\refl}_B(E_0,P,E_1) \fracsquareslash{B} (B \times \pi)$
  from \Cref{constr:sigma-fact-tc} is such that
  \begin{enumerate}
    \item\label{itm:sigma-fact-tc-over} For the $\pi$-fibration
    $\ol{\Id}_B(E_0,P,E_1) \twoheadrightarrow P \times_{E_0 \times_B E_0}(E_1
    \times_B E_1)$, the following diagram commutes.
    \begin{equation*}
      \begin{tikzcd}[cramped, row sep=small]
        {E_1} & {\ol{\Id}_B(E_1,P,E_2)} \\
        & {P \times_{E_0 \times_B E_0} (E_1 \times_B E_1)} \\
        {E_0} & P
        \arrow["{\ol{\refl}_B(E_1,P,E_2)}", from=1-1, to=1-2]
        \arrow[two heads, from=1-1, to=3-1]
        \arrow[two heads, from=1-2, to=2-2]
        \arrow[two heads, from=2-2, to=3-2]
        \arrow["r"', from=3-1, to=3-2]
      \end{tikzcd}
    \end{equation*}
    \item\label{itm:sigma-fact-tc-diag} The composite
    \begin{equation*}
      \begin{tikzcd}[cramped, sep=small]
        {E_1} & {\Id_{E_0}(E_1)} & {\ol{\Id}_B(E_0,P,E_1)} & {P \times_{E_0 \times_B E_0} (E_1 \times_B E_1)} & {E_1 \times_B E_1}
        \arrow["\refl"', from=1-1, to=1-2]
        \arrow["{\ol{\refl}_B(E_0,P,E_1)}", curve={height=-12pt}, from=1-1, to=1-3]
        \arrow[from=1-2, to=1-3]
        \arrow[two heads, from=1-3, to=1-4]
        \arrow["{\ev_\partial(E_0,P,E_1)}", curve={height=-12pt}, two heads, from=1-3, to=1-5]
        \arrow[two heads, from=1-4, to=1-5]
      \end{tikzcd}
    \end{equation*}
    is the diagonal.
    \item\label{itm:sigma-fact-tc-lift} There is a lifting structure
    $\ol{\refl}_B(E_0,P,E_1) \fracsquareslash{B} (B \times \pi)$.
  \end{enumerate}
\end{lemma}
\begin{proof}
  Summarising \Cref{constr:sigma-fact-tc} with the statement of
  \Cref{lem:transport-iter-id-pb}, we arrive at the following diagram.
  \begin{equation*}
    \begin{tikzcd}[cramped]
      {E_1} & {\Id_{E_0}(E_1)} & {\ol{\Id}_B(E_0,P,E_1)} \\
      & {E_1 \times_{E_0} E_1} & {P \times_{E_0 \times_B E_0} (E_1 \times_B E_1)} & {E_1 \times_B E_1} \\
      & {E_0} & P & {E_0 \times_B E_0}
      \arrow["\refl"', from=1-1, to=1-2]
      \arrow["{\ol{\refl}_B(E_0,P,E_1)}", curve={height=-12pt}, from=1-1, to=1-3]
      \arrow["\Delta"', curve={height=12pt}, from=1-1, to=2-2]
      \arrow[from=1-2, to=1-3]
      \arrow[two heads, from=1-2, to=2-2]
      \arrow[two heads, from=1-3, to=2-3]
      \arrow["{\ev_\partial(E_0,P,E_1)}", two heads, from=1-3, to=2-4]
      \arrow[from=2-2, to=2-3]
      \arrow[two heads, from=2-2, to=3-2]
      \arrow[two heads, from=2-3, to=2-4]
      \arrow[two heads, from=2-3, to=3-3]
      \arrow[two heads, from=2-4, to=3-4]
      \arrow["r", from=3-2, to=3-3]
      \arrow["\Delta"', curve={height=12pt}, from=3-2, to=3-4]
      \arrow["p", two heads, from=3-3, to=3-4]
      \arrow["\lrcorner"{anchor=center, pos=0.15, scale=1.5}, draw=none, from=1-2, to=2-3]
      \arrow["\lrcorner"{anchor=center, pos=0.15, scale=1.5}, draw=none, from=2-2, to=3-3]
      \arrow["\lrcorner"{anchor=center, pos=0.15, scale=1.5}, draw=none, from=2-3, to=3-4]
    \end{tikzcd}
  \end{equation*}
  Chasing the diagram proves \Cref{itm:sigma-fact-tc-over}.

  For \Cref{itm:sigma-fact-tc-diag}, the composite as below is the diagonal
  \begin{equation*}
    \begin{tikzcd}[cramped, sep=small]
      {E_1} & {\Id_{E_0}(E_1)} & {\ol{\Id}_B(E_0,P,E_1)} & {P \times_{E_0 \times_B E_0} (E_1 \times_B E_1)} & {E_1 \times_B E_1}
      \arrow["\refl"', from=1-1, to=1-2]
      \arrow["{\ol{\refl}_B(E_0,P,E_1)}", curve={height=-12pt}, from=1-1, to=1-3]
      \arrow[from=1-2, to=1-3]
      \arrow[two heads, from=1-3, to=1-4]
      \arrow["{\ev_\partial(E_0,P,E_1)}", curve={height=-12pt}, two heads, from=1-3, to=1-5]
      \arrow[two heads, from=1-4, to=1-5]
    \end{tikzcd}
  \end{equation*}
  because, referring to the diagram above, the pullback of the diagonal
  $\Delta \colon E_0 \to E_0 \times_B E_0$ along
  $E_1 \times_B E_1 \twoheadrightarrow E_0 \times_B E_0$ is the connecting map
  $E_1 \times_{E_0} E_1 \to E_1 \times_B E_1$ whose both components are the
  identity map.

  Finally, for \Cref{itm:sigma-fact-tc-lift}, to see that one has a lifting
  structure $\ol{\refl}_B(E_0,P,E_1) \fracsquareslash{B} (B \times \pi)$, it
  suffices to note that lifting structures are preserved by composition by
  \cite[Construction 2.5]{struct-lift}, and $E_1 \to \Id_{E_0}(E_1)$ has a
  lifting structure by definition of the $\Id$-type structure.
  Therefore, it suffices to see that $\Id_{E_0}(E_1) \to \ol{\Id}_B(E_0,P,E_1)$
  has a lifting structure against $B \times \pi$.
  But it is the pullback of $r$ under a $\pi$-fibration, so by the $\Pi$-type
  structure and \cite[Construction 3.6]{struct-lift}, it also has a lifting
  structure against $B \times \pi$.
\end{proof}

\subsubsection{Id-type for Reedy Fibrations}
We now switch our attention back to constructing the $\Id$-type for Reedy
fibrations.
The idea is to use the $\Id$-type in the 0-component and
\Cref{constr:sigma-fact-fib,constr:sigma-fact-tc} for the 1-component.
We already know that $\refl_0$ has a left lifting structure and from
\Cref{lem:sigma-fact-tc}\Cref{itm:sigma-fact-tc-lift} that so does $\ol{\refl}$
of \Cref{constr:sigma-fact-tc}.
We also recall as follows that left lifting structures in the gluing category
are assembled pointwise.

\begin{lemma}\label{lem:gl-tc-ptwise-tc}
  For each $B \in \Gl(M)$ and map $s \colon Y \to E$ in the Reedy fibrant slice
  over $B$, taking $i$ to be either $i=0$ or $i=1$, one has a map
  \begin{equation*}
    (s \fracsquareslash{B} (B \times \tau)) \xrightarrow{\quad} (s_i \fracsquareslash{B_i} (B_i \times \pi_i))
  \end{equation*}
  subject to the commutativity condition that for each $f \colon B' \to B \in \Gl(M)$, one has
  \begin{equation*}
    \begin{tikzcd}[cramped, row sep=small]
      (s \fracsquareslash{B} (B \times \tau))
      \ar[r, "{}"]
      \ar[d, "{\text{\cite[Construction 3.2]{struct-lift}}}"']
      &
      (s_i \fracsquareslash{B_i} (B_i \times \pi_i))
      \ar[d, "{\text{\cite[Construction 3.2]{struct-lift}}}"]
      \\
      (f^*s \fracsquareslash{B'} (B' \times \tau))
      \ar[r, "{}"]
      &
      (f_i^*s_i \fracsquareslash{B_i'} (B_i' \times \pi_i))
    \end{tikzcd}
  \end{equation*}
\end{lemma}
\begin{proof}
  The 0-component is straightforward by noting that the map
  $\id_{M\tMcU_0} \to \id_{M\McU_0} \in \Gl(M)$ is a Reedy fibration by the
  $\Unit$-type structure.

  For the 1-component, we must take
  $\ell \in (s \fracsquareslash{B} (B \times \tau))$ and construct a
  corresponding $\ell_1 \in (s_1 \fracsquareslash{B_1} (B_1 \times \pi_1))$.
  To do so, suppose one has a map $f_1 \colon B_1' \to B_1$ and a lifting
  problem of $f_1^*Y_1 \to f_1^*E_1$ against $\tMcU_1 \to \McU_1$ in $\bC_1$, as
  in the back face below.
  Because $M$ is a right adjoint, it preserves the terminal object.
  Thus, the right face of the cube below is a Reedy fibration as identities are
  $\pi_1$-fibrations by the $\Unit$-type structure.
  \begin{equation*}
    \begin{tikzcd}[cramped, sep=small]
      {f_0^*Y_1} && {\tMcU_1} \\
      & {M(f_0^*Y_0)} && M1 \\
      {f_1^*E_1} && {\McU_1} \\
      & {M(f_0^*E_0)} && M1
      \arrow[from=1-1, to=1-3]
      \arrow[from=1-1, to=2-2]
      \arrow[from=1-1, to=3-1]
      \arrow["{!}", from=1-3, to=2-4]
      \arrow[from=1-3, to=3-3]
      \arrow["{!}"{description, pos=0.3}, from=2-2, to=2-4]
      \arrow[from=2-2, to=4-2]
      \arrow[from=2-4, to=4-4]
      \arrow["{\ell_1}"{description, pos=0.3}, dashed, from=3-1, to=1-3]
      \arrow[from=3-1, to=3-3]
      \arrow[from=3-1, to=4-2]
      \arrow["{!}", from=3-3, to=4-4]
      \arrow["{M\ell_0}"{description, pos=0.3}, dashed, from=4-2, to=2-4]
      \arrow["{!}"', from=4-2, to=4-4]
    \end{tikzcd}
  \end{equation*}
  Using $\ell$, one can solve the above lifting problem in $\Gl(M)$, whose
  1-component is a solution to the original lifting problem of $f_1^*s_1$
  against $\pi_1$.

  The required commutativity conditions are straightforwardly inherited.
\end{proof}

We also have the converse result of the above, which is that pointwise left maps
are left maps to Reedy fibrations.

\begin{lemma}\label{lem:ptwise-tc-gl-tc}
  For each $B \in \Gl(M)$ and map $s \colon Y \to E$ in the Reedy fibrant slice
  over $B$, one has a map
  \begin{equation*}
    (s_0 \fracsquareslash{B_0} (B_0 \times \pi_0))
    \times
    (s_1 \fracsquareslash{B_1} (B_1 \times \pi_1))
    \xrightarrow{\quad}
    (s \fracsquareslash{B} (B \times \tau))
  \end{equation*}
  subject to the commutativity condition that for each $f \colon B' \to B \in \Gl(M)$, one has
  \begin{equation*}
    \begin{tikzcd}[cramped, row sep=small]
      (s_0 \fracsquareslash{B_0} (B_0 \times \pi_0))
      \times
      (s_1 \fracsquareslash{B_1} (B_1 \times \pi_1))
      \ar[r, "{}"]
      \ar[d]
      &
      (s \fracsquareslash{B} (B \times \tau))
      \ar[d]
      \\
      (f_0^*s_0 \fracsquareslash{B_0'} (B_0' \times \pi_0))
      \times
      (f_1^*s_1 \fracsquareslash{B_1'} (B_1' \times \pi_1))
      \ar[r, "{}"]
      &
      (f^*s \fracsquareslash{B'} (B' \times \tau))
    \end{tikzcd}
  \end{equation*}
\end{lemma}
\begin{proof}
  Given pointwise lifts, one produces a lift in the gluing category by first
  solving the lifting problem in the 0-component and then lifting the
  1-component against the fibrant relative matching map.
\end{proof}

We are now ready to show that Reedy fibrations also admit an $\Id$-type structure.
\begin{proposition}\label{prop:reedy-fib-id}
  For each (internal) Reedy fibration $E \to B \in \Gl(M)$, the diagonal
  $E \to E \times_B E \in \sfrac{\Gl(M)}{B}$ admits a factorisation
  \begin{equation*}
    E \xrightarrow{\refl} \Id^{\Gl(M)}_B(E) \xrightarrow{\ev_\partial} E \times_B E \in \sfrac{\Gl(M)}{B}
  \end{equation*}
  where
  \begin{enumerate}
    \item\label{itm:reedy-fib-id-fib} The map $\ev_\partial$ is a (internal) Reedy fibration
    \item\label{itm:reedy-fib-id-res} Pulling back along the map
    $\WhM_B(\refl) \colon \WhM_B(E) \to \WhM_B(\Id_BE)$ between the relative
    matching objects sends the relative matching map
    $(\Id_BE)_1 \twoheadrightarrow \WhM_{E \times_B E}(\Id_BE) \in
    \sfrac{\bC_1}{\WhM_B(\Id_BE)}$ to the boundary evaluation
    $\Id_{\WhM_B(E)}(E_1) \twoheadrightarrow E_1 \times_{\WhM_BE}E_1 \in
    \sfrac{\bC}{\WhM_B(E)}$.
    In other words, one has the following iterated pullbacks.
    \begin{equation*}
      \begin{tikzcd}[cramped]
        {\Id_{\WhM_B(E)}(E_1)} & {(\Id_BE)_1} \\
        {E_1 \times_{\WhM_BE}E_1} & {\WhM_{E \times_B E}(\Id_BE)} \\
        {\WhM_B(E)} & {\WhM_B(\Id_BE)}
        \arrow[from=1-1, to=1-2]
        \arrow["{\ev_\partial}"', two heads, from=1-1, to=2-1]
        \arrow["{\Whm_{E \times_B E}(\Id_BE)}", two heads, from=1-2, to=2-2]
        \arrow[from=2-1, to=2-2]
        \arrow[two heads, from=2-1, to=3-1]
        \arrow["\lrcorner"{anchor=center, pos=0.15, scale=1.5}, draw=none, from=1-1, to=2-2]
        \arrow["\lrcorner"{anchor=center, pos=0.15, scale=1.5}, draw=none, from=2-1, to=3-2]
        \arrow[two heads, from=2-2, to=3-2]
        \arrow["{\WhM_B(\refl)}"', from=3-1, to=3-2]
      \end{tikzcd}
    \end{equation*}
    \item\label{itm:reedy-fib-id-lift} There is a choice of a lifting structure $\refl \fracsquareslash{B} (B \times \tau)$
  \end{enumerate}
\end{proposition}
\begin{proof}
  As before, the internal Reedy fibration case and ambient Reedy fibration case
  are identical, so we only show the proof for the ambient Reedy fibration case.

  We begin by constructing the required factorisation.
  The 0-component is given by the $\Id$-type structure on $\bC_0$.
  For the 1-component, we first note that by \Cref{asm:mrq} \Cref{itm:mrq-tc},
  one has a lifting structure
  $M(\refl_0) \fracsquareslash{MB_0} (MB_0 \times \pi_1)$, which gives rise to a
  lifting structure $\WhM_B(\refl_0) \fracsquareslash{B_1} (B_1 \times \pi_1)$ by
  rebasing.
  Hence, one has a factorisation of the diagonal
  $\WhM_B(E) \to \WhM_B(E) \times_{B_1} \WhM_B(E)$ as
  \begin{equation*}
    \WhM_B(E) \xrightarrow{\WhM_B(\refl)} \WhM_B(\Id_BE) \xrightarrow{\WhM_B(\ev_\partial)}
    \WhM_B(E) \times_{B_1} \WhM_B(E)
  \end{equation*}
  satisfying the conditions to apply \Cref{constr:sigma-fact-fib}.
  Doing so then produces a fibrant object
  $\ol{\Id}_{B_1}(E_0,\WhM_B(\ev_\partial),E_1)$ over
  $\WhM_B(\Id_BE) \times_{\WhM_B(E_0) \times_{B_1} \WhM_B(E_0)} (E_1 \times_{B_1} E_1)$.
  But recognising that $E_1 \times_{B_1} E_1 \to ME_0 \times_{MB_0} ME_0$
  factors as
  \begin{equation*}
    E_1 \times_{B_1} E_1 \xrightarrow{\quad} \WhM_BE \times_{B_1} \WhM_BE \xrightarrow{\quad} ME_0 \times_{MB_0} ME_0
  \end{equation*}
  it then follows that
  $\WhM_B(\Id_BE) \times_{\WhM_B(E_0) \times_{B_1} \WhM_B(E_0)} (E_1 \times_{B_1}
  E_1) \cong \WhM_{E \times_B E}(\Id_BE)$ is the relative matching object for
  $\Id_B(E)$ as a fibrant object over $E \times_B E$.
  Hence, we may take the 1-component as
  \begin{equation*}
    (\Id_BE)_1 \coloneqq \ol{\Id}_{B_1}(E_0,\WhM_B(\ev_\partial),E_1)
  \end{equation*}
  In summary, we have the following iterated grid of pullbacks.
  \begin{equation*}
    \begin{tikzcd}[cramped, column sep=small]
      {(\Id_BE)_1 \coloneqq \ol{\Id}_{B_1}(E_0,\WhM_B(\ev_\partial),E_1)} & {\WhM_{E \times_B E}(\Id_BE)} & {E_1 \times_{B_1} E_1} \\
      & {\WhM_B(\Id_BE)} & {\WhM_B(E) \times_{B_1} \WhM_B(E)} & {B_1} \\
      & {M(\Id_{B_0}(E_0))} & {ME_0 \times_{MB_0} ME_0} & {MB_0}
      \arrow[two heads, from=1-1, to=1-2]
      \arrow[two heads, from=1-2, to=1-3]
      \arrow[two heads, from=1-2, to=2-2]
      \arrow[two heads, from=1-3, to=2-3]
      \arrow["{\WhM_B(\ev_\partial)}", two heads, from=2-2, to=2-3]
      \arrow[from=2-2, to=3-2]
      \arrow[two heads, from=2-3, to=2-4]
      \arrow[from=2-3, to=3-3]
      \arrow["\lrcorner"{anchor=center, pos=0.15, scale=1.5}, draw=none, from=1-2, to=2-3]
      \arrow["\lrcorner"{anchor=center, pos=0.15, scale=1.5}, draw=none, from=2-2, to=3-3]
      \arrow["\lrcorner"{anchor=center, pos=0.15, scale=1.5}, draw=none, from=2-3, to=3-4]
      \arrow[from=2-4, to=3-4]
      \arrow["{M(\ev_\partial)}"', two heads, from=3-2, to=3-3]
      \arrow[two heads, from=3-3, to=3-4]
    \end{tikzcd}
  \end{equation*}
  Then, by construction, $\Id_B^{\Gl(M)}(E) \twoheadrightarrow E \times_B E$ is
  Reedy fibrant, proving \Cref{itm:reedy-fib-id-fib}.
  Using \Cref{lem:transport-iter-id-pb} also immediately gives
  \Cref{itm:reedy-fib-id-res}.

  For \Cref{itm:reedy-fib-id-lift}, one applies \Cref{constr:sigma-fact-tc} to
  obtain the 1-component of $\refl$ as the composite
  \begin{equation*}
    \refl_1 \coloneqq \ol{\refl}_B(E_0,\WhM_B(\ev_\partial),E_1)
    = \left(
      E_1 \xrightarrow{\refl} \Id_{\WhM_B(E)}(E_1) \xrightarrow{\quad} (\Id_BE)_1
    \right)
  \end{equation*}
  so that \Cref{lem:sigma-fact-tc}\Cref{itm:sigma-fact-tc-over} shows $\refl_1$
  is a map over $\WhM_B(\refl)$.
  As a result, one has the following commutative squares of $\bC_1$ forming a
  map $E \to \Id_B(E)$ in $\Gl(M)$.
  \begin{equation*}
    \begin{tikzcd}[cramped, sep=small]
      {E_1} & {\WhM_B(E)} & {M(E_0)} \\
      {(\Id_BE)_1} & {\WhM_B(\Id_BE)} & {M(\Id_{B_0}(E_0))} \\
      & {B_1} & {MB_0}
      \arrow[from=1-1, to=1-2]
      \arrow["{\refl_1}"', from=1-1, to=2-1]
      \arrow[from=1-2, to=1-3]
      \arrow["{\WhM_B(\refl)}"', from=1-2, to=2-2]
      \arrow["{M(\refl_0)}", from=1-3, to=2-3]
      \arrow[from=2-1, to=2-2]
      \arrow[from=2-2, to=2-3]
      \arrow[from=2-2, to=3-2]
      \arrow["\lrcorner"{anchor=center, pos=0.15, scale=1.5}, draw=none, from=1-2, to=2-3]
      \arrow["\lrcorner"{anchor=center, pos=0.15, scale=1.5}, draw=none, from=2-2, to=3-3]
      \arrow[from=2-3, to=3-3]
      \arrow[from=3-2, to=3-3]
    \end{tikzcd}
  \end{equation*}
  \Cref{lem:sigma-fact-tc}\Cref{itm:sigma-fact-tc-lift} also equips $\refl_1$
  with a lifting structure $\refl_1 \fracsquareslash{B_1} (B_1 \times \pi_1)$.
  Therefore, the required lifting structure
  $\refl \fracsquareslash{B} (B \times \tau)$ of \Cref{itm:reedy-fib-id-lift} is
  provided by assembling the lifting structures of the 0- and 1-components using
  \Cref{lem:ptwise-tc-gl-tc}.

  Finally, we note that this provides a factorisation of the diagonal by
  \Cref{lem:sigma-fact-tc}\Cref{itm:sigma-fact-tc-diag}.
\end{proof}

\begin{theorem}\label{thm:glue-universe-id}
  Both $\tau \colon \tMcV \to \McV \in \Gl(M)$ and
  $\tau^\ds \colon \tMcV^\ds \to \McV^\ds \in \Gl(M)$ admits an $\Id$-type
  structure preserved by the 0-component projection $\Gl(M) \to \bC_0$.
\end{theorem}
\begin{proof}
  Immediate by applying \Cref{prop:reedy-fib-id} on $\tau$ and $\tau^\ds$
  respectively.
\end{proof}


\subsection{Pointed Functional Extensionality}\label{subsec:glue-universe-funext}
The goal is now to show that if $\bC_i$ has pointed functional extensionality
then so does $\Gl(M)$.
The idea is to use \Cref{lem:ptd-funext-sdf} so that by
\Cref{lem:gl-tc-ptwise-tc,lem:ptwise-tc-gl-tc}, we can reduce the problem to
checking that the 0- and 1-components of the image of a left map under the
pushforward in the gluing category again are respectively left maps.
Due to how the 1-component of the pushforward in the gluing category behaves as
observed in \cite{fkl24}, we require the following cube lemma from
\cite{shu15}.
\begin{lemma}[{\cite[Lemma 11.7]{shu15}}]\label{lem:trv-cof-cube}
  Let $\pi \colon \tMcU \to \McU$ be a universal map in a finitely complete
  category $\bC$ with a $\Pi$-structure $\Pi \colon \bP_\pi(\McU) \to \McU$ and
  an $\Id$-type structure $\Id \colon \tMcU \times_\McU \tMcU \to \McU$.
  Also suppose that $\pi$ has a $\Sigma$-type structure.

  Suppose that one has the following diagram between $\pi$-fibrant objects over
  some $B \in \bC$ where the vertical maps on the right face are
  $\pi$-fibrations and the back and front faces are pullbacks.
  \begin{equation*}
    \begin{tikzcd}[cramped, sep=small]
      {X_4} && {X_3} \\
      & {Y_4} && {Y_3} \\
      {X_2} && {X_1} \\
      & {Y_2} && {Y_1}
      \arrow[from=1-1, to=1-3]
      \arrow["{u_4}"{}, from=1-1, to=2-2]
      \arrow[from=1-1, to=3-1]
      \arrow["{u_3}"{description}, from=1-3, to=2-4]
      \arrow[two heads, from=1-3, to=3-3]
      \arrow[from=2-2, to=2-4]
      \arrow[from=2-2, to=4-2]
      \arrow[two heads, from=2-4, to=4-4]
      \arrow[from=3-1, to=3-3]
      \arrow["{u_2}"{description}, from=3-1, to=4-2]
      \arrow["{u_1}"{description}, from=3-3, to=4-4]
      \arrow[from=4-2, to=4-4]
      \arrow["\lrcorner"{anchor=center, pos=0.15, scale=1.5}, draw=none, from=1-1, to=3-3]
      \arrow["\lrcorner"{anchor=center, pos=0.15, scale=1.5}, draw=none, from=2-2, to=4-4]
    \end{tikzcd}
  \end{equation*}
  Then, one has a map
  \begin{equation*}
    (u_1 \fracsquareslash{B} (B \times \pi)) \times
    (u_2 \fracsquareslash{B} (B \times \pi)) \times
    (u_3 \fracsquareslash{B} (B \times \pi)) \xrightarrow{\quad} (u_1 \fracsquareslash{B} (B \times \pi))
  \end{equation*}
  subject to the commutativity condition as below for each
  $f \colon B' \to B \in \bC$.
  \begin{lrbox}{\BUbox}\scriptsize \begin{tikzcd}[cramped,sep=small]  B \times \tMcU \ar[d] \\ B \times \McU \end{tikzcd} \end{lrbox}
  \newsavebox{\BpUbox}
  \begin{lrbox}{\BpUbox}\scriptsize \begin{tikzcd}[cramped,sep=small] B' \times \tMcU \ar[d] \\ B' \times \McU \end{tikzcd} \end{lrbox}
  \newsavebox{\XYabox}
  \begin{lrbox}{\XYabox}\scriptsize \begin{tikzcd}[cramped,sep=small] X_1 \ar[d, "{u_1}"'] \\ Y_1 \end{tikzcd} \end{lrbox}
  \newsavebox{\XYapbox}
  \begin{lrbox}{\XYapbox}\scriptsize \begin{tikzcd}[cramped,sep=small] f^*X_1 \ar[d, "{f^*u_1}"'] \\ f^*Y_1 \end{tikzcd} \end{lrbox}
  \newsavebox{\XYbbox}
  \begin{lrbox}{\XYbbox}\scriptsize \begin{tikzcd}[cramped,sep=small] X_2 \ar[d, "{u_2}"'] \\ Y_2 \end{tikzcd} \end{lrbox}
  \newsavebox{\XYbpbox}
  \begin{lrbox}{\XYbpbox}\scriptsize \begin{tikzcd}[cramped,sep=small] f^*X_2 \ar[d, "{f^*u_2}"'] \\ f^*Y_2 \end{tikzcd} \end{lrbox}
  \newsavebox{\XYcbox}
  \begin{lrbox}{\XYcbox}\scriptsize \begin{tikzcd}[cramped,sep=small] X_3 \ar[d, "{u_3}"'] \\ Y_3 \end{tikzcd} \end{lrbox}
  \newsavebox{\XYcpbox}
  \begin{lrbox}{\XYcpbox}\scriptsize \begin{tikzcd}[cramped,sep=small] f^*X_3 \ar[d, "{f^*u_3}"'] \\ f^*Y_3 \end{tikzcd} \end{lrbox}
  \newsavebox{\XYdbox}
  \begin{lrbox}{\XYdbox}\scriptsize \begin{tikzcd}[cramped,sep=small] X_4 \ar[d, "{u_4}"'] \\ Y_4 \end{tikzcd} \end{lrbox}
  \newsavebox{\XYdpbox}
  \begin{lrbox}{\XYdpbox}\scriptsize \begin{tikzcd}[cramped,sep=small] f^*X_4 \ar[d, "{f^*u_4}"'] \\ f^*Y_4 \end{tikzcd} \end{lrbox}
  \begin{equation*}
    \begin{tikzcd}[cramped]
      \left(\usebox{\XYabox} \fracsquareslash{B} \usebox{\BUbox} \right)
      \times
      \left(\usebox{\XYbbox} \fracsquareslash{B} \usebox{\BUbox} \right)
      \times
      \left(\usebox{\XYcbox} \fracsquareslash{B} \usebox{\BUbox} \right)
      \ar[r]
      \ar[d]
      &
      \left(\usebox{\XYdbox}  \fracsquareslash{B} \usebox{\BUbox} \right)
      \ar[d]
      \\
      \left(\usebox{\XYapbox} \fracsquareslash{B'} \usebox{\BpUbox} \right)
      \times
      \left(\usebox{\XYbpbox} \fracsquareslash{B'} \usebox{\BpUbox} \right)
      \times
      \left(\usebox{\XYcpbox} \fracsquareslash{B'} \usebox{\BpUbox} \right)
      \ar[r]
      &
      \left(\usebox{\XYdpbox}  \fracsquareslash{B'} \usebox{\BUbox} \right)
    \end{tikzcd}
  \end{equation*}
\end{lemma}
\begin{proof}
  By \cite[Lemma 11.7]{shu15} and noting that the $\pi$-fibrant slice over each
  object $B$ forms a type-theoretic fibration category and pulling back along a
  map $B' \to B$ induces a map between type-theoretic fibration categories.
\end{proof}

We are now ready to prove pointed functional extensionality of the gluing
category.

\begin{theorem}\label{thm:gl-ptd-fe}
  Suppose that each $\pi_i$ is equipped with a pointed $\Id_i$-functional
  extensionality structure.
  Then, any $\Id$-type structure on $\tau$ has a pointed functional
  extensionality structure.
\end{theorem}
\begin{proof}
  By the characterisation $\Id$-type pointed functional extensionality of
  \Cref{lem:ptd-funext-sdf}, it suffices to show that pushforwards along Reedy
  fibrations preserve left classes to Reedy fibrations, assuming that pushing
  forwards along $\pi_i$-fibrations preserve left classes to $\pi_i$-fibrations.

  Specifically, take a map $s \colon Y \to E$ in the Reedy fibrant slice over
  some $B \in \Gl(M)$ and a Reedy fibration $q \colon B \twoheadrightarrow A$.
  The goal is to construct a map
  \begin{equation*}
    \begin{tikzcd}[cramped]
      (s \fracsquareslash{B} (B \times \tau)) \ar[r,dashed] & (q_*s \fracsquareslash{A} (B \times \tau))
    \end{tikzcd}
  \end{equation*}
  satisfying the compatibility condition of \Cref{lem:ptd-funext-sdf}.
  By \Cref{lem:gl-tc-ptwise-tc,lem:ptwise-tc-gl-tc}, one has solid maps as
  follows satisfying the required compatibility conditions.
  \begin{equation*}
    \begin{tikzcd}[cramped, column sep=large]
      (s \fracsquareslash{B} (B \times \pi))
      \ar[r,"{\text{\Cref{lem:gl-tc-ptwise-tc}}}"]
      &
      (s_0 \fracsquareslash{B_0} (B_0 \times \pi_0))
      \times
      (s_1 \fracsquareslash{B_1} (B_1 \times \pi_0))
      \ar[d, dashed]
      \\
      (q_*s \fracsquareslash{A} (B \times \pi))
      &
      ((q_*s)_0 \fracsquareslash{B_0} (B_0 \times \pi_0))
      \times
      ((q_*s)_1 \fracsquareslash{B_1} (B_1 \times \pi_0))
      \ar[l,"{\text{\Cref{lem:ptwise-tc-gl-tc}}}"']
    \end{tikzcd}
  \end{equation*}
  Therefore, it suffices to construct the dashed map above stable under
  reindexing.

  Because the map $\Gl(M) \to \bC_0$ preserves pushforwards by \cite[Theorem
  2.19]{fkl24}, one has that $(q_*s)_0 = (q_0)_*s_0$.
  Thus, the $\Id_0$-functional extensionality on $\pi_0$ provides map of
  the 0-component
  \begin{equation*}
    (s_0 \fracsquareslash{B_0} (B_0 \times \pi))
    \xrightarrow{\quad}
    ((q_*s)_0 \fracsquareslash{B_0} (B_0 \times \pi))
  \end{equation*}

  For the 1-component, \cite[Construction 2.12]{fkl24} shows the 1-component
  $(q_*s)_1 \colon (q_*Y)_1 \to (q_*E)_1$ of the pushforward of $s$ along $q$ is
  given by the following connecting map of pullbacks in $\bC_1$ over $B_1$
  \begin{equation*}
    \begin{tikzcd}[cramped, sep=small]
      {(q_*Y)_1} && {(q_1)_*Y_1} \\
      & {(q_*E)_1} && {(q_1)_*E_1} \\
      {\WhM_A(q_*Y)} && {(q_1)_*(\WhM_BY)} \\
      & {\WhM_A(q_*E)} && {(q_1)_*(\WhM_BE)}
      \arrow[from=1-1, to=1-3]
      \arrow[two heads, from=1-1, to=3-1]
      \arrow["{(q_1)_*s_1}", from=1-3, to=2-4]
      \arrow["{(q_1)_*\Whm_BY}"{pos=0.8}, two heads, from=1-3, to=3-3]
      \arrow[from=2-2, to=2-4]
      \arrow[two heads, from=2-2, to=4-2]
      \arrow["{(q_1)_*\Whm_BE}", two heads, from=2-4, to=4-4]
      \arrow[from=3-1, to=3-3]
      \arrow["{\WhM_A(q_*s)}"', from=3-1, to=4-2]
      \arrow["{(q_1)_*(\WhM_Bs)}"'{}, from=3-3, to=4-4]
      \arrow[from=4-2, to=4-4]
      \arrow["\lrcorner"{anchor=center, pos=0.15, scale=1.5}, draw=none, from=1-1, to=3-3]
      \arrow["\lrcorner"{anchor=center, pos=0.15, scale=1.5}, draw=none, from=2-2, to=4-4]
      \arrow["{(q_*s)_1}"'{}, from=1-1, to=2-2]
    \end{tikzcd}
  \end{equation*}
  The vertical maps $(q_1)_*E_1 \to (q_1)_*(\WhM_BE)$ and
  $(q_1)_*Y_1 \to (q_1)_*(\WhM_BY)$ are $\pi_1$-fibrations because $\Whm_BE$ and
  $\Whm_BY$ are $\pi_1$-fibrations by Reedy fibrancy and the distributivity law
  for pushforwards (\cite[Paragraph 2.3]{gk13}, \cite[Lemma 2.3]{inv-psfw})
  implies implies pushforwards of $\pi_1$-fibrations along $\pi_1$-fibrations
  remain $\pi_1$-fibrations.
  We would now like to invoke the cube lemma of \Cref{lem:trv-cof-cube}, so we
  examine the three connecting maps from cospan of the back face to the cospan
  of the front face and hope for lifting structures:
  \begin{enumerate}
    \item $\WhM_A(q_*s)$ is the pullback along $A_1 \to MA_0$ of the map
    $M((q_0)_*s_0)$.
    Therefore, one has a composite
    \begin{equation*}
      \begin{tikzcd}[row sep=small]
        (s_0 \fracsquareslash{B_0} (B_0 \times \pi_0))
        \ar[r,"{\ptdfunext_0}"]
        &
        ((q_0)_*s_0 \fracsquareslash{A_0} (A_0 \times \pi_0))
        \ar[d, "{\text{\Cref{asm:mrq}\Cref{itm:mrq-tc}}}"]
        \\
        (\WhM_A(q_*s) \fracsquareslash{A_1} (A_1 \times \pi_1))
        &
        (M((q_0)_*s_0) \fracsquareslash{MA_0} (MA_0 \times \pi_1))
        \ar[l]
      \end{tikzcd}
    \end{equation*}
    \item $(q_1)_*(\WhM_Bs)$ is the pushforward along $q_1$ of the pullback of
    $Ms_0$ along $B_1 \to MB_0$.
    This means that one has a composite
    \begin{equation*}
      \begin{tikzcd}[row sep=small, column sep=large]
        (s_0 \fracsquareslash{B_0} (B_0 \times \pi_0))
        \ar[r, "{\text{\Cref{asm:mrq}\Cref{itm:mrq-tc}}}"]
        &
        (Ms_0 \fracsquareslash{MB_0} (MB_0 \times \pi_1))
        \ar[d]
        \\
        ((q_1)_*(\WhM_Bs) \fracsquareslash{A_1} (A_1 \times \pi_1))
        &
        (\WhM_Bs \fracsquareslash{B_1} (B_1 \times \pi_1))
        \ar[l, "{\ptdfunext_1}"]
      \end{tikzcd}
    \end{equation*}
    \item $(q_1)_*s_1$ is the pushforward of $s_1$ along $q_1$, which is a
    fibration by the $\Sigma$-type structure, so by the pointed functional
    extensionality assumption one has the map
    \begin{equation*}
      (s_1 \fracsquareslash{B_1} (B_1 \times \pi_1))
      \xrightarrow{\ptdfunext_1}
      ((q_1)_*s_1 \fracsquareslash{A_1} (A_1 \times \pi_1))
    \end{equation*}
  \end{enumerate}
  All of the maps above are stable under re-indexing.
  Using the cube lemma of \Cref{lem:trv-cof-cube}, one then obtains a map
  \begin{equation*}
    (s_0 \fracsquareslash{B_0} (B_0 \times \pi_0))
    \times
    (s_1 \fracsquareslash{B_1} (B_1 \times \pi_0))
    \xrightarrow{\quad}
    ((q_1)_*s_1 \fracsquareslash{A_1} (A_1 \times \pi_1))
  \end{equation*}
  stable under reindexing, as required.
\end{proof}


\subsection{Pointed Univalence}\label{subsec:glue-universe-univalence}
We are now ready to prove that the gluing category $\Gl(M)$ has pointed univalence.
Specifically, we will show that the internal universe of internal Reedy
fibrations equipped with the $\Id$-type structure from \Cref{thm:gl-ptd-fe} is
pointed univalent relative to the universe of ambient Reedy fibrations.
Formally, we will prove the following theorem.

\def\glueua{
  Suppose that
  \begin{enumerate}[itemsep=0em]
    \item $M \colon \bC_0 \to \bC_1$ has a left adjoint
    $D \colon \bC_1 \to \bC_0$
    \item There is a $\pi_i$-univalence structure on the \emph{internal}
    (universe, $\Id$-type)-pair $(\pi_i^\ds, \Id_i^\ds)$.
    \item There is an $\Id_1$-functional extensionality structure on the
    \emph{ambient} $\Id$-type $\Id_1$ of $\pi_1$ of the 1-component.
  \end{enumerate}
  Then, there is a pointed $\tau$-univalence structure on the internal
  (universe, $\Id$-type)-pair $(\tau^\ds, \Id_{\Gl(M)}^\ds)$, where
  $\tau,\tau^\ds$ are from \Cref{thm:glue-universe-univ} and $\Id_{\Gl(M)}^\ds$
  is from \Cref{thm:glue-universe-id}.
}

\begin{theorem}\label{thm:glue-universe-ua}\glueua\end{theorem}
We will obtain the proof of this \lcnamecref{thm:glue-universe-ua} in few steps.
The bulk of the technicality is in the following part, where one also can
observe why the functional extensionality premise is required.

\subsubsection{Relative matching map of the trivial homotopy equivalence as a pushforward}
Applying \Cref{constr:hiso,constr:trv-htpy} on the internal
universal Reedy fibration $\tau^\ds$ with the internal $\Id$-type $\Id_{\Gl(M)}^\ds$,
one obtains a map
\begin{equation*}
  \trv_{\Gl(M)} \colon \McV^\ds \to \HIso_{\McV^\ds}^{\Id^\ds_{\Gl(M)}}(\tMcV^\ds)
\end{equation*}
Because $\Gl(M) \to \bC_0$ preserves all logical structures, the 0-component is
just $\trv_0 \colon \McU_0 \to \HIso_{\McU_0}^{\Id_0}(\tMcU_0)$.
It also has a relative matching map $\Whm(\trv_{\Gl(M)})$ indicated as the
dotted map below.
\begin{equation*}\label{eqn:gl-trv}\tag{\text{$\Gl$-$\trv$}}
  \begin{tikzcd}[cramped]
    {\McV_1^\ds} \\
    & {\WhM(\trv_{\Gl(M)})} & {M(\McU_0^\ds)} \\
    & {\HIso_{\McV^\ds}^{\Id^\ds_{\Gl(M)}}(\tMcV^\ds)_1} & {M(\HIso_{\McU_0}^{\Id_0^\ds}(\tMcU_0^\ds))}
    \arrow["{\Whm(\trv_{\Gl(M)})}"{description}, dotted, from=1-1, to=2-2]
    \arrow[curve={height=-12pt}, from=1-1, to=2-3]
    \arrow["{(\trv_{\Gl(M)})_1}"', curve={height=12pt}, dashed, from=1-1, to=3-2]
    \arrow[from=2-2, to=2-3]
    \arrow[dashed, from=2-2, to=3-2]
    \arrow["\lrcorner"{anchor=center, pos=0.15, scale=1.5}, draw=none, from=2-2, to=3-3]
    \arrow["{M(\trv_0)}", dashed, from=2-3, to=3-3]
    \arrow[from=3-2, to=3-3]
  \end{tikzcd}
  \in \bC_1
\end{equation*}
The goal of this part is to give an alternative formulation of the relative
matching map $\Whm(\trv_{\Gl(M)})$.
Specifically, we show the following.

\newsavebox\relmatchtrvpsfwbox
\begin{lrbox}{\relmatchtrvpsfwbox}
  \begin{tikzcd}[cramped, column sep=small]
    {\McU_1^\ds} & {M\tMcU^\ds_0 \times \McU^\ds_1} & {\bP_{M\pi^\ds_0}(\McU^\ds_1)} \\
    {\HIso_{\McU_1^\ds}^{\Id_1}(\tMcU_1^\ds)} & {M\tMcU_0 \times \HIso_{\McU_1^\ds}^{\Id_1}(\tMcU_1^\ds)} & {\bP_{M\pi^\ds_0}(\HIso_{\McU_1^\ds}^{\Id_1}(\tMcU_1^\ds))} & {\WhM(\trv_{\Gl(M)})} \\
    1 & {M\tMcU^\ds_0} & {M\McU_0}
    \arrow["{\trv_1}"', from=1-1, to=2-1]
    \arrow[from=1-2, to=1-1]
    \arrow["{M\tMcU_0 \times \trv_1}", from=1-2, to=2-2]
    \arrow["{\bP_{M\pi_0}(\trv_1)}"', from=1-3, to=2-3]
    \arrow["{\Whm(\trv_{\Gl(M)})}", from=1-3, to=2-4]
    \arrow[from=2-1, to=3-1]
    \arrow[from=2-2, to=2-1]
    \arrow["\lrcorner"{anchor=center, pos=0.125, rotate=270, scale=1.5, pos=0.15}, draw=none, from=1-2, to=2-1]
    \arrow["\lrcorner"{anchor=center, pos=0.125, rotate=270, scale=1.5, pos=0.15}, draw=none, from=2-2, to=3-1]
    \arrow[from=2-2, to=3-2]
    \arrow["\cong", tail reversed, from=2-3, to=2-4]
    \arrow[from=2-3, to=3-3]
    \arrow[from=2-4, to=3-3]
    \arrow[from=3-2, to=3-1]
    \arrow["{M\pi^\ds_0}"', from=3-2, to=3-3]
  \end{tikzcd}
\end{lrbox}

\def\relmatchtrvpsfw{
  Suppose $M \colon \bC_0 \to \bC_1$ has a left adjoint
  $D \colon \bC_1 \to \bC_0$.

  The relative matching map
  $\Whm(\trv_{\Gl(M)}) \colon \McV^\ds_1 \to \WhM(\trv_{\Gl(M)})$ from
  \Cref{eqn:gl-trv} is the image of
  $\trv_1 \colon \McU_1^\ds \to \HIso_{\McU_1^\ds}^{\Id_1}(\tMcU_1^\ds)$ under
  the polynomial functor associated with
  $M\pi^\ds_0 \colon M\tMcU^\ds_0 \to M\McU^\ds_0$.
}
\begin{proposition}\label{prop:rel-match-trv-psfw}
  \relmatchtrvpsfw
  \begin{equation*}\usebox{\relmatchtrvpsfwbox}\end{equation*}
\end{proposition}

To show \Cref{prop:rel-match-trv-psfw}, we need to know what $\WhM(\trv^{\Gl(M)})$
from \Cref{eqn:gl-trv} as an object of $\bC_1$ represents.
\begin{lemma}\label{lem:rel-mat-trv-rep}
  Suppose $M \colon \bC_0 \to \bC_1$ has a left adjoint
  $D \colon \bC_1 \to \bC_0$.

  The relative matching object $\WhM(\trv_{\Gl(M)}) \in \bC_1$ from
  \Cref{eqn:gl-trv} represents the presheaf whose action on objects takes each
  $B \in \bC_1$ to the following tuples of data:
  \begin{enumerate}
    \item A $\pi_0$-fibration $E_0 \twoheadrightarrow DB \in \bC_0$
    \item A choice of two $\pi_1$-fibrant objects
    $E_1,E_1' \twoheadrightarrow \eta_B^*(ME_0)$
    \item The data for an $\Id_1$-homotopy isomorphism between $E_1,E_1'$ over $\eta_B^*(ME_0)$
  \end{enumerate}
  where $\eta_B \colon B \to MD(B)$ is the $(D \dashv M)$-adjunction unit at $B$.
  \begin{equation*}
    \begin{tikzcd}[cramped, sep=small]
      {E_1} && {E_1'} \\
      & {\WhM_\eta E} & {ME_0} \\
      & B & {MD(B)}
      \arrow["\simeq", from=1-1, to=1-3]
      \arrow[two heads, from=1-1, to=2-2]
      \arrow[two heads, from=1-3, to=2-2]
      \arrow[from=2-2, to=2-3]
      \arrow[from=2-2, to=3-2]
      \arrow["\lrcorner"{anchor=center, pos=0.15, scale=1.5}, draw=none, from=2-2, to=3-3]
      \arrow[from=2-3, to=3-3]
      \arrow["\eta"', from=3-2, to=3-3]
    \end{tikzcd}
  \end{equation*}
\end{lemma}
\begin{proof}
  We first calculate what $\HIso_\McV^{\Id^\ds_{\Gl(M)}}(\tMcV)_1$ represents.
  Because $D \dashv M$, the projection $\Gl(M) \to \bC_1$ also has a left
  adjoint $\bC_1 \to \Gl(M)$ taking each $B \in \bC_1$ to the unit arrow
  $\eta_B \colon B \to MD(B)$.
  In other words, one has
  \begin{equation*}
    \bC_1(B, \HIso_\McV^{\Id^\ds_{\Gl(M)}}(\tMcV)_1) \cong \Gl_M(B \xrightarrow{\eta_B} MD(B), \HIso_{\McV^\ds}^{\Id^\ds_{\Gl(M)}}(\tMcV^\ds))
  \end{equation*}
  Because we are trying to calculate what
  $\HIso_\McV^{\Id^\ds_{\Gl(M)}}(\tMcV)_1$ represents, we may obtain the result
  using the representable nature of
  $\HIso_{\McV^\ds}^{\Id^\ds_{\Gl(M)}}(\tMcV^\ds)$.
  Namely, maps $\eta_B \to \HIso_{\McV^\ds}^{\Id^\ds_{\Gl(M)}}(\tMcV^\ds)$ are
  in natural bijective correspondence with $\Id^\ds_{\Gl(M)}$-homotopy equivalences
  over $\eta \colon B \to MD(B)$.
  Explicitly, these consist of:
  \begin{enumerate}
    \item Internal Reedy fibrations $E, E' \twoheadrightarrow \eta_B \in \Gl(M)$
    \item Maps $f \colon E \to E' \in \sfrac{\Gl(M)}{\eta_B}$ and
    $s,r \colon E' \rightrightarrows E \in \sfrac{\Gl(M)}{\eta_B}$
    \item Homotopies
    $H^s \colon E \to \Id^{\Gl(M)}_B(E) \in \sfrac{\Gl(M)}{\eta_B}$ and
    $H^r \colon E' \to \Id^{\Gl(M)}_B(E') \in \sfrac{\Gl(M)}{\eta_B}$ such that
    \begin{align*}
      \ev_\partial H^s = (fs, \id) && \ev_\partial H^r = (rf, \id)
    \end{align*}
  \end{enumerate}

  The 0-components above respectively give rise to maps $f_0 \colon E_0 \to E_0'$
  and $s_0,r_0 \colon E_0' \rightrightarrows E_0$ fibrant over $DB$ in $\bC_0$
  along with homotopies $H^s_0 \colon E_0 \to \Id_{DB}(E_0)$ and
  $H^r_0 \colon E_0' \to \Id_{DB}(E_0')$ such that
  \begin{align*}
    \ev_\partial H^s_0 = (f_0s_0, \id) && \ev_\partial H^r_0 = (r_0f_0, \id)
  \end{align*}

  As $\WhM(\trv_{\Gl(M)})$ is obtained as the pullback of
  $\HIso_{\McV^\ds}^{\Id^\ds_{\Gl(M)}}(\tMcV^\ds)$ along
  $M(\trv_0) \colon M(\McU_0^\ds) \to
  M(\HIso_{\McU_0}^{\Id_0^\ds}(\tMcU_0^\ds))$, to know what
  $\WhM(\trv^{\Gl(M)})$ represents is to know what the 1-component looks like
  when the 0-components make up the trivial equivalence.

  For this, we first analyse the 1-component in general.
  The commutativity condition says that $f_1 \colon E_1 \to E_1'$ is a map over
  $\WhM_\eta(f) \colon \WhM_\eta(E) \to \WhM_\eta(E')$, which are both over $B$,
  and likewise for $s,r$.
  In other words, combining the 0-component and 1-component of $s,f,r$ gives the
  following diagram
  \begin{equation*}\label{eqn:sfr}\tag{$s,f,r$}
    \begin{tikzcd}[cramped]
      {E_1} && {E_1'} \\
      {\WhM_\eta E} && {\WhM_\eta E'} \\
      B && {ME_0} && {ME_0'} \\
      && {MD(B)}
      \arrow[two heads, from=1-1, to=2-1]
      \arrow["{r_1}"{description}, curve={height=-12pt}, dashed, from=1-3, to=1-1]
      \arrow["{s_1}"{description}, curve={height=12pt}, dashed, from=1-3, to=1-1]
      \arrow["{f_1}"{description}, dashed, from=1-1, to=1-3]
      \arrow[two heads, from=1-3, to=2-3]
      \arrow[from=2-1, to=3-1]
      \arrow["{\WhM_\eta r}"{description}, curve={height=-12pt}, dotted, from=2-3, to=2-1]
      \arrow["{\WhM_\eta s}"{description}, curve={height=12pt}, dotted, from=2-3, to=2-1]
      \arrow["{\WhM_\eta f}"{description}, dotted, from=2-1, to=2-3]
      \arrow[from=2-3, to=3-1]
      \arrow[from=2-3, to=3-5]
      \arrow["\eta"{description}, from=3-1, to=4-3]
      \arrow[from=3-3, to=4-3]
      \arrow["{Mr_0}"{description}, curve={height=-12pt}, dashed, from=3-5, to=3-3]
      \arrow["{Ms_0}"{description}, curve={height=12pt}, dashed, from=3-5, to=3-3]
      \arrow["{Mf_0}"{description}, dashed, from=3-3, to=3-5]
      \arrow[from=3-5, to=4-3]
      \arrow[from=2-1, to=3-3, crossing over]
    \end{tikzcd}
  \end{equation*}

  Analysing the commutativity condition for $H^s$ then says that
  $H_1^s \colon E_1 \to (\Id_\eta E)_1$ is a map over
  $\WhM_\eta(H^s) \colon\WhM_\eta E \to \WhM_\eta(\Id_\eta E)$, where the map
  $(\Id_\eta E)_1 \xrightarrow{\quad} \WhM_\eta(\Id_\eta E)$ factors through the
  relative matching map
  $\Whm_{E \times_\eta E}(\Id_\eta E) \colon (\Id_\eta E)_1 \twoheadrightarrow
  \WhM_{E \times_\eta E}(\Id_\eta E)$ of the boundary evaluation.
  In other words, for, once $s,f$ are fixed, $H^s$ is determined by the dashed
  maps $H_0^s,H_1^s$ making the following diagram commute
  \begin{equation*}\label{eqn:Hs}\tag{$H^s$}
    \begin{tikzcd}[cramped, column sep=small]
      && {(\Id_\eta E)_1} \\
      && {\WhM_{E \times_\eta E}(\Id_\eta E)} \\
      {E_1} && {E_1 \times_B E_1} && {\WhM_\eta(\Id_\eta E)} \\
      && {\WhM_\eta E} && {\WhM_\eta E \times_B \WhM_\eta E} && {M(\Id_{DB}(E_0))} \\
      &&&& {ME_0} && {M(E_0 \times_{DB} E_0)} \\
      &&&& B \\
      &&&&&& {MD(B)}
      \arrow["{\Whm_{E \times_\eta E}(\Id_\eta E)}", two heads, from=1-3, to=2-3]
      \arrow[from=2-3, to=3-3]
      \arrow[from=2-3, to=3-5]
      \arrow["{H_1^s}", dashed, from=3-1, to=1-3]
      \arrow["{(f_1s_1,\id)}"{description}, dashed, from=3-1, to=3-3]
      \arrow[two heads, from=3-1, to=4-3]
      \arrow[from=3-3, to=4-5]
      \arrow["{\WhM_\eta(\ev_\partial)}"{description}, from=3-5, to=4-5]
      \arrow[from=3-5, to=4-7]
      \arrow["{\WhM_\eta (H^s)}"{description, pos=0.7}, dotted, from=4-3, to=3-5]
      \arrow["{\WhM_\eta (f_0s_0,\id)}"'{}, dashed, from=4-3, to=4-5]
      \arrow[from=4-3, to=5-5]
      \arrow[from=4-5, to=5-7]
      \arrow[from=4-5, to=6-5]
      \arrow["{M(\ev_\partial)}", from=4-7, to=5-7]
      \arrow["{MH^s_0}"{description, pos=0.7}, dashed, from=5-5, to=4-7]
      \arrow["{M(f_0s_0,\id)}"', dashed, from=5-5, to=5-7]
      \arrow[from=5-7, to=7-7]
      \arrow["\eta"{description}, from=6-5, to=7-7]
    \end{tikzcd}
  \end{equation*}
  We can also do the same calculation on $H^r$ to get a corresponding diagram for
  it.

  But recall that we are interested in knowing what $\WhM(\trv_{\Gl(M)})$ looks
  like, which is to know what the 1-component of the homotopy equivalence above
  $\eta_B \colon B \to MD(B)$ look like when the 0-component is the trivial
  equivalence.
  So we do this now by replacing all of the 0-components in \Cref{eqn:sfr,eqn:Hs}
  with the corresponding components for the trivial equivalence.
  Doing this replacement for \Cref{eqn:sfr}, we set $s_0,f_0,r_0$ all to be the
  identity (so in particular $ME_0 = ME_0'$ and $\WhM_\eta E = \WhM_\eta E'$).
  Consequently, $\WhM_\eta s, \WhM_\eta f, \WhM_\eta r$ are all also identities.
  Then, $s_1,f_1,r_1$ is in the slice over $\WhM_\eta E$.
  So, \Cref{eqn:sfr} is now
  \begin{equation*}\label{eqn:sfr-refl0}\tag{$(s,f,r)\text{-}\refl_0$}
    \begin{tikzcd}[cramped, sep=small]
      {E_1} && {E_1'} \\
      & {\WhM_\eta E} && {ME_0} \\
      & B && {MD(B)}
      \arrow["{f_1}"{description}, dashed, from=1-1, to=1-3]
      \arrow[two heads, from=1-1, to=2-2]
      \arrow["{r_1}"{description}, curve={height=-12pt}, dashed, from=1-3, to=1-1]
      \arrow["{s_1}"{description}, curve={height=12pt}, dashed, from=1-3, to=1-1]
      \arrow[from=2-2, to=3-4, draw=none, "\lrcorner"{anchor=center, scale=1.5, pos=0.15}]
      \arrow[two heads, from=1-3, to=2-2]
      \arrow[from=2-2, to=2-4]
      \arrow[from=2-2, to=3-2]
      \arrow[from=2-4, to=3-4]
      \arrow["\eta"', from=3-2, to=3-4]
    \end{tikzcd}
  \end{equation*}

  We now do the same replacement of the 0-components of \Cref{eqn:Hs}.
  In addition to replacing the 0-component $s_0,f_0,r_0$ and the pullbacks
  $\WhM_\eta s, \WhM_\eta f, \WhM_\eta r$ all with identities, we also set
  $H_0^s$ to be $\refl$.
  Then, $\WhM_\eta(f_0s_0,\id)$ is the pullback of the diagonal, which is still
  the diagonal, and $\WhM_\eta(H^s)$ is $\WhM_\eta(\refl)$.
  Focusing in on the component over $B$, we get the diagram on the left.
  \begin{center}\scriptsize
    \begin{minipage}{0.45\linewidth}
      \begin{equation*}
        \begin{tikzcd}[cramped, column sep=tiny, row sep=small]
          &&& {(\Id_\eta E)_1} \\
          &&& {\WhM_{E \times_\eta E}(\Id_\eta E)} \\
          \\
          {E_1} &&& {E_1 \times_B E_1} && {\WhM_\eta(\Id_\eta E)} \\
          \\
          && {\WhM_\eta E} &&& {\WhM_\eta E \times_B \WhM_\eta E}
          \arrow["{\Whm_{E \times_\eta E}(\Id_\eta E)}", two heads, from=1-4, to=2-4]
          \arrow[from=2-4, to=4-4]
          \arrow[two heads, from=2-4, to=4-6]
          \arrow["\lrcorner"{anchor=center, pos=0.15, scale=1.5}, draw=none, from=2-4, to=6-6]
          \arrow["{H_1^s}", dashed, from=4-1, to=1-4]
          \arrow["{(f_1s_1,\id)}"{description}, dashed, from=4-1, to=4-4]
          \arrow[two heads, from=4-1, to=6-3]
          \arrow[two heads, from=4-4, to=6-6]
          \arrow["{\WhM_\eta(\ev_\partial)}"{description}, from=4-6, to=6-6]
          \arrow["{\WhM_\eta(\refl)}"{description}, from=6-3, to=4-6]
          \arrow["\Delta"', from=6-3, to=6-6]
        \end{tikzcd}
      \end{equation*}
    \end{minipage}
    \begin{minipage}{0.45\linewidth}
      \begin{equation*}
        \begin{tikzcd}[cramped, column sep=tiny, row sep=small]
          &&& {(\Id_\eta E)_1} \\
          &&& {\WhM_{E \times_\eta E}(\Id_\eta E)} \\
          && {\WhM_\eta(\refl)^*(\Id_\eta E_1)} \\
          {E_1} &&& {E_1 \times_B E_1} && {\WhM_\eta(\Id_\eta E)} \\
          && {E_1 \times_{\WhM_\eta E} E_1} \\
          && {\WhM_\eta E} &&& {\WhM_\eta E \times_B \WhM_\eta E}
          \arrow["{\Whm_{E \times_\eta E}(\Id_\eta E)}", two heads, from=1-4, to=2-4]
          \arrow["{(f_1s_1,\id)}"{description}, dashed, from=4-1, to=4-4]
          \arrow[from=2-4, to=4-4]
          \arrow[two heads, from=2-4, to=4-6]
          \arrow[from=3-3, to=1-4]
          \arrow[two heads, from=3-3, to=5-3]
          \arrow["{H_1^s}", dashed, from=4-1, to=1-4]
          \arrow[dotted, from=4-1, to=3-3]
          \arrow[dotted, from=4-1, to=5-3]
          \arrow[two heads, from=4-1, to=6-3]
          \arrow[two heads, from=4-4, to=6-6]
          \arrow["{\WhM_\eta(\ev_\partial)}"{description}, from=4-6, to=6-6]
          \arrow[from=5-3, to=2-4]
          \arrow[from=5-3, to=4-4]
          \arrow[from=5-3, to=6-3]
          \arrow["\lrcorner"{anchor=center, pos=0.15, scale=1.5}, draw=none, from=2-4, to=6-6]
          \arrow["\lrcorner"{anchor=center, pos=0.15, scale=1.5, rotate=45}, draw=none, from=3-3, to=2-4]
          \arrow["\lrcorner"{anchor=center, pos=0.05, scale=1.5, rotate=45}, draw=none, from=5-3, to=6-6]
          \arrow["{\WhM_\eta(\refl)}"{description}, from=6-3, to=4-6]
          \arrow["\Delta"', from=6-3, to=6-6]
        \end{tikzcd}
      \end{equation*}
    \end{minipage}
  \end{center}
  The diagram on the right is constructed from the internals of the diagrams on
  the left.
  We can first form the pullback in the bottom face to get
  $E_1 \times_{\WhM_\eta E} E_1$, which gives a connecting map
  $E_1 \times_{\WhM_\eta E} E_1 \xrightarrow{\quad} \WhM_{E \times_\eta
    E}(\Id_\eta E)$, along which one can pullback
  $(\Id_\eta E)_1 \xrightarrow{\quad} \WhM_{E \times_\eta E}(\Id_\eta E)$ to get
  $\WhM_\eta(\refl)^*(\Id_\eta E_1) \xrightarrow{\quad} E_1 \times_{\WhM_\eta E}
  E_1$.
  In other words, we have just pulled back
  \begin{equation*}
   (\Id_\eta E)_1 \twoheadrightarrow \WhM_{E \times_\eta E}(\Id_\eta E)
   \twoheadrightarrow \WhM_\eta(\Id_\eta E)
  \end{equation*}
  along $\WhM_\eta(\refl) \colon \WhM_\eta E \to \WhM_\eta(\Id_\eta E)$.
  Taking the pullback in the bottom face also gives the connecting map
  $E_1 \xrightarrow{\quad} E_1 \times_{\WhM_{E \times_\eta E}} E_1$ where the
  first component is again
  $f_1s_1 \colon E_1 \to E_1' \to E_1$.
  After this decomposition, we can then see that a choice of a homotopy
  $H_1^s \colon E_1 \xrightarrow{\quad} (\Id_\eta E)_1$ such that
  $\ev_\partial H_1^s = (f_1s_1,\id) \colon E_1 \to E_1 \times_B E_1$ as in the
  back triangle corresponds bijectively to a dotted map
  $E_1 \xrightarrow{\quad} \WhM_\eta(\refl)^*(\Id_\eta E_1)$ factoring the
  induced dotted map
  $(f_1s_1,\id) \colon E_1 \xrightarrow{\quad} E_1 \times_{\WhM_{E \times_\eta
      E}} E_1$.

  But by \Cref{prop:reedy-fib-id}\Cref{itm:reedy-fib-id-res}, one has that
  $\WhM_\eta(\refl)^*(\Id_\eta E_1) \twoheadrightarrow E_1 \times_{\WhM_\eta E}
  E_1$ is in fact just the $\Id$-type boundary evaluation
  $\Id_{\WhM_\eta E}(E_1) \twoheadrightarrow E_1 \times_{\WhM_\eta E} E_1$.
  Thus, a factorisation of $(f_1s_1,\id) \colon E_1 \to E_1 \times_B E_1$
  through $(\Id_\eta E)_1 \to E_1 \times_B E_1$ over $B$ is precisely a
  factorisation of $(f_1s_1,\id) \colon E_1 \to E_1 \times_{\WhM_\eta E} E_1$
  through
  $\Id_{\WhM_\eta E}(E_1) \twoheadrightarrow E_1 \times_{\WhM_\eta E} E_1$.
  In other words, $s_1$ is a homotopy section of $f_1$ over $\WhM_\eta E$.
  Repeating the same argument for $E'$, we see that $r_1$ is a homotopy
  retraction of $f_1$ over $\WhM_\eta E$.

  To summarise, we have calculated that $\WhM(\trv_{\Gl(M)})$ represents the
  presheaf taking each $B \in \bC_1$ to a $\pi_0$-fibration
  $E_0 \twoheadrightarrow DB \in \bC_0$ along with an $\Id_1$-homotopy
  isomorphism over $\eta_B^*(ME_0) = \WhM_\eta E$, as depicted in
  \Cref{eqn:sfr-refl0}, which is exactly as claimed in the statement of this
  \lcnamecref{lem:rel-mat-trv-rep}.
\end{proof}

\hspace{\parindent} With \Cref{prop:rel-match-trv-psfw} we now know that
$\WhM(\trv_{\Gl(M)})$ of \Cref{eqn:gl-trv} represents.
To prove \Cref{prop:rel-match-trv-psfw}, the next step is to know what
$\bP_{M\pi^\ds_0}(\HIso_{\McU_1^\ds}^{\Id_1}(\tMcU_1^\ds))$ of \Cref{eqn:gl-trv}
from \Cref{prop:rel-match-trv-psfw} represents.
To do so, we need to know what the polynomial functor associated with
$M\pi^\ds_0 \colon M\tMcU^\ds_0 \xrightarrow{\quad} M\McU^\ds_0$ does to objects
of $\bC_1$.
The adjoint pair assumption $D \dashv M$ also makes for a simple description.

\begin{lemma}\label{lem:gl-poly-rep}
  Suppose $M \colon \bC_0 \to \bC_1$ has a left adjoint
  $D \colon \bC_1 \to \bC_0$.

  The image of an object $Z \in \bC_1$ under the polynomial functor
  $\bP_{M\pi^\ds_0} \colon \bC_1 \to \bC_1$ associated with
  $M\pi^\ds_0 \colon M\tMcU^\ds_0 \to M\McU^\ds_0 \in \bC_1$ represents the presheaf
  \begin{equation*}
    \left(
      \coprod_{E_0 \colon DB \to \McU^\ds_0} \bC_1(B.M(DB.E_0), Z)
    \right)_{B \in \bC_1}
  \end{equation*}
  where $B.M(DB.E_0)$ is obtained by the following consecutive pullback
  \begin{equation*}
    \begin{tikzcd}[cramped, sep=small]
      {B.M(DB.E_0)} & {M(DB.E_0)} & {M\tMcU^\ds_0} \\
      B & {MD(B)} & {M\McU^\ds_0}
      \arrow[from=1-1, to=1-2]
      \arrow[from=1-1, to=2-1]
      \arrow[from=1-2, to=1-3]
      \arrow[from=1-2, to=2-2]
      \arrow["\lrcorner"{anchor=center, pos=0.15, scale=1.5}, draw=none, from=1-1, to=2-2]
      \arrow["\lrcorner"{anchor=center, pos=0.15, scale=1.5}, draw=none, from=1-2, to=2-3]
      \arrow[from=1-3, to=2-3]
      \arrow["\eta"', from=2-1, to=2-2]
      \arrow["{ME_0}"', from=2-2, to=2-3]
    \end{tikzcd}
  \end{equation*}
  via choices of pullbacks selected by the universe structure
  $\tMcU_1 \to \McU$.
\end{lemma}
\begin{proof}
  By universal property of the pushforwards, $\bP_{M\pi^\ds_0}(Z)$ represents the
  presheaf taking each $B \in \bC_1$ to
  \begin{equation*}
    \coprod_{B \to M\tMcU^\ds_0} \bC_1(B \times_{M\McU^\ds_0} M\tMcU^\ds_0, Z)
  \end{equation*}
  By adjointness, each $B \to M\tMcU^\ds_0$ factors uniquely through the unit
  $\eta_B \colon B \to MD(B)$ via its adjoint transpose
  $E_0 \colon DB \to \McU^\ds_0$ as follows
  \begin{equation*}
    B \xrightarrow{\eta} MD(B) \xrightarrow{ME_0} M\McU^\ds_0
  \end{equation*}
  The result now follows because $M$ is also continuous as a right adjoint.
\end{proof}

We can now complete the proof of \Cref{prop:rel-match-trv-psfw}.
\begin{proposition*}{\ref{prop:rel-match-trv-psfw}}
  \relmatchtrvpsfw
\end{proposition*}
\begin{proof}
  The result follows from \Cref{lem:gl-poly-rep} used on
  $\HIso_{\McU_1^\ds}^{\Id_1}(\tMcU_1^\ds)$, which gives the same presheaf of
  \Cref{lem:rel-mat-trv-rep} says $\WhM(\trv_{\Gl(M)})$ represents.
\end{proof}


\subsubsection{Proof of Pointed Univalence}
One now recalls the statement of pointed univalence for Reedy fibrations because
one can prove it.
\begin{theorem*}{\ref{thm:glue-universe-ua}}
  \glueua
  \def\qedsymbol{\relax}
\end{theorem*}
\begin{proof}
  Because of the intensional type theory structure on $\tau$ and $\tau^\ds$
  provided by
  \Cref{prop:glue-universe-unit,prop:glue-universe-sigma,prop:glue-universe-pi,thm:glue-universe-id},
  it follows by \Cref{cor:src-dest-int-uni-fib} that it suffices to show that
  the map
  \begin{equation*}
    \trv_{\Gl(M)} \colon \McV^\ds \to \HIso_{\McV^\ds}^{\Id^\ds_{\Gl(M)}}(\tMcV^\ds)
  \end{equation*}
  obtained by applying \Cref{constr:hiso,constr:trv-htpy} on the internal
  universal Reedy fibration $\tau^\ds$ with the internal $\Id$-type
  $\Id_{\Gl(M)}^\ds$ of \Cref{thm:glue-universe-id} lifts uniformly on the left
  against $\tau \colon \tMcV \to \McV$.

  By \Cref{lem:gl-tc-ptwise-tc,lem:ptwise-tc-gl-tc}, the problem is now
  equivalent to showing that
  \begin{equation*}
    (\trv_{\Gl(M)})_i \colon \McV_i^\ds \to (\HIso_{\McV^\ds}^{\Id^\ds_{\Gl(M)}}(\tMcV^\ds))_i
  \end{equation*}
  lifts uniformly on the left against $\tau \colon \tMcV \to \McV \in \Gl(M)$
  while getting to assume that
  \begin{equation*}
    \trv_i \colon \McU^\ds_i \to \HIso_{\pi^\ds_i}(\Id^\ds_i) \in \bC_i
  \end{equation*}
  lifts uniformly on the left against $\pi_i \colon \tMcU_i \to \McU_i \in \bC_i$.
  Lifting for the 0-component is immediate because $\Gl(M) \to \bC_0$ preserves
  all logical structures, so the 0-component of $\trv_{\Gl(M)}$ is just
  $\trv_0 \colon \McU_0 \to \HIso_{\McU_0}^{\Id_0}(\tMcU_0)$.

  For the 1-component, we refer to the factorisation of $(\trv_{\Gl(M)})_1$ via
  the relative matching map from \Cref{eqn:gl-trv}.
  \begin{equation*}
    \McV_1^\ds \xrightarrow{\Whm(\trv_{\Gl(M)})} \WhM(\trv_{\Gl(M)}) \xrightarrow{\quad} \HIso_{\McV^\ds}^{\Id^\ds_{\Gl(M)}}(\tMcV^\ds)_1
  \end{equation*}
  In particular, the second map is a pullback of $M(\trv_0)$ along
  $\HIso_{\McV^\ds}^{\Id^\ds_{\Gl(M)}}(\tMcV^\ds) \in \Gl(M)$ viewed as a map in
  $\bC_1$.
  By the $\Sigma,\Pi,\Id$-structures combined with
  \Cref{cor:src-dest-int-uni-fib},
  $\HIso_{\McV^\ds}^{\Id^\ds_{\Gl(M)}}(\tMcV^\ds) \in \Gl(M)$ is a Reedy fibrant
  object of $\Gl(M)$ and so it is a $\pi_1$-fibration of $\bC_1$.
  In other words, the second map
  $\WhM(\trv_{\Gl(M)}) \xrightarrow{\quad}
  \HIso_{\McV^\ds}^{\Id^\ds_{\Gl(M)}}(\tMcV^\ds)_1$ is the pullback of
  $M(\trv_0)$ along a $\pi_1$-fibration.
  But then $\trv_0$ lifts uniformly against $\pi_0$ by pointed univalence and so
  $M(\trv_0)$ lifts uniformly on the left of $\pi_1$ by
  \Cref{asm:mrq}\Cref{itm:mrq-tc}.
  Therefore,
  $\WhM(\trv_{\Gl(M)}) \xrightarrow{\quad}
  \HIso_{\McV^\ds}^{\Id^\ds_{\Gl(M)}}(\tMcV^\ds)_1$ is also equipped with a
  uniform lifting structure against $\pi_1$ by the $\Pi$-type structure.

  By \Cref{prop:rel-match-trv-psfw} the first map $\Whm(\trv_{\Gl(M)})$, which
  is the relative matching map of the trivial homotopy equivalence, is the image
  of $\trv_1$ under the polynomial associated with the $\pi_1$-fibration
  $\pi_1^\ds$.
  Pointed univalence on $\bC_1$ gives that $\trv_1$ is in the left class to
  $\pi_1$, so we see that $\Whm(\trv_{\Gl(M)})$ is also in the left class by
  \Cref{lem:ptd-funext-sdf} applying on the pointed functionality structure of
  $\bC_1$.
\end{proof}



\section{Inverse Diagrams of Universe Categories}\label{sec:inverse-universe}
We now show that the structure of intensional type theory, pointed functional
extensionality and pointed univalence is stable under the formation of inverse
diagrams.

We first recall the definition of an inverse category.
\begin{definition}
  An \emph{inverse category} $\McI$ is a category such that there exists a
  grading on its objects by a degree function $\deg \colon \ob\McI \to \bN$ such
  that if $f \colon i \to j \in \McI$ is not an identity map, then
  $\deg i > \deg j$.

  For each $n \in \bN$, denote by:
  \begin{enumerate}
    %
    %
    \item $\McI_{\leq n}$ the full subcategory of $\McI$
    spanned by the objects of degree strictly at most $n$
    \item $\McI_{= n}$ the discrete subcategory of $\McI$ spanned by objects
    of degree exactly $n$
    \item $\partial(i/\McI)$ the full subcategory of $i/\McI$ excluding the
    identity map.
  \end{enumerate}
\end{definition}

Central to the theory of inverse categories is the concept of coskeletons and
its derived matching constructions, which we now recall.
\begin{definition}\label{def:deg-match}
  Let $\McI$ be a finite inverse category and $\bC$ be a finitely complete
  category.
  For each $n \in \bN$, the \emph{$(n+1)$-th coskeleton} functor is defined as the right
  adjoint to the restriction along $\McI_{\leq n} \hookrightarrow \McI_{\leq n+1}$
  \begin{equation*}
    \begin{tikzcd}[cramped]
      {\bC^{\McI_{\leq n}}} & {\bC^{\McI_{\leq n+1}}} & \bC^{\McI_{=n+1}}
      \arrow[""{name=0, anchor=center, inner sep=0}, "{\cosk_{n+1}\relax}", shift left=1.5, from=1-1, to=1-2]
      \arrow[""{name=1, anchor=center, inner sep=0}, "{(-)|_{\leq n}}", shift left=1.5, from=1-2, to=1-1]
      \arrow["{(-)|_{=n+1}}", from=1-2, to=1-3]
      \arrow["\dashv"{anchor=center, rotate=90}, draw=none, from=1, to=0]
    \end{tikzcd}
  \end{equation*}

  The \emph{$(n+1)$-th matching object functor} is the composite
  \begin{equation*}
    M_{n+1}(-) \coloneqq (\cosk_{n+1}-)|_{={n+1}} \colon \bC^{\McI_{\leq n}} \to
    \bC^{\McI_{\leq n+1}} \to \bC^{\McI_{=n+1}}
  \end{equation*}
  and the \emph{$(n+1)$-th matching map} is the natural transformation
  \begin{equation*}
    m_{n+1}(-) \colon (-)|_{=n+1} \to M_{n+1}((-)|_{\leq n}) \in \Cat(\bC^{\McI_{\leq n+1}}, \bC^{\McI_{=n+1}})
  \end{equation*}
  obtained by restricting the unit of the coskeleton adjunction along
  $(-)|_{n+1}$.

  For each $i \in \McI_{=n+1}$, one also has the \emph{matching object functor}
  and \emph{matching map} at $i$ given by applying
  $\ev_i \colon \bC^{\McI_{=n+1}} \to \bC^{\set{i}} \cong \bC$ to $M_{n+1}$ and
  $m_{n+1}$.
  That is,
  \begin{equation*}
    M_i(-) \coloneqq M_{n+1}(-)_i \colon \bC^{\McI_{\leq n}} \to \bC^{\McI_{=n+1}} \to \bC
  \end{equation*}
  and
  \begin{equation*}
    m_i(-) \colon (-)_i \to M_i((-)|_{\leq n}) \in \Cat(\bC^{\McI_{\leq n+1}}, \bC)
  \end{equation*}
\end{definition}

For ease of notation, we often drop the restriction $(-)|_{\leq n}$ when its
application is clear from context.
Then, for each $X \in \bC^{\McI_{\leq n+1}}$, and $i \in \McI_{=n+1}$, the
matching object is the limit
\begin{equation*}
  M_iX = \lim(\partial(i/\McI) \to \McI \xrightarrow{X} \bC)
\end{equation*}
and the matching map $m_iX \colon X_i \to M_iX$ is the unique map induced by the
cone $(X_f \colon X_i \to X_j)_{f \colon i \to j \in \partial(i/\McI)}$.

The inductive nature in which inverse diagrams are constructed can also be
stated in terms of gluing categories as follows.

\begin{proposition}[{\cite[Proposition 4.7]{fkl24}}]\label{prop:inv-glue}
  If $\bC$ is a finitely complete category then for each $n \in \bN$, one has an
  equivalence of categories
  \begin{equation*}
    \bC^{\McI_{\leq n+1}} \simeq \Gl\left(\bC^{\McI_{\leq n}} \xrightarrow{M_{n+1}} \bC^{\McI_{=n+1}}\right)
  \end{equation*}
  taking each $B \in \bC^{\McI_{\leq n+1}}$ to the absolute matching map
  $m_{n+1} \colon B|_{=n+1} \to M_{n+1}(B|_{\leq n})$.
  \def\endingmark{\qedsymbol}
\end{proposition}

Our goal is to induce from a universe model $\bC$ of intensional type theory
another universe model on the category $\bC^\McI$ of inverse diagrams of shape
$\McI$ and show the induced diagram model has pointed functional extensionality
and pointed univalence.
Types in the diagram model $\bC^\McI$ are supposed to model telescopes of types
satisfying certain definitional inter-dependencies specified by the shape
$\McI$.
The notion of such inter-dependencies are formalised as Reedy fibrations defined
using the absolute matching objects and relative matching maps, which we now
recall.

\begin{definition}\label{def:inv-reedy-fib}
  Given a map $p \colon E \to B \in \bC^\McI$, its \emph{relative matching
    object} is the object over $B_i$ given by
  $\WhM_ip \to B_i \in \sfrac{\bC}{B_i}$ obtained as rebase of $M_ip$ along
  $m_iB$ and its \emph{relative matching map} is the connecting map
  $\Whm_ip \colon E_i \to \WhM_ip$ into the pullback, observed as follows.
  \begin{equation*}
    \begin{tikzcd}[cramped]
      {E_i} \\
      & {\WhM_ip} & {M_iE} \\
      & {B_i} & {M_iB}
      \arrow["{\Whm_ip}"{description}, dashed, from=1-1, to=2-2]
      \arrow["{m_iE}", curve={height=-12pt}, from=1-1, to=2-3]
      \arrow["{p_i}"', curve={height=12pt}, from=1-1, to=3-2]
      \arrow[from=2-2, to=2-3]
      \arrow[from=2-2, to=3-2]
      \arrow["\lrcorner"{anchor=center, pos=0.15, scale=1.5}, draw=none, from=2-2, to=3-3]
      \arrow["{M_ip}", from=2-3, to=3-3]
      \arrow["{m_iB}"', from=3-2, to=3-3]
    \end{tikzcd}
  \end{equation*}

  Suppose $\pi \colon \tMcU \to \McU \in \bC$ is a universal map of $\bC$.
  Then, map $p \colon E \to B \in \bC^\McI$ is a \emph{Reedy $\pi$-fibration}
  when all of its relative matching maps $\Whm_ip$ are $\pi$-fibrations.

  A \emph{(weakly) classifying Reedy $\pi$-fibration} is a map
  $\pi^\McI \colon \tMcU^\McI \to \McU^\McI\in \bC^\McI$ such that $E \to B$ is
  a $\pi^\McI$-fibration if and only if it is a Reedy $\pi$-fibration.
\end{definition}

For the rest of this section, we fix a Reedy category $\McI$ and a finitely
complete category $\bC$ with a universal map
\begin{align*}
  \pi \colon \tMcU \to \McU \in \bC
\end{align*}
equipped with an internal universe
\begin{align*}
  \pi^\ds \colon \tMcU^\ds \to \McU^\ds \in \bC
\end{align*}
We also assume that $\pi$ and $\pi^\ds$ are equipped with
$\Unit,\Sigma,\Pi,\Id$-type structures denoted $\Unit,\Sigma,\Pi,\Id$ and
$\Unit^\ds,\Sigma^\ds,\Pi^\ds,\Id^\ds$ respectively.

The first task is to show that these type-theoretic structures on $\bC$ all
transfer to the corresponding structures on inverse diagram categories
$\bC^\McI$ where $\McI$ is inverse.
In light of the presentation of inverse diagram categories as gluing categories
as from \Cref{prop:inv-glue}, we accomplish this by applying the series of
results in
\Cref{prop:glue-universe-unit,prop:glue-universe-sigma,prop:glue-universe-pi,thm:glue-universe-id}.
However, these results all depend on the ``right Quillen'' properties of the
matching object functor, as in \Cref{asm:mrq}.
Therefore, we start by a series of results building up to the verification that
\Cref{asm:mrq} is satisfied by each $M_{n+1}$.

\begin{lemma}[{\cite[Proposition 2.4]{inv-psfw}}]\label{lem:mat-fib}
  Fix $i \in \McI_{=n+1}$.
  The functor $M_i \colon \bC^{\McI_{\leq n}} \to \bC$ computing the absolute
  matching object at $i$ sends Reedy $\pi$-fibrations to $\pi$-fibrations and
  Reedy $\pi^\ds$-fibrations to $\pi^\ds$-fibrations.
  \def\endingmark{\qedsymbol}
\end{lemma}

\begin{lemma}\label{lem:mat-tc}
  Let $\tau \colon \tMcV \to \McV \in \bC^{\McI_{\leq n+1}}$ be a classifying
  Reedy $\pi$-fibration.

  Fix $i \in \McI_{=n+1}$.
  For every $B \in \bC^{\McI_{\leq n}}$ and map
  $s \colon E \to Y \in \sfrac{\bC^{\McI_{\leq n}}}{B}$ in the Reedy
  $\pi$-fibrant slice over $B$, one has a map
  \begin{equation*}
    (s \fracsquareslash{B} (B \times \tau)) \xrightarrow{\quad} (M_is \fracsquareslash{M_iB} (M_iB \times \pi))
  \end{equation*}
  subject to the condition that for each $f \colon B' \to B$, the following
  diagram commutes.
  \begin{equation*}
    \begin{tikzcd}[cramped, row sep=small, column sep=huge]
      (s \fracsquareslash{B} (B \times \tau))
      \ar[r, "{}"]
      \ar[d, "{\text{\cite[Construction 3.2]{struct-lift}}}"]
      &
      (M_is \fracsquareslash{M_iB} (M_iB \times \pi))
      \ar[d, "{\text{\cite[Construction 3.2]{struct-lift}}}"]
      \\
      (f^*s \fracsquareslash{f^*B} (f^*B \times \tau))
      \ar[r, "{}"]
      &
      (M_i(f^*s) \fracsquareslash{M_iB'} (M_iB' \times \pi))
    \end{tikzcd}
  \end{equation*}
\end{lemma}
\begin{proof}
  By \cite[Proposition 2.1]{inv-psfw}, the map $M_iE \to M_iB$ viewed as an
  object over $M_iB$ is the limit of the diagram
  \begin{equation*}
    (\proj_f^*E_j \to M_iB)_{f \colon i \to j \in \partial(i/\McI)}
  \end{equation*}
  over $M_iB$, where each $\proj_f \colon M_iB \to B_j$ is limiting leg of the
  matching object cone.
  In \cite[Proposition 2.4]{inv-psfw}, it was proven that Reedy fibrancy of $E$
  results in Reedy fibrancy of the diagram
  $(\proj_f^*E_j \to M_iB)_{f \colon i \to j \in \partial(i/\McI)}$ over $M_iB$.
  Moreover, \cite[Proposition 2.4]{inv-psfw} proves that each $E_j \to B_j$ is a
  $\pi$-fibration and so
  $(\proj_f^*E_j \to M_iB)_{f \colon i \to j \in \partial(i/\McI)}$ is a Reedy
  fibrant diagram valued in the fibrant slice over $M_iB$.
  Applying the same argument to $M_iY \to M_iB$, we see that it is also the
  limit of a Reedy fibrant diagram
  \begin{equation*}
    (\proj_f^*Y_j \to M_iB)_{f \colon i \to j \in \partial(i/\McI)}
  \end{equation*}
  over $M_iB$ valued in the fibrant slice over $M_iB$.

  Therefore, $M_is \colon M_iE \to M_iY \in \sfrac{\bC}{M_iB}$ is the map
  \begin{equation*}
    M_is = \lim_{f \colon i \to j \in \partial(i/\McI)}(\proj_f^*E_j \xrightarrow{\proj_f^*s} \proj_f^*Y_j)
  \end{equation*}
  obtained by applying the limit functor to a natural transformation between
  Reedy fibrant diagrams valued in fibrant objects over $M_iB$.
  But the $\Unit,\Sigma,\Pi,\Id$-type structure of $\bC$ gives rise to a
  type-theoretic fibration category structure on each fibrant slice of $\bC$.
  Therefore, by \cite[Lemma 11.8]{shu15}, it follows that $M_is$ lifts against
  $\pi$.
  The required compatibility conditions in the statement of this
  \lcnamecref{lem:mat-tc} is routinely verified by referring to the computations
  involved in \cite{shu15}.
\end{proof}

\begin{lemma}\label{lem:discrete-tt}
  For any set $S$ viewed as a discrete category, the constant map at $\pi$, also denoted
  \begin{align*}
    \pi \colon \tMcU \to \McU \in \bC^S
  \end{align*}
  is equipped with an internal universe given by the constant map at $\pi^\ds$, also denoted
  \begin{align*}
    \pi^\McI \colon \tMcU^\ds \to \McU^\ds \in \bC^S
  \end{align*}
  They are also both equipped with $\Unit,\Sigma,\Pi,\Id$-type structures
  defined using the corresponding structures from $\bC$ pointwise.
\end{lemma}
\begin{proof}
  Straightforwards.
\end{proof}

We can show that various type-theoretic structures of $\bC^\McI$ are preserved
by formation of inverse diagrams by way of iterated gluing.
For pointed univalence, a slight complication arises due to the technical
reliance of a left adjoint to the matching object functor in
\Cref{subsec:glue-universe-univalence} for \Cref{prop:rel-match-trv-psfw}.
Recall from \Cref{def:deg-match} that the matching object functor
$M_{n+1} \colon \bC^{\McI_{\leq n}} \to \bC^{\McI_{\leq n+1}} \to \bC^{\McI_{=n+1}}$
is defined as the composite of the coskeleton functor
$\cosk_{n+1} \colon \bC^{\McI_{\leq n}} \to \bC^{\McI_{\leq n+1}}$
with the restriction along $\McI_{=n+1} \hookrightarrow \McI_{\leq n+1}$.
While $\cosk_{n+1}$ is by definition a right adjoint, the restriction functor
$\bC^{\McI_{\leq n+1}} \to \bC^{\McI_{=n+1}}$ does not necessarily have a left
adjoint unless $\bC$ is cocomplete.
Fortunately, this problem can be solved by passing into the free cocompletion of
$\bC$, so that pointed univalence for inverse categories can still be obtained.

\begin{theorem}\label{thm:inverse-tt}
  For any finite inverse category $\McI$, the diagram category $\bC^\McI$ also
  admits a universal map
  \begin{align*}
    \pi^\McI \colon \tMcU^\McI \to \McU^\McI \in \bC^\McI
  \end{align*}
  equipped with an internal universe
  \begin{align*}
    (\pi^\McI)^\ds \colon (\tMcU^\McI)^\ds \to (\McU^\McI)^\ds \in \bC^\McI
  \end{align*}
  such that
  \begin{enumerate}[itemsep=0em]
    \item\label{itm:inverse-tt-univ} A map $E \to B \in \bC^\McI$ is a $\pi$-Reedy fibration (respectively
    $\pi^\ds$-Reedy fibration) if and only if it is a $\pi^\McI$-fibration
    (respectively $(\pi^\McI)^\ds$-fibration).
    \item\label{itm:inverse-tt-uspi} $\pi^\McI$ and $(\pi^\McI)^\ds$ are both equipped with
    $\Unit,\Sigma,\Pi,\Id$-type structures.
    \item\label{itm:inverse-tt-fe} If the $\Id$-type structure on $\pi$ has a pointed functional
    extensionality structure then so does the induced $\Id$-type structure on
    $\bC^\McI$.
    \item\label{itm:inverse-tt-ua} Suppose $\bC$ is cocomplete.
    If the $\Id$-type structure on $\pi$ has
    a pointed functional extensionality structure and $\bC$ has a pointed
    $\pi$-univalence structure on $(\pi^\ds,\Id^\ds)$ then $\bC^\McI$ has a
    $\pi^\McI$-univalence structure on $((\pi^\ds)^\McI,(\Id^\ds)^\McI)$, where
    $(\Id^\ds)^\McI$ is the induced $\Id$-type structure on $(\pi^\ds)^\McI$.
  \end{enumerate}
\end{theorem}
\begin{proof}
  Because $\McI$ is finite, it suffices to show the result for each
  $\McI_{\leq n}$ by induction on $n$.

  If $n=0$ then the result is directly by \Cref{lem:discrete-tt}.
  Now assume that we have established such universe category structures on
  $\bC^{\McI_{\leq n}}$ and the goal is to establish these structures on
  $\bC^{\McI_{\leq n+1}}$.
  Then, by \Cref{lem:mat-fib,lem:mat-tc} combined with the pointwise universe
  category structure on $\bC^{\McI_{=n+1}}$ from \Cref{lem:discrete-tt}, the
  $(n+1)$-th matching object functor
  $M_{n+1} \colon \bC^{\McI_{\leq n}} \to \bC^{\McI_{=n+1}}$ satisfies
  \Cref{asm:mrq}.
  Therefore, the series of results in
  \Cref{prop:glue-universe-unit,prop:glue-universe-sigma,prop:glue-universe-pi,thm:glue-universe-id}
  equips $\Gl(M_{n+1})$ with a universal Reedy fibration with an internal
  universal Reedy fibration along with $\Unit,\Sigma,\Pi,\Id$-type structures on
  both universal maps.
  These universal structures are transferred to $\bC^{\McI_{\leq n+1}}$ by
  noting that the equivalence $\bC^{\McI_{\leq n+1}} \simeq \Gl(M_{n+1})$ of
  \Cref{prop:inv-glue} preserves and reflects Reedy fibrations and pushforwards.
  This proves \Cref{itm:inverse-tt-univ,itm:inverse-tt-uspi}.

  \Cref{thm:gl-ptd-fe} also inductively establishes a pointed $\Id$-type
  functional extensionality structure on $\Gl(M_{n+1})$ by provided it is
  present in $\bC$ since it is also present in $\bC^{\McI_{=n+1}}$ in this case
  by \Cref{lem:discrete-tt}.
  The pointed functional extensionality structure also transfers to
  $\bC^{\McI_{\leq n+1}}$ under the equivalence
  $\bC^{\McI_{\leq n+1}} \simeq \Gl(M_{n+1})$ of \Cref{prop:inv-glue} as it is
  formulated as a structured lift and thus stable under equivalence of
  categories.
  This proves \Cref{itm:inverse-tt-fe}.

  Finally, for \Cref{itm:inverse-tt-ua}, we use \Cref{thm:glue-universe-ua} and
  the fact that cocompleteness of $\bC$ gives a left adjoint to the matching
  object functor $M_{n+1}$.
\end{proof}

For the final result, we remove the cocompleteness assumption of
\Cref{thm:inverse-tt}\Cref{itm:inverse-tt-ua} by passing into the presheaf
category.

\begin{lemma}\label{lem:yoneda-preserve-universe}
  Under the Yoneda embedding $\bC \hookrightarrow \widehat{\bC}$,
  \begin{enumerate}[itemsep=0em]
    \item The map $\pi \colon \tMcU \to \McU \in \widehat{\bC}$ is also equipped
    with an internal universe given by the map
    $\pi^\ds \colon \tMcU^\ds \to \McU^\ds \in \widehat{\bC}$.
    Both $\pi,\pi^\ds \in \widehat{\bC}$ also inherit the
    $\Unit,\Sigma,\Pi,\Id$-type structures they had in $\bC$ by the Yoneda
    embedding.
    \item If $\bC$ has a pointed $\pi$-univalence structure on $(\pi_0,\Id_0)$
    then so does $\widehat{\bC}$.
    \item If $\pi$ has a pointed $\Id$-functional extensionality structure in
    $\bC$ then so does $\pi$ in $\widehat{\bC}$.
  \end{enumerate}
\end{lemma}
\begin{proof}
  The Yoneda embedding is continuous and preserves pushforwards.
  Therefore, the $\Unit,\Sigma,\Pi$-type structures of $\bC$ on $\pi$ give rise
  to the corresponding structures on $\pi$ in $\widehat{\bC}$.

  For the $\Id$-type structure, we recall that a $\MsJ$-elimination structure is
  precisely a section to a pullback-Hom map due to \cite[Corollary
  1.7]{struct-lift}, and so $\MsJ$-elimination structures on pre-$\Id$-type
  structures in $\bC$ also give rise to their counterparts in $\widehat{\bC}$.
  The same argument also applies for pointed univalence and pointed functional
  extensionality as they are also defined using uniform lifting structures in
  \Cref{def:axm-univalence,def:funext}.
\end{proof}

\begin{theorem}\label{thm:inv-ua}
  For any finite inverse category $\McI$, the diagram category $\bC^\McI$ also
  admits a universal Reedy $\pi$-fibration
  \begin{equation*}
    \pi^\McI \colon \tMcU^\McI \to \McU^\McI \in \bC^\McI
  \end{equation*}
  with an internal universal Reedy $\pi^\ds$-fibration
  \begin{equation*}
    (\pi^\ds)^\McI \colon (\tMcU^\ds)^\McI \to (\McU^\ds)^\McI \in \bC^\McI
  \end{equation*}
  such that
  \begin{enumerate}[itemsep=0em]
    \item $\pi^\McI$ and $(\pi^\McI)^\ds$ are both equipped with
    $\Unit,\Sigma,\Pi,\Id$-type structures.
    \item If $\bC$ has a pointed $\pi$-univalence structure on
    $(\pi^\ds,\Id^\ds)$ and $\pi$ has a pointed $\Id$-functional extensionality
    structure then $\bC^\McI$ has a pointed $\pi^\McI$-univalence structure on
    $((\pi^\ds)^\McI,(\Id^\ds)^\McI)$, where $(\Id^\ds)^\McI$ is the induced
    $\Id$-type structure on $(\pi^\ds)^\McI$.
  \end{enumerate}
\end{theorem}
\begin{proof}
  By \Cref{thm:inverse-tt}\Cref{itm:inverse-tt-univ,itm:inverse-tt-uspi},
  $\bC^\McI$ admits a universal Reedy $\pi$-fibration $\pi^\McI$ with an
  internal universal Reedy $\pi^\ds$-fibration $(\pi^\ds)^\McI$ where both
  universal maps are equipped $\Unit,\Sigma,\Pi,\Id$-type structures
  $\Unit^\McI,\Sigma^\McI,\Pi^\McI,\Id^\McI$ and
  $(\Unit^\ds)^\McI,(\Sigma^\ds)^\McI,(\Pi^\ds)^\McI,(\Id^\ds)^\McI$
  respectively.

  Applying \Cref{thm:inverse-tt} on $\widehat{\bC}$ with the logical structures
  on $\widehat{\bC}$ from \Cref{lem:yoneda-preserve-universe} given by the
  Yoneda embedding, we obtain a universal Reedy $\pi$-fibration $\ul{\pi^\McI}$
  on $\widehat{\bC}^\McI$ with an internal universal Reedy $\pi^\ds$-fibration
  $\ul{(\pi^\ds)^\McI}$ where both universal maps are equipped
  $\Unit,\Sigma,\Pi,\Id$-type structures
  $\ul{\Unit^\McI},\ul{\Sigma^\McI},\ul{\Pi^\McI},\ul{\Id^\McI}$ and
  $\ul{(\Unit^\ds)^\McI},\ul{(\Sigma^\ds)^\McI},\ul{(\Pi^\ds)^\McI},\ul{(\Id^\ds)^\McI}$
  respectively.
  Moreover, in the case of $\widehat{\bC}$, one additionally has a
  $\ul{\pi^\McI}$-univalence structure on
  $(\ul{(\pi^\ds)^\McI},\ul{(\Id^\ds)^\McI})$ by
  \Cref{thm:inverse-tt}\Cref{itm:inverse-tt-ua}.

  However, by the formula for pushforwards in inverse diagram categories in
  \cite[Lemma 2.2]{inv-psfw} and the fact that the Yoneda embedding preserves
  pushforwards \cite[Lemma 6.6(3)]{uem23}, one observes that
  $\bC^\McI \hookrightarrow \widehat{\bC}^\McI$ also preserves pushforwards.
  Consequently, an examination of the constructions of \Cref{sec:glue-universe}
  reveals that the universal maps
  $\ul{\pi^\McI},\ul{(\pi^\ds)^\McI} \in \widehat{\bC}^\McI$ arise as the image
  of $\pi^\McI,(\pi^\ds)^\McI \in \bC^\McI$ under
  $\bC^\McI \hookrightarrow \widehat{\bC}^\McI$ and likewise for the
  $\Unit,\Sigma,\Pi,\Id$-type structures as they are constructed relying only on
  pullbacks and pushforwards.
  In particular, the map
  \begin{equation*}
    \trv^\McI \colon (\pi^\ds)^\McI \xrightarrow{\quad}
    \HIso^{(\Id^\ds)^\McI}_{(\McU^\ds)^\McI}((\tMcU^\ds)^\McI) \in \bC^\McI
  \end{equation*}
  is sent by $\bC^\McI \hookrightarrow \widehat{\bC}^\McI$ to
  \begin{equation*}
    \ul{\trv^\McI} \colon \ul{(\pi^\ds)^\McI} \xrightarrow{\quad}
    \HIso^{\ul{(\Id^\ds)^\McI}}_{\ul{(\McU^\ds)^\McI}}(\ul{(\tMcU^\ds)^\McI}) \in \widehat{\bC}^\McI
  \end{equation*}
  Therefore, $\pi^\McI$-univalence on $((\pi^\ds)^\McI,(\Id^\ds)^\McI)$ follows
  by full faithfulness of $\bC^\McI \hookrightarrow \widehat{\bC}^\McI$, as all
  lifting problems of $X \times \trv^\McI$ against $\pi^\McI$ for any
  $X \in \bC^\McI$ can be solved in $\widehat{\bC}^\McI$, with the solution in
  fact being already in $\bC^\McI$.
\end{proof}


\printbibliography

\end{document}